\documentclass[conference]{IEEEtran}
\IEEEoverridecommandlockouts
\usepackage{cite}
\usepackage{amsmath,amssymb,amsfonts,amsthm}
\usepackage[version=4]{mhchem}
\usepackage[ruled, vlined, linesnumbered]{algorithm2e}
\usepackage{graphicx}
\usepackage{textcomp}
\usepackage{booktabs}
\usepackage{subfigure}
\usepackage[table]{xcolor}
\usepackage[hyphens]{url}
\definecolor{winered}{rgb}{0.5, 0.0, 0.0}
\definecolor{navyblue}{rgb}{0.0, 0.0, 0.5}
\usepackage[colorlinks=true,linkcolor=winered,citecolor=winered,urlcolor=navyblue,anchorcolor=winered,filecolor=winered]{hyperref}
\usepackage{cleveref}
\usepackage{multirow}
\usepackage{soul}
\usepackage{makecell}
\usepackage{adjustbox}
\usepackage{quantikz}
\usepackage{lipsum}
\usepackage{pifont}
\usepackage{enumitem}
\usepackage{tabularray}
\usepackage[normalem]{ulem}
\usepackage{makecell}
\usepackage{balance}
\usepackage{fontawesome5}

\newcommand{\email}[1]{\href{mailto:#1}{\textcolor{black}{\nolinkurl{#1}}}}

\Crefname{figure}{Fig.}{Figs.}

\colorlet{MutedCyan}{cyan!50!black!20}
\definecolor{SkyBlue}{RGB}{100, 180, 220}
\definecolor{TextHighlightBlue}{RGB}{72,120,208}

\newcommand{\dquote}[1]{``#1''}
\newcommand{\code}{\texttt}

\newcommand{\note}[1]{{#1}}

\newcommand{\ZY}[1]{{\color{purple}[ZY: #1]}}

\NewDocumentCommand{\delbyzy}{m O{ZY}}{
    \sout{#1}
    \textsuperscript{\color{purple}[#2]}
}

\newcommand{\canopus}{\textsc{Canopus}}
\newcommand{\fullNameOfCanopus}{\textbf{Can}onical-\textbf{O}ptimized \textbf{P}lacement \textbf{U}tility \textbf{S}uite}

\newcommand{\sabre}{\textsc{Sabre}}

\newcommand{\toqm}{\textsc{TOQM}}
\newcommand{\bqskit}{\textsc{BQSKit}}
\newcommand{\qiskit}{\textsc{Qiskit}}
\newcommand{\tket}{\textsc{TKet}}

\newcommand{\Can}{\mathrm{Can}}
\newcommand{\Uthree}{\mathrm{U3}}
\newcommand{\XX}{\mathrm{XX}}
\newcommand{\YY}{\mathrm{YY}}
\newcommand{\ZZ}{\mathrm{ZZ}}
\newcommand{\CX}{\mathrm{CX}}
\newcommand{\CZ}{\mathrm{CZ}}

\newcommand{\iSWAP}{\mathrm{iSWAP}}
\newcommand{\SQiSW}{\mathrm{\sqrt{iSWAP}}}
\newcommand{\ECP}{\mathrm{ECP}}
\newcommand{\SWAP}{\mathrm{SWAP}}
\newcommand{\CPhase}{\mathrm{CPhase}}
\newcommand{\pSWAP}{\mathrm{pSWAP}}

\newcommand{\countCost}{C_{\mathrm{count}}}
\newcommand{\depthCost}{C_{\mathrm{depth}}}

\newcommand{\CXISA}{\code{CX}}
\newcommand{\HetISA}{\code{Het}}
\newcommand{\ZZPhaseISA}{\code{ZZPhase}}
\newcommand{\ZZPhaseWithMirrorISA}{\code{ZZPhase\_}}
\newcommand{\SQiSWISA}{\code{SQiSW}}
\newcommand{\iSWAPISA}{\code{iSWAP}}
\newcommand{\SQiSWWithMirrorISA}{\code{SQiSW\_}}

\newcommand{\StabISA}{\code{CX}-\code{iSWAP}}

\newtheorem{theorem}{Theorem}\Crefname{theorem}{Theorem}{Theorems}  % 定理编号形如 "1.1", "1.2"
\newtheorem{definition}{Definition}\Crefname{definition}{Definition}{Definitions}  % 定义编号形如 "1.1", "1.2"
\newtheorem{lemma}{Lemma}\Crefname{lemma}{Lemma}{Lemmas}  % 引理编号形如 "1.1", "1.2"
\Crefname{corollary}{Corollary}{Corollaries}  % 推论编号形如 "1.1", "1.2"

\SetAlFnt{\small}
\def\BibTeX{{\rm B\kern-.05em{\sc i\kern-.025em b}\kern-.08em
    T\kern-.1667em\lower.7ex\hbox{E}\kern-.125emX}}
\begin{document}

% Ensure letter paper
\pdfpagewidth=8.5in
\pdfpageheight=11in

%%%%%%%%%%%---SETME-----%%%%%%%%%%%%%
\newcommand{\iscasubmissionnumber}{320}
%%%%%%%%%%%%%%%%%%%%%%%%%%%%%%%%%%%%

\pagenumbering{arabic}

\iffalse

\twocolumn
\input{cover_letter.tex}
\clearpage

\pagenumbering{arabic}
\setcounter{page}{1}
\setcounter{section}{0}

\fi

%%%%%%%%%%%---SETME-----%%%%%%%%%%%%%
\title{Unifying Qubit Routing Across Diverse Quantum ISAs via Canonical Representation
\thanks{{\faEnvelope[regular]} Corresponding author: \email{chenjianxin@tsinghua.edu.cn}.}
}
% \author{\normalsize{ISCA 2026 Submission
    % \textbf{\#\iscasubmissionnumber} -- Confidential Draft -- Do NOT Distribute!!}}
%%%%%%%%%%%%%%%%%%%%%%%%%%%%%%%%

\author{
    \IEEEauthorblockN{
        Zhaohui Yang\IEEEauthorrefmark{1},
        Kai Zhang\IEEEauthorrefmark{2}\IEEEauthorrefmark{3},
        Xinyang Tian\IEEEauthorrefmark{4},
        Xiangyu Ren\IEEEauthorrefmark{5},
        Yingjian Liu\IEEEauthorrefmark{6},
        Yunfeng Li\IEEEauthorrefmark{7},\\
        Dawei Ding\IEEEauthorrefmark{8}\IEEEauthorrefmark{9},
        Jianxin Chen\textsuperscript{\faEnvelope[regular]}\IEEEauthorrefmark{2},
        Yuan Xie\IEEEauthorrefmark{1}
    }

    \IEEEauthorblockA{\textit{\IEEEauthorrefmark{1}Department of Electronic and Computer Engineering, \href{https://ror.org/00q4vv597}{The Hong Kong University of Science and Technology}, Hong Kong}}
    \IEEEauthorblockA{\textit{\IEEEauthorrefmark{2}Department of Computer Science and Technology, \href{https://ror.org/03cve4549}{Tsinghua University}, Beijing 100084, China}}
    \IEEEauthorblockA{\textit{\IEEEauthorrefmark{3}Department of Intelligent Computing, \href{https://www.pcl.ac.cn}{Pengcheng Laboratory}, Guangdong 518066, China}}
    \IEEEauthorblockA{\textit{\IEEEauthorrefmark{4}Institute for Interdisciplinary Information Sciences, \href{https://ror.org/03cve4549}{Tsinghua University}, Beijing 100084, China}}
    \IEEEauthorblockA{\textit{\IEEEauthorrefmark{5}Institute for Computer System Architecture, \href{https://ror.org/01nrxwf90}{The University of Edinburgh}, Edinburgh EH8 9AB, UK}}
    \IEEEauthorblockA{\textit{\IEEEauthorrefmark{6}Instituut-Lorentz for Theoretical Physics, \href{https://ror.org/027bh9e22}{Leiden University}, 2300 RA Leiden, The Netherlands}}
    \IEEEauthorblockA{\textit{\IEEEauthorrefmark{7}Department of Integrated Circuit, \href{https://ror.org/00d2w9g53}{Shenzhen Polytechnic University}, Guangdong 518055, China}}
    \IEEEauthorblockA{\textit{\IEEEauthorrefmark{8}Center for Mathematics and Interdisciplinary Sciences, \href{https://ror.org/013q1eq08}{Fudan University}, Shanghai 200433, China}}
    \IEEEauthorblockA{\textit{\IEEEauthorrefmark{9}\href{https://www.simis.cn}{Shanghai Institute for Mathematics and Interdisciplinary Sciences}, Shanghai 200433, China}}
}

% \author{
% \IEEEauthorblockN{Zhaohui Yang}
% \IEEEauthorblockA{
% \textit{HKUST} \\
% Hong Kong}
% \and
% \IEEEauthorblockN{Kai Zhang}
% \IEEEauthorblockA{
% \textit{Tsinghua University} \\
% Beijing, China}
% \and
% \IEEEauthorblockN{Xinyang Tian}
% \IEEEauthorblockA{
% \textit{Tsinghua University} \\
% Beijing, China}
% \and
% \IEEEauthorblockN{Xiangyu Ren}
% \IEEEauthorblockA{
% \textit{The University of Edinburgh} \\
% Edinburgh, UK}
% \and
% \IEEEauthorblockN{Yingjian Liu}
% \IEEEauthorblockA{
% \textit{Leiden University} \\
% Leiden, The Netherlands}
% \and
% \IEEEauthorblockN{Yunfeng Li}
% \IEEEauthorblockA{
% \textit{The University of Hong Kong} \\
% Hong Kong SAR, China}
% \and
% \IEEEauthorblockN{Dawei Ding}
% \IEEEauthorblockA{
% \textit{Fudan University} \\
% Beijing, China}
% \and
% \IEEEauthorblockN{Jianxin Chen}
% \IEEEauthorblockA{
% \textit{Tsinghua University} \\
% Beijing, China}
% \and
% \IEEEauthorblockN{Yuan Xie}
% \IEEEauthorblockA{
% \textit{HKUST} \\
% Hong Kong}
% }

\maketitle
\thispagestyle{plain}
\pagestyle{plain}

%%%%%% -- PAPER CONTENT STARTS-- %%%%%%%%

\begin{abstract}

    Qubit mapping/routing is a critical stage in compilation for both near-term and fault-tolerant quantum computers, yet existing scalable methods typically impose several times the routing overhead in terms of circuit depth or duration. This inefficiency stems from a fundamental disconnect: compilers rely on an abstract routing model (e.g., three-$\mathrm{CX}$-unrolled $\mathrm{SWAP}$ insertion) that completely ignores the idiosyncrasies of native gates supported by physical devices.

    Recent hardware breakthroughs have enabled high-precision implementations of diverse instruction set architectures (ISAs) beyond standard $\mathrm{CX}$-based gates. Advanced ISAs involving gates such as $\mathrm{\sqrt{iSWAP}}$ and $\mathrm{ZZ}(\theta)$ gates offer superior circuit synthesis capabilities and can be realized with higher fidelities. However, systematic compiler optimization strategies tailored to these advanced ISAs are lacking.

    To address this, we propose \canopus, a unified qubit mapping/routing framework applicable to diverse quantum ISAs. Built upon the canonical representation of two-qubit gates, \canopus\ centers on qubit routing to perform deep co-optimization in an ISA-aware approach. \canopus\ leverages the two-qubit canonical representation and the monodromy polytope theory to model the synthesis cost for more intelligent $\mathrm{SWAP}$ insertion during qubit routing. We also formalize the commutation relations between two-qubit gates through the canonical form, providing a generalized approach to commutativity-based optimization. Experiments show that \canopus\ consistently reduces routing overhead by 15\%-35\% compared to state-of-the-art methods across various backend ISAs and device topologies. More broadly, this work establishes a coherent method for co-exploration of program patterns, quantum ISAs, and hardware topologies, yielding concrete guidelines for hardware-software co-design. This is the first practical demonstration of how to efficiently utilize advanced quantum ISAs, opening the door to designing more powerful and synergistic quantum systems.

    % We have for the first time demonstrated that advanced quantum ISAs can be efficiently utilized within a unified routing framework, paving the way for more effective co-design of quantum software and hardware.
    %  Our work also presents a coherent method to simultaneously take into consideration the quantum program structure, the ISA, and the hardware topology. For the first time, we have demonstrated that efficient routing is possible for more advanced quantum ISAs, paving the way for more effective co-design of quantum software and hardware.

\end{abstract}


\section{Introduction}\label{sec:introduction}

Quantum computing is a revolutionary computational paradigm leveraging quantum mechanical principles such as superposition and entanglement of qubit states~\cite{nielsen2010quantum}. It has grown rapidly in recent decades due to the potential speedup in tasks such as integer factorization~\cite{shor1994algorithms}, solving linear equations~\cite{harrow2009quantum}, and simulation of quantum systems~\cite{lloyd1996universal}. 

\begin{figure}[tbp]
    \centering
    \includegraphics[width=\columnwidth]{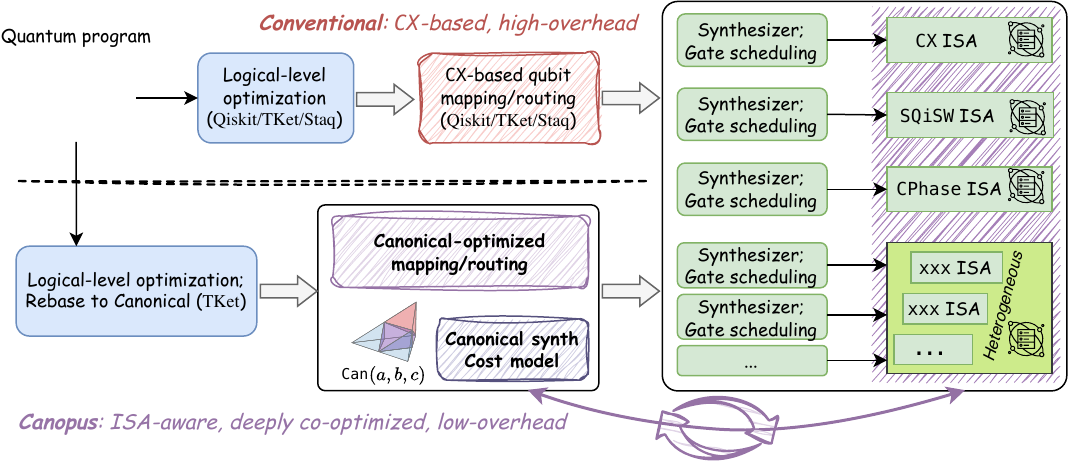}
    \caption{Compilation workflows by means of conventional approaches (top) and \canopus\ (bottom) targeting diverse quantum ISAs. \canopus\ integrates the synthesis cost model (monodromy polytopes within the Weyl chamber) to consider backend ISA properties during the routing stage, enabling deeply co-optimized, ISA-aware compilation across heterogeneous hardware backends. \canopus\ routing operates in the 2Q canonical representation while the specific synthesis is completed by the backend synthesizer.}
    %  \canopus\ exhibits a ISA-aware and deep co-optimized to achieve lower routing overhead.
    \label{fig:motivation}
\end{figure}

Holistic benchmarks of quantum computers such as quantum volume~\cite{cross2019validating} are predicated on concurrent advancements in both hardware and software. Recently, numerous systematic techniques regarding compiler optimization and architecture design have been presented to push the limit of hardware performance. Quantum compilers play a pivotal role in this process, translating high-level programs into executable instructions, usually the native single-qubit (1Q) and two-qubit (2Q) gates on realistic quantum hardware. This typically involves several stages: (1) compiling programs into basic quantum gates, (2) performing hardware-agnostic (logical-level) circuit optimization, (3) resolving backend topology constraints via qubit placement and routing, and (4) converting circuits to native gates for further optimization and scheduling. The primary goal of compiler optimization is to lower the 2Q gate count and circuit depth while resolving backend constraints, with a particular emphasis on 2Q gates due to their significantly higher error rates compared to 1Q gates.

For mainstream quantum platforms such as superconducting qubits~\cite{krantz2019quantum}, 2Q gates can only operate between nearest-neighbor physical qubit pairs (e.g., Google's devices with 2D square topology~\cite{arute2019quantum}, IBM's devices with 2D heavy-hex topology~\cite{chamberland2020topological}). Consequently, qubit placement and $ \SWAP $-based routing are crucial for resolving this connectivity constraint by dynamically remapping logical qubits to physical ones by inserting $ \SWAP $ gates acting on adjacent physical qubit pairs. This introduces a routing overhead that typically increases the gate count and circuit depth by a factor of 2--5$\times$ relative to the pre-mapped circuits when using state-of-the-art (SOTA) scalable routing methods~\cite{li2019tackling,zulehner2018efficient,zhang2021time,liu2023tackling}. Therefore, mitigating this routing overhead remains a central and long-standing challenge in compiler optimization. 
% \note{Importantly, the core principle we advocate---evaluating routing permutations contextually via ISA-aware synthesis instead of assigning a fixed $ \SWAP $ decomposition cost---can potentially generalize to non-$\SWAP$ movement primitives in future work, provided corresponding cost models are established.}

% Within the scope of NISQ and low-level fault-tolerant compilation, we target static, nearest-neighbor superconducting topologies, as opposed to dynamically reconfigurable architectures (e.g., neutral-atom shuttling),

% In the scope of NISQ-era and low-level fault-tolerant compiler optimization, especially targeting static coupling graphs with nearest-neighbor connectivity in superconducting systems that differ from architectures with dynamically reconfigurable connectivity (e.g., neutral-atom shuttling), 

\note{Within the scope of NISQ and low-level fault-tolerant compilation especially for static, nearest-neighbor superconducting topologies---as opposed to dynamically reconfigurable systems like neutral atoms---}most studies on qubit routing rely on a simplified routing model, where circuit cost is quantified by the $ \CX $-based gate count and circuit depth while each $ \SWAP $ gate is unrolled into three $ \CX $ gates according to the textbook pattern $ \SWAP_{q_0,q_1}=\CX_{q_0,q_1}\CX_{q_1,q_0}\CX_{q_0,q_1} $. However, this $ \CX $-centric view is misaligned with the physical reality of modern quantum devices. Although quantum algorithms are typically expressed in terms of $ \CX $ gates, the underlying hardware may not execute native $ \CX $-equivalent gates, nor does this gate cost or circuit cost quantification method accurately reflect the true operational cost. Indeed, beyond the native support for $ \CX $-equivalent gates (e.g., $ \CZ $~\cite{krantz2019quantum}, Cross-Resonance~\cite{rigetti2010fully}, Mølmer-Sørensen~\cite{bruzewicz2019trapped}), modern quantum hardware increasingly features diverse native 2Q basis gates in recent years. These alternative basis gates, or the abstracted instruction set architectures (ISAs) in a narrow sense, can be more powerful than $ \CX $-equivalent gates in terms of synthesis capabilities \note{(i.e., the efficiency of decomposing arbitrary two-qubit unitaries into hardware-native basis gates)} and fidelity of realization, such as $ \SQiSW $~\cite{huang2023quantum}, the $ \iSWAP $-family and $ \CX $-family fractional gates~\cite{mckinney2024mirage,qiskitXXDecomposer}, and heterogeneous basis gates~\cite{peterson2022optimal,mckinney2024mirage}. With such ISAs, $ \SWAP $ can be implemented with a lower cost than three $ \CX $ gates or even be natively realized with high fidelity~\cite{wei2024native,chen2025efficient,nguyen2024programmable}. Therefore, the simplified routing model completely ignores the backend ISA properties, severely limiting the potential of compiler optimization. Furthermore, the absence of systematic compiler optimization methods across these diverse (even complex, heterogeneous) ISAs has prevented the community from fully exploiting their power and exploring the rich software-hardware co-design space.

In our work, we propose a unified qubit mapping/routing framework \canopus\ (\fullNameOfCanopus) tailored to diverse quantum ISAs. Unlike conventional $ \CX $-based routing approaches, \canopus\ is fundamentally ISA-aware. As illustrated in \Cref{fig:motivation}, it considers the properties of the target ISA by formulating an appropriate cost model to facilitate deep co-optimization of routing and synthesis. By means of the canonical 2Q gate representation~\cite{zhang2003geometric}, \canopus\ fully exploits the synthesis capabilities of the given ISA \note{}. This approach demonstrates that advanced ISAs can achieve significantly lower routing overhead than conventional models suggest.

The main ideas of \canopus\ are as follows: \ding{172} Significant optimization opportunities emerge when native gate synthesis costs are directly incorporated into the qubit routing process. For instance, synthesizing a 2Q block and a subsequent $\SWAP$ with the same qubit pair acted on as a single composite operation is often more efficient than synthesizing them individually. \ding{173} Expanding the quantum ISA is crucial for boosting the performance of real-world quantum applications. For example, the fractional $ \ZZ(\theta) $ gate set widely adopted by hardware vendors (e.g., IBM~\cite{ibmFractionalGates}, Quantinuum~\cite{quantinuumArbitraryAngleGates}, IonQ~\cite{ionqPartialGates}) enables more efficient execution of chemistry simulation kernels within which many 2-local Pauli rotations are involved. The combination of $ \CX $ and $ \iSWAP $ gates have been demonstrated to benefit stabilizer circuits to protect error-corrected qubit information~\cite{zhou2024halma}. \ding{174} The monodromy polytope theory~\cite{peterson2020fixed} based on the canonical representation of 2Q gates~\cite{zhang2003geometric} provides a formal, universal, and quantitative description of the 2Q synthesis cost for arbitrary quantum ISAs, establishing a foundation for unified compiler optimization. Guided by these insights, \canopus\ performs intelligent $ \SWAP $ insertion during qubit routing to holistically minimize post-mapping circuit cost (in terms of both gate count and depth) given any quantum ISA, thus performing deep routing-synthesis co-optimization with significantly lower routing overhead induced. Importantly, while \canopus\ is ISA-aware, it always operates on the canonical-form circuits, and the gate/circuit cost quantification via monodromy polytope is independent of backend's specific ISA rebase implementation. In this sense, \canopus\ offers LLVM-style compiler optimization.

% Our framework can be extended to integrate more fine-grain hardware information such as qubit-specific basis gate fidelities.

Experimental results demonstrate that \canopus\ consistently provides 15\%-35\% reduction (in terms of both gate count and depth) of routing overhead compared to other SOTA methods across representative quantum ISAs, \emph{including the conventional \CXISA\ ISA}. This cross-ISA comparison also reveals some consistent or program-specific and topology-specific guidelines for hardware-software co-design. 
% Furthermore, our case studies of real-machine experiments for QFT kernel execution \DD{What's that?} on 1D chain topology and the end-to-end QEC circuit simulation demonstrate the practical superiority of \canopus\ in both NISQ and fault-tolerant applications. 
Source code and data are available on \href{https://github.com/Youngcius/canopus}{GitHub}~\cite{canopusGitHub}.
% Source code and data are available via the \canopusGitHub\footnote{\url{https://anonymous.4open.science/r/canopus-isca2026-FC02/}}.
Our work makes the following key contributions:

% \note{Our work addresses the \dquote{Babel Tower dilemma} in quantum compilation by establishing a canonical language for diverse two-qubit gates, enabling unified optimization across heterogeneous quantum ISAs.} 

% \ding{182} We propose \canopus, the first unified qubit routing framework applicable to diverse quantum ISAs. Unlike conventional methods fixed on a $\CX$-based routing paradigm, \canopus\ is fundamentally ISA-aware and operates on the canonical-form circuits, able to exploit the unique capabilities of any given quantum ISA.
% designed to operate across diverse quantum ISAs. Unlike conventional methods fixed on a $\CX$-based model, \canopus\ is fundamentally ISA-aware, enabling it to exploit the unique synthesis capabilities of any given hardware.
  
\ding{182} We utilize the canonical 2Q gate representation and the monodromy polytope theory to quantify costs of 2Q gates and the overall circuit. This formal approach accurately guides synthesis-routing co-optimization and cross-ISA evaluation. 

\ding{183} We formalize the analysis of commutation relations between arbitrary 2Q canonical gates that share one qubit. This offers a generalized commutativity-based optimization mechanism, moving beyond those tailored only for $\CX$ gates~\cite{liu2022not}.
% \item We propose sophisticated heuristic algorithms and data structures that balance the gate count related circuit cost and circuit depth related circuit cost.
  
\ding{184} We conduct comprehensive experiments across a wide range of real-world benchmarks, hardware topologies, and representative ISAs, showing that \canopus\ consistently reduces routing overhead by 15\%-35\% compared to scalable SOTA methods. Our results also yield holistic guidelines for the co-design of quantum programs, ISAs, and hardware.

\ding{185} We confirm that theoretically expressive ISAs exhibit superior performance to the conventional \CXISA\ ISA, challenging the conclusions of prior works~\cite{kalloor2024quantum}. We demonstrate some co-design guidelines for ISA-program-topology co-exploration. %Furthermore, \canopus\ confirms that the routing-synthesis co-optimization is a unified, highly-effective compiler optimization paradigm to utilize advanced quantum ISAs.

Our case studies, including the real-machine QFT kernel execution and the end-to-end QEC circuit simulation, unequivocally showcase \canopus' superiority in both near-term and fault-tolerant applications. For example, on the task of mapping QFT on 1D chain topology, \canopus\ finds the provably optimal routing scheme, surpassing the results previously reported as optimal in prior work~\cite{zhang2021time}; and experiments on IBM's QPUs demonstrate that, compared to \qiskit, \canopus\ reduces errors by an average of 26.89\% and 34.98\% for the $\CZ$ and $\ZZ(\theta)$ gate sets, respectively.

\section{Background}\label{sec:background}

% \begin{figure}[tbp]
%     \centering
%     \begin{minipage}[t]{0.48\columnwidth}
%         \centering
%         \includegraphics[height=0.5\columnwidth,trim={0.2cm 0 0 0.2cm},clip]{figures/qubit_mapping.pdf}
%         \caption{\small Mapping/routing to resolve physical-qubit topology constraints via $ \SWAP $ insertion.}
%         \label{fig:qubit_mapping}
%     \end{minipage}
%     \hfill
%     \begin{minipage}[t]{0.48\columnwidth}
%         \centering
%         \includegraphics[height=0.5\columnwidth,trim={0.8cm 0.8cm 0.8cm 0.8cm}]{figures/weyl_chamber.pdf}
%         \caption{\small Geometric illustration of canonical gates confined to the Weyl chamber.}
%         \label{fig:weyl_chamber}
%     \end{minipage}
% \end{figure}

% execute logical, long-range 2Q gates on physical qubits, through finding a good initial logical-to-physical layout  and dynamically remapping qubits by inserting $ \SWAP $ gate to route qubit states to near-neighbor positions step by step.

\subsection{Qubit mapping/routing}

\begin{figure}[tbp]
    \centering
    \includegraphics[width=0.9\columnwidth]{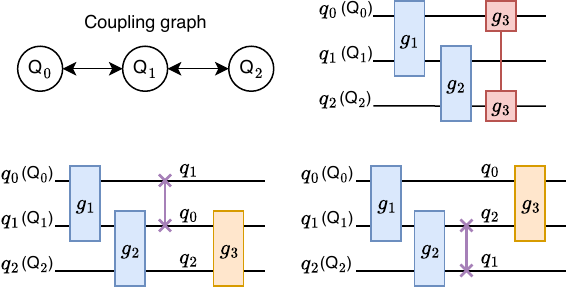}
    \caption{Mapping/routing to resolve topology constraints via $ \SWAP $ insertion. With the initial mapping $ \{q_i$: $Q_i\} $ (upper right), $ g_3 $ is not hardware compliant. Both $\SWAP_{q_0, q_1}$ and $ \SWAP_{q_1,q_2} $ are sufficient to make $ g_3 $ executable.}
    \label{fig:qubit_mapping}
\end{figure}

Real quantum hardware typically has connectivity constraints, whereas algorithms often assume arbitrary interactions. To execute quantum circuits on topology-constrained hardware, logical qubits must first be mapped to physical qubit positions. This is called the initial mapping. In most cases, even an optimal initial mapping cannot guarantee all logical 2Q gates are mapped on physically connected qubit pairs. The common solution is to dynamically change logical-to-physical qubit mappings by inserting $ \SWAP $ gates, as a $ \SWAP $ gate exchanges state subspaces of two operand qubits, such that non-adjacent logical qubit states can be moved next to each other. Therefore, the qubit placement and routing compilation stage takes a logical circuit and hardware coupling graph as the input and outputs a transformed circuit within which each 2Q gate, with respect to a qubit mapping, is hardware compliant. An example is depicted in \Cref{fig:qubit_mapping}.

\subsection{Canonical description of 2Q gates}

\begin{figure}[tbp]
    \centering
    \includegraphics[width=0.8\columnwidth]{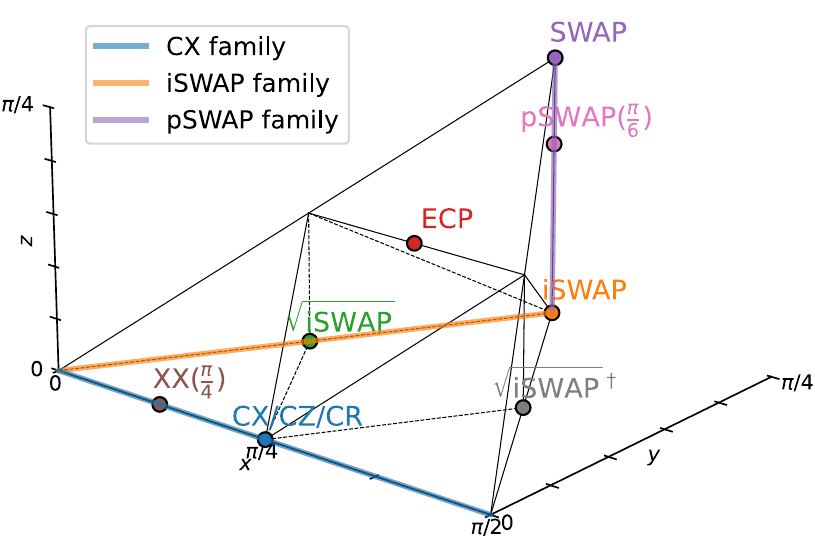}
    \caption{Geometric illustration of canonical gates confined to the Weyl chamber. For visualization convenience, herein the Weyl chamber is confined to $\left\{\frac{\pi}{4}\geq x \geq y \geq z\geq 0\right\} \cup \left\{\frac{\pi}{4} \geq \frac{\pi}{2}-x \geq y \geq z \geq 0\right\}$, equivalent to the canonical coefficient convention $\left\{(a,b,c)\,|\, \frac{1}{2}\geq a\geq b\geq |c| \right\}$.}
    % \caption{Geometric illustration of canonical gates confined to the Weyl chamber. E.g., $ \CX/\mathrm{CZ}/\mathrm{CR} \sim \Can(1/2,0,0)$; $ \SQiSW \sim \Can(1/4, 1/4,0)$; $ \pSWAP(\pi/6,0,0)\sim\Can(1/2,1/2,1/3) $, $ \ECP\sim\Can(1/2,1/4,/1/4) $. $ \CX $ family: $ \Can(a,0,0) $; $ \iSWAP $ family: $ \Can(a,a,0) $, $ \pSWAP $ family: $ \Can(1/2, 1/2, c) $.}
    \label{fig:weyl_chamber}
\end{figure}

Any 2Q gate can be represented by a $4\times 4$ matrix in $ \mathbf{SU}(4) $, up to a global phase, with its canonical form defined as:
\begin{definition}[Canonical gate]
    Any 2Q gate $ U \in \mathbf{SU}(4)$ can be expressed by the composition of its unique \emph{canonical} form
    \begin{align*}
        \Can(a,b,c) := e^{-i \frac{\pi}{2}(a\, XX + b\, YY + c\, ZZ)},\, \frac{1}{2} \geq a \geq b \geq \lvert c \rvert
        % U = (A_0\otimes A_1) \Can(a,b,c) (B_0\otimes B_1) = (A_0\otimes A_1) e^{-i \frac{\pi}{2}(a\, XX + b\, YY + c\, ZZ)} (B_0\otimes B_1)
    \end{align*}
    sandwiched by 1Q gates such that we say $ U $ is locally equivalent to ($\sim$) the canonical form $ \Can(a,b,c) $. %  $ U = (A_0\otimes A_1) \Can(a,b,c) (B_0\otimes B_1) $.
\end{definition}
The canonical coefficients $ (a,b,c) $ are confined to a tetrahedron known as the \emph{Weyl chamber}, which provides a geometric representation of all local equivalence classes of 2Q gates~\cite{zhang2003geometric}. \Cref{fig:weyl_chamber} visualizes some common 2Q gates. E.g.,
\begin{itemize}
    % \item The $\CX$, $\CZ$, and $\mathrm{CR}$ gates are all locally equivalent to $\Can(1/2,0,0)$.
    \item $\CX$, $\CZ$, and $\mathrm{CR}$ are all equivalent to $\Can(\frac{1}{2},0,0)$.
    \item $ \CX $ family: $ \XX(\theta) \sim \YY(\theta) \sim \ZZ(\theta) \sim \Can(\frac{\theta}{\pi}, 0, 0) $.
    \item Param-$\SWAP$ family: $\pSWAP(\theta)\sim\Can(\frac{1}{2}, \frac{1}{2}, \frac{1}{2}-\frac{\theta}{\pi})$.
\end{itemize}
% this canonical form is known as the KAK decomposition~\cite{tucci2005introduction} and has been ubiquitously used in quantum computing~\cite{zhang2003geometric,bullock2003arbitrary,zulehner2019compiling,chen2024one}.
In practice, the canonical form is acquired by KAK decomposition~\cite{tucci2005introduction} and has been widely used~\cite{bullock2003arbitrary,chen2024one}. Appendix \ref{appendix:canonical_form} and Appendix \ref{appendix:isa_analysis} provides a more detailed introduction to the canonical form and its properties.

% Although there are other conventions .... 
% This definition aligns with the \code{TK2} operation definition in \tket, ...

\subsection{Gate realization cost on hardware}\label{sec:background_gate_cost}

The transformed circuits via qubit routing will be ultimately converted into basis gates for execution on hardware. Basis gates refer to those natively implemented and calibrated on physical platforms. 
Typical native gates in superconducting platforms are $\mathrm{CR}$~\cite{rigetti2010fully}, $\CZ$, and $\iSWAP$ gates~\cite{krantz2019quantum}. 
% gate~\cite{rigetti2010fully} and $\iSWAP$ on
% It could be $ \CX $-equivalent ones like IBM's Cross-Resonance gate~\cite{rigetti2010fully} or $ \iSWAP $-family gates like $ \SQiSW $ and $ \iSWAP $ on flux-tunable transmons~\cite{huang2023quantum,arute2019quantum}.
The realization cost of basis gates involves multiple aspects, including the benchmarked fidelity, gate duration, calibration efficiency, etc. For example, gates with shorter duration are more likely to achieve high fidelity, as qubit decoherence dominates the noise source; although some gate schemes can now implement more basis gates~\cite{chen2025efficient,nguyen2024programmable}, those with simpler pulse control are more likely to be calibrated with high precision, such as the $ \iSWAP $-family gates on flux-tunable transmons.

% ---specifically, the canonical coordinate is proportional to the coupling Hamiltonian coupling coefficients---are more likely to 

% as the half evolution of $ \iSWAP $, $ \SQiSW $ gate is more easily to be implemented in higher fidelity 

2Q gates are not natively implemented and must be synthesized by native gates. Their realization cost is determined by the basis gates used for synthesis. For example, any 2Q gate can be minimally synthesized by 3 $ \CX $ gates, except for $ \Can(a,b,0) $ for which the required $ \CX $ count is 2. Conventionally, $ \SWAP $ is regarded as 3 times that of $ \CX $ realization cost, while it can also be synthesized by \dquote{1 $ \CX $ + 1 $ \iSWAP $} or \dquote{3 $ \SQiSW $} gates. The monodromy polytope theory was recently proposed to determine the optimal synthesis cost for any 2Q gate given a specific set of basis gates through analysis of local invariants of canonical gates~\cite{peterson2020fixed}. By this method, the set of gates realizable by a specified number of 2Q gates from the basis set, with arbitrary 1Q gates, corresponds to a polytope within the Weyl chamber. For instance, the polytope reachable by 2 $ \SQiSW $ gates with arbitrary 1Q gates is a tetrahedron confined to $ \left\{1/2\geq a \geq b + \lvert c \rvert \right\} $~\cite{huang2023quantum}.

% Conventional iSWAP better than CZ ... however now CZ is dedicatedly implemented with better fidelity than iSWAP

\section{Motivation}\label{sec:motivation}

\paragraph{Limitations of conventional qubit routing models}
Conventional qubit routing models are ill-equipped to exploit the versatility of modern quantum hardware. First, whether optimizing for gate count or circuit depth, they typically assume that a $ \SWAP $ costs three $ \CX $ gates according to the textbook decomposition. This assumption is divorced from hardware reality. For example, a combination of $ \CX $ and $ \iSWAP $ is sufficient to realize a $ \SWAP $ while both $ \CZ $ (locally equivalent to $ \CX $) and $ \iSWAP $ are natively supported on mainstream superconducting platforms like Google's Sycamore~\cite{arute2019quantum}. Such platforms can even directly implement a high-fidelity $ \SWAP $ gate, with the pulse duration only $1.5$ times that of $ \CZ $~\cite{chen2025efficient}. Thus, $ \SWAP $ is not as costly as assumed in previous qubit routing frameworks. Second, while prior works~\cite{tan2021optimal,liu2022not} do assume the cost of a $\SWAP$ is context-dependent, their analysis remains strictly confined to the $\CX$-centric routing model. By relying on this overly simplistic model, conventional routers cannot accurately predict circuit execution costs and remain blind to the substantial optimization opportunities offered by richer, more diverse ISAs.

%  Second, although prior works pointed out that not all $ \SWAP $ gates cost the same when they are inserted on different positions within the circuit, they are still limited to $ \CX $-based representation. For example, \citet{liu2022not} leverages the feature than gate commutativity and selected $ \SWAP $ insertion can lead to lower insertion cost, however, the \dquote{not-all-the-same} cost they quantified is still based on three-$ \CX $-unrolled pattern. Using $ \CX $ gate as the basis is not accurate to predicate the actual circuit cost and thus constrains the co-optimization space for qubit routing.

% For example, the depth-driven $ \SWAP $ insertion strategy in TOQM~\cite{zhang2021time} simply and the commutativity-based

% Although there is attempts~\cite{liu2022not,mckinney2024mirage} under the consensus that not all $ \SWAP $ insertions cost the same, they are either limited to $ \CX $-based 

% \emph{Routing-synthesis co-optimization as the largest optimization space to unlock the superiority of advanced ISAs}

% 为什么协同设计关键。。。单独的2Q sythesis是native。。。而且quantum algorithms are natively constructed into CX gates ... approximate synthesis做无差别的compilation 并不能达到很好的效果，而且扩展性受限（exponential computational computation)

\paragraph{Co-optimization as the key to unlocking the superiority of advanced ISAs}
% Recently, some sophisticated quantum processor architectures have been developed by exploring advanced ISAs beyond the conventional $ \CX $-based paradigm. However, they are mostly still in the proof-of-concept stage, and there is an absence of a systematic approach to exploiting their superiority in real-world applications. 

In response to the limitations of the $\CX$-only paradigm, a new generation of sophisticated quantum processors has emerged, featuring advanced ISAs with more powerful basis gates. Notable examples include the $ \SQiSW $ gate proposed by Huang et al.~\cite{huang2023quantum}, the continuous $ \ZZ(\theta) $ (equivalent to $ \XX(\theta) $, $ \mathrm{ZX}(\theta) $, $ \mathrm{MS}(0,0,\theta/2) $) jointly adopted by major vendors~\cite{ibmFractionalGates,quantinuumArbitraryAngleGates,ionqPartialGates}, and selected fractional or heterogeneous basis gates~\cite{mckinney2024mirage}. Despite their theoretical promise for greater synthesis power and noise resilience, these advanced ISAs have largely remained in the proof-of-concept stage, with no systematic framework to harness their full potential in real-world quantum applications.

% Although designing a good quantum ISA is a complicated task, those basis gates with greater synthesis capability, shorter gate duration, and higher calibration efficiency are capable of providing superior performance for real-world quantum applications. 

% $ \iSWAP $-family or $ \CX $-family gates~\cite{mckinney2024mirage}, and even heterogeneous/combinatorial basis gates. 

% to exploit these diverse basis gates are mostly limited to 2Q unitary synthesis~\cite{huang2023quantum,tan2020optimal} and efficient synthesis for dedicated multi-qubit gates~\cite{tang2024quantum}. Such rebase passes are developed given a specific quantum ISA and cannot solely yield significant benefits for real-world circuit patterns. The brute-force numerical optimization based synthesis method~\cite{davis2019heuristics,wu2020qgo,kukliansky2023qfactor,younis2021qfast} could help explore cross-ISA performance, while the quantum hardware roofline evaluation proposed by \citet{kalloor2024quantum} does not validated apparent advantages of advanced ISAs for generic real-world quantum applications. 

Prior efforts have been narrowly focused on local 2Q or multi-qubit synthesis tasks~\cite{huang2023quantum,tang2024quantum} or brute-force numerical optimizations~\cite{davis2019heuristics,younis2021qfast}. Such rebase passes are tailored to a specific quantum ISA and fail to deliver clear benefits when applied to advanced ISAs in realistic workloads~\cite{kalloor2024quantum}. This has led to a critical question lingering in the community: \dquote{Are these more expressive, noise-resilient ISAs actually better?}
Recently there have been attempts to harness the properties of advanced ISAs, although through manual, ad-hoc heuristics, such as the $ \SQiSW $-based routing-synthesis optimization~\cite{mckinney2024mirage} and the $ \CX $-$ \iSWAP $ based routing for defect effect mitigation~\cite{zhou2024halma}. In our work, we highlight that collaborative compiler optimization, especially at the stage following logical-level circuit optimization and followed by the final ISA rebase pass, is a key to fully exploiting the capabilities of those powerful ISAs: First, high-level algorithms are predominantly expressed in the $\CX$ representation, which then undergo template-based and peephole optimizations that are highly sophisticated and tailored for $\CX$-based circuit patterns (e.g., commutativity, Clifford equivalence); Second, the disconnect between na\"ive qubit routing models and backend ISA properties apparently leaves a large untapped co-optimization space. Thus we aim to validate this point through a systematic ISA-aware routing framework.

\paragraph{The \dquote{Tower of Babel dilemma} for utilizing diverse ISAs}
The proliferation of diverse quantum ISAs---from monolithic to complex, heterogeneous basis gate sets supported by various physical platforms---has created a \dquote{Tower of Babel dilemma} in the architecture and systems community. Developing bespoke compiler optimizations for each unique hardware backend is unsustainable, leading to the same software fragmentation that we have encountered in classical computing. Consequently, it is important to seek a unified approach that can effectively handle various platform-specific abstractions resembling the LLVM compiler~\cite{Lattner2004}.
% Unlike classical computing, quantum computing are based more on formal algebraic principles. 
The recently proposed monodromy polytope theory~\cite{peterson2020fixed}, for example, provides a method for optimal analysis of ISA synthesis capabilities. Specifically regarding the circuit-level compiler optimization, the monodromy polytope with canonical 2Q gate representations offers a unified approach to evaluating circuit cost and modeling routing-synthesis co-optimization. Building on this, \canopus\ proves to be an elegant and unified solution to the Tower of Babel dilemma at the compiler level.

% If compiler optimization fails to respect the unique characteristics of each quantum ISA, the resulting hardware execution will be inefficient and ... 
\paragraph{Coherent cross-ISA, topology, and program pattern co-exploration}
Ultimately, the goal of quantum computing systems is not just to optimize software for existing hardware, but to co-design the entire stack---from algorithms to architecture---to build the most efficient system possible. This requires a holistic exploration of a vast and complex design space, asking critical questions like: Which ISA is best suited for a given class of applications (e.g., quantum error correction vs. quantum simulation)? How does the choice of qubit topology interact with the ISA to affect performance? Answering these questions is currently an ad-hoc, labor-intensive process, hindering systematic progress. 
Therefore, our work aims to provide the missing piece: a unified and automated framework for this co-exploration. By integrating qubit routing with a formal, ISA-aware synthesis cost model, \canopus\ can systematically evaluate the performance of various program patterns across heterogeneous ISAs and diverse hardware topologies.
% In contrast to prior cross-ISA evaluation based on brute-force numerical optimization~\cite{kalloor2024quantum}, our work provides a more coherent evaluation method. \canopus\ is designed to be extended to integrate more fine-grained hardware information such as qubit-specific gate fidelities. It transforms the compiler from a mere optimization tool into a powerful scientific instrument for architectural exploration. 
% This enables researchers and hardware designers to make informed, data-driven decisions, accelerating the systematic discovery of optimal co-design points and paving a systematic path toward robust, fault-tolerant quantum computer systems.
This empowers researchers with informed, data-driven insights to identify optimal co-design points, accelerating the development of robust, fault-tolerant quantum systems.

% Our work aims to provide the missing piece: a unified and automated framework for this essential co-exploration. By leveraging qubit routing with a unified, ISA-aware synthesis cost model, \canopus\ can systematically evaluate the performance of different program patterns (e.g., QFT kernel, QEC stabilizer circuit) across heterogeneous ISAs (e.g., combinations of basis gates and mirror gates, fractional basis gates) and diverse hardware topologies.

% Our work aims to provide the missing piece: a unified and automated framework for this co-exploration. By integrating qubit routing with a formal, ISA-aware synthesis cost model, \canopus\ can systematically evaluate the performance of different program patterns across heterogeneous ISAs and diverse hardware topologies.
% In contrast to prior cross-ISA evaluation based on brute-force numerical optimization~\cite{kalloor2024quantum}, our work provides a more coherent evaluation method. \canopus\ is designed to be extended to integrate more fine-grained hardware information such as qubit-specific gate fidelities. It transforms the compiler from a mere optimization tool into a powerful scientific instrument for architectural exploration. This enables researchers and hardware designers to make informed, data-driven decisions, accelerating the discovery of optimal co-design points and paving a systematic path toward robust, fault-tolerant quantum computer systems.

\section{\canopus\ framework}\label{sec:canopus}

\subsection{Overview}

\begin{figure*}[t]
    \centering
    \includegraphics[width=\textwidth]{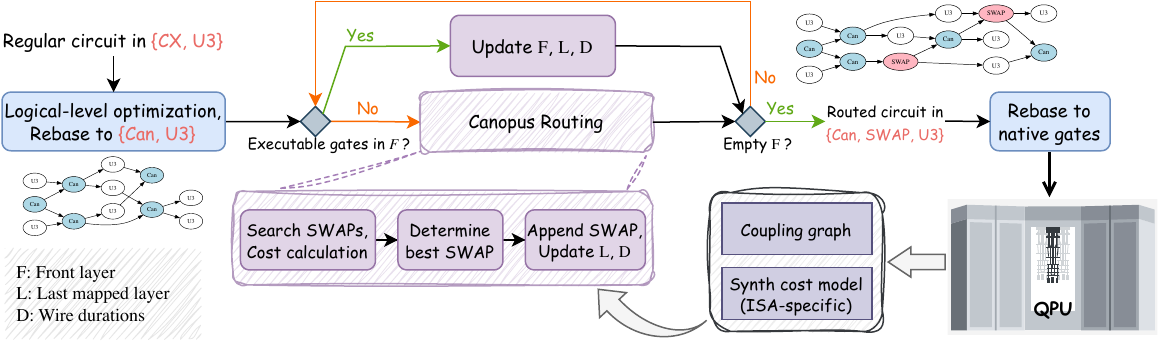}
    \caption{Overview of the \canopus\ framework.}
    %  Input circuit is rebased to $\left\{\Can,\,\Uthree\right\}$.  and output are both circuits in the 2Q canonical form. \canopus\ performs deep co-optimization of routing and synthesis, with the synthesis cost model tailored to the target ISA.}
    \label{fig:overview}
\end{figure*}

The overall qubit routing procedure of \canopus\ is illustrated in \Cref{fig:overview}. Prior to routing, the input circuit is rebased to $\left\{\Can,\,\Uthree\right\}$ gate set. All subsequent processes operate on the directed acyclic graph (DAG) representation of the circuit. During the routing pass, \canopus\ integrates the ISA-specific synthesis cost model into its $ \SWAP $ search process, dynamically determining the most appropriate $ \SWAP $ at each route step. The routing cost is efficiently computed via formal analysis of 2Q canonical forms, without explicitly performing any ISA rebase process. Thus, the output is still a circuit DAG represented in $ \left\{\Can,\,\Uthree\right\} $ with inserted $ \SWAP $ gates.

Notably, \canopus\ inherits the basic concepts and data structures introduced in \sabre~\cite{li2019tackling}, which is one of the industrial-standard qubit mapping/routing algorithms. Given the input circuit DAG, \canopus\ first attempts to map 2Q gates layer by layer via extracting the front layer $ F $, peeling executable gates and searching $ \SWAP $ gates to minimize the routing cost according to a unified heuristic cost function. On the backbone of \sabre, we further introduce several key data structures such as the last mapped layer $ L $ and the wire duration record $ D $, to support the efficient implementation of ISA-aware routing and reducing both gate count and depth related routing overhead.

% Mirroring-SABRE inherits the basic thought and data structures introduced in SABRE~\cite{li2019tackling}. SABRE attempts to map 2Q gates layer by layer via extracting the \dquote{front layer} $F$, peeling executable gates and searching $\SWAP$ gates to minimize the heuristic cost $H$ of the unresolved front layer. $H$ involving the current front layer and a $\SWAP$ candidate is designed to minimize the topological distance between upcoming qubits and promote parallelism. \note{In our mirroring-SABRE algorithm, we additionally define the \dquote{last mapped layer} $L$ as the set of 2Q gates that has no succeeding ones within the DAG constructed by already mapped 2Q gates, such that the $\SWAP$ search process prioritizes $\SWAP$s that the last mapped layer can absorb, as illustrated in \Cref{fig:mirroring_sabre}. We also define and initial heuristic cost function $H_0(F, \mathrm{DAG}, \pi, D)$ before each $\SWAP$ search epoch,
% \begin{align*}
%     H_0 \coloneqq \frac{1}{\left|{F}\right|} \sum\nolimits_{g \in F} D[\pi(g.q_1), \pi(g.q_2)] + \frac{W}{\left|{E}\right|} \sum\nolimits_{g \in E} D[\pi(g.{q_1}),\pi(g.{q_2})],
% \end{align*}
% apart from the original cost function $H(F, \mathrm{DAG}, \pi, D, \SWAP)$ from SABRE. If there is $\SWAP$ candidate that can be absorbed by $L$ and also leads to a lower $H$ than $H_0$, it will be selected as a $\SWAP$ mirroring effect that induces no \#2Q overhead. Otherwise, the $\SWAP$ search is proceeded similar to SABRE.}

\subsection{2Q synthesis cost modeling}\label{sec:2q_synthesis_cost_modeling}

\begin{figure}[tbp]
    \centering
    \begin{tikzpicture}
        % 插入整张 PDF 图片作为基础节点
        \node[anchor=south west, inner sep=0] (image) at (0,0) {\includegraphics[width=\linewidth]{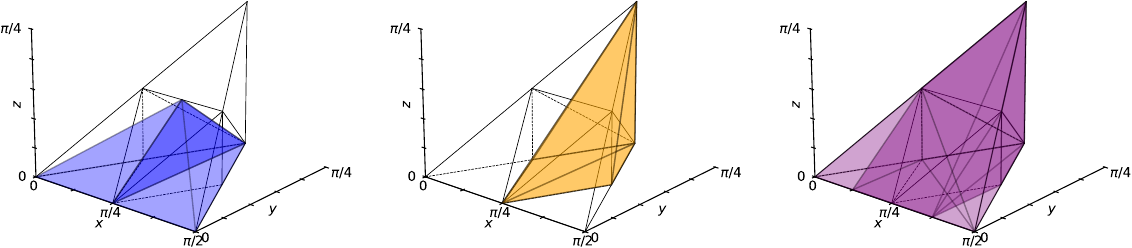}};
        
        % 建立相对坐标系，(0,0) 是左下角，(1,1) 是右上角
        \begin{scope}[x={(image.south east)},y={(image.north west)}]
            % 根据实际视觉效果微调 x 和 y 的坐标数值
            \node at (0.16, -0.12) {\footnotesize(a)};
            \node at (0.50, -0.12) {\footnotesize(b)};
            \node at (0.84, -0.12) {\footnotesize(c)};
        \end{scope}
    \end{tikzpicture}

    \caption{Synthesis coverage for $ \bigl\{ \SQiSW, \ECP \bigr\} $ gate set. The trivial points ($ \SQiSW $ and $ \ECP $ themselves) are not shown in this figure. 2Q overage regions correspond to those that require (a) 2 $ \SQiSW $ gates or 2 $ \ECP $ gates; (b) 1 $ \SQiSW $ + 1 $ \ECP $; (c) 3 gates (3 $ \SQiSW $, 3 $ \ECP $, 2 $ \SQiSW $ + 1 $ \ECP $, etc.) from this gate set for synthesis, respectively.}
    \label{fig:sqisw_ecp_coverage}
\end{figure}
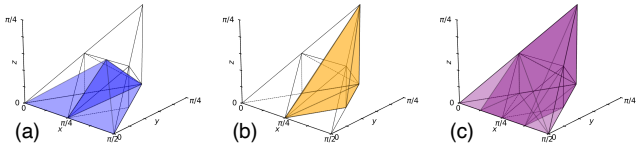

As introduced in \Cref{sec:background_gate_cost}, given any basis gate set, the synthesis cost of a target 2Q gate can be exactly computed through monodromy polytope~\cite{peterson2020fixed}. This cost (which basis gates are sufficient for synthesis) only depends on the canonical coefficients of the target gate. For example, \Cref{fig:sqisw_ecp_coverage} illustrates various polytopes for the gate set $ \bigl\{\SQiSW,\,\ECP\bigr\} $. In practice, the costs of each basis gate are pre-defined, thus the whole set of polytopes helps decide the optimal synthesis scheme with the minimal circuit cost we should prioritize. For example, if $ \SQiSW $ and $ \ECP $ have the same unit cost, the $ \SWAP\sim \Can\bigl(\frac{1}{2},\frac{1}{2},\frac{1}{2}\bigr) $ gate realization will prioritize the \dquote{1 $\SQiSW$ + 1 $ \ECP $} combination; if $ \ECP $ cost is set to be more than twice that of $ \SQiSW $, the $ \SWAP $ realization prioritizes the \dquote{3 $ \SQiSW $} pattern.

\subsection{Routing in canonical form}\label{sec:routing_in_canonical_form}

\begin{figure}[tbp]
    \centering
    % \subfigure[ISA-aware $ \SWAP $ insertion cost. \ZY{Redraw this figure into a 3-6 qubit demonstrative example. E.g., show a local window where routing would insert a SWAP, what 2Q unitaries are around it, and how the re-synthesis step “absorbs” the SWAP (or makes it cheaper). Even a toy 3-6 gate snippet would help a lot.}]{\includegraphics[width=\columnwidth]{figures/canonical_routing.pdf}\label{fig:canonical_routing_a}}
    \subfigure[\note{ISA-aware $\SWAP$ insertion in a local circuit window.}]{\includegraphics[width=\columnwidth]{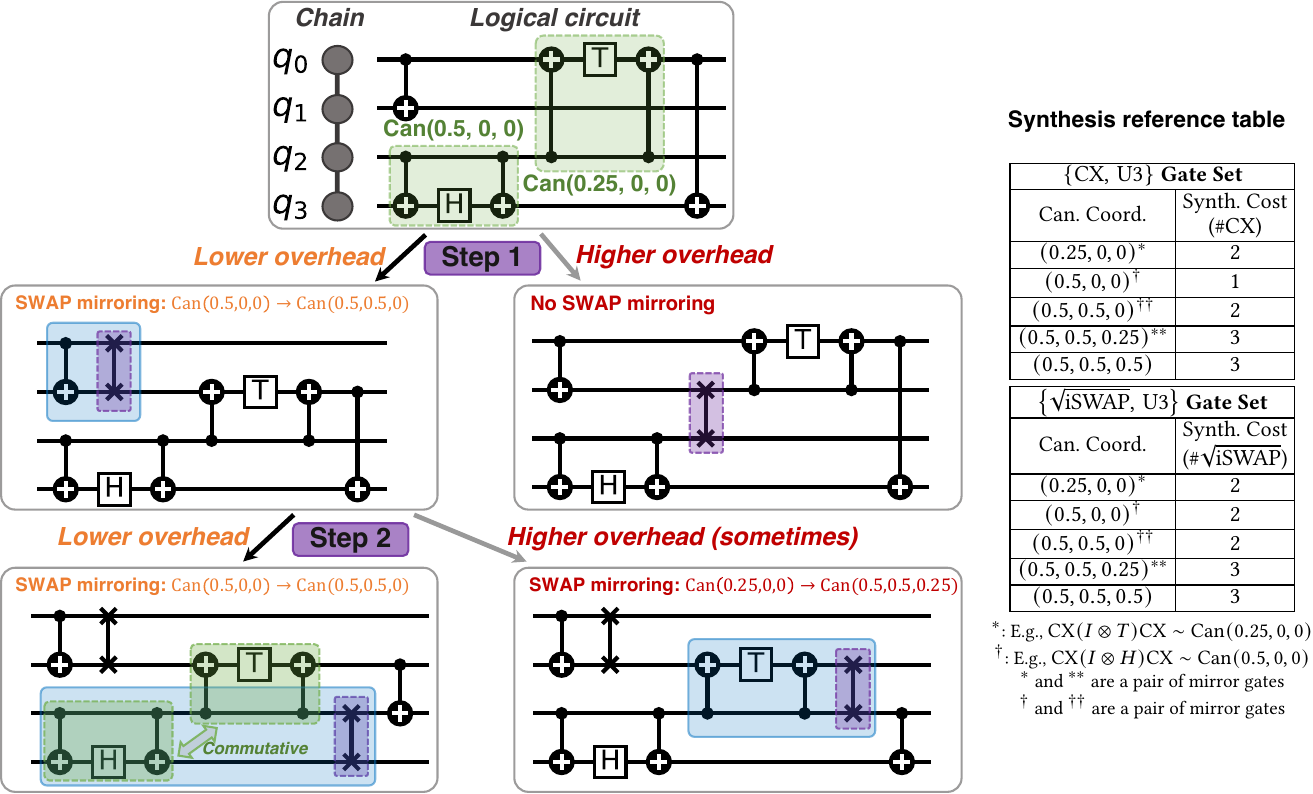}\label{fig:canonical_routing_a}}
    \subfigure[$ \SWAP $ insertion patterns with different gate count and depth costs.]{\includegraphics[width=\columnwidth]{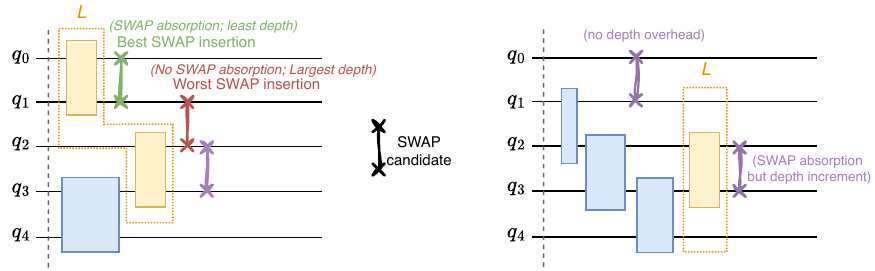}\label{fig:canonical_routing_b}}
    \caption{\note{Qubit routing with the canonical 2Q gate representation.}}
    \label{fig:canonical_routing}
\end{figure}

Our ISA-aware routing primarily leverages the mechanism that some inserted $ \SWAP  $ gates can \dquote{piggyback} a preceding 2Q gate with the same qubit pair acted on and thus result in lower (even negative) routing overhead than what na\"ive $ \SWAP $ synthesis cost may imply. Based on the ISA-specific synthesis cost model, \canopus\ utilizes a holistic heuristic cost function that considers various requirements of qubit routing for simultaneous reduction of both gate count and circuit depth overhead in a unified, quantitative approach.

% incorporating such minimal synthesis cost of $ \SWAP $ insertion, overall circuit depth cost, and quantitative cost  to guide co-optimization.

% \begin{table}[tbp]
%     \centering
%     \caption{Key technical notations used in the \canopus\ framework.}
%     \label{tab:notations}
%     \begin{tabular}{|c|l|}
%         \hline
%         \textbf{Notation} & \textbf{Description} \\
%         \hline
%         $F$, $E$ & Front layer (executable 2Q gates) and extended set \\
%         \hline
%         $L$ & Last mapped layer of 2Q canonical gates without successors \\
%         \hline
%         $D$ & Wire duration record tracking depth on each physical qubit \\
%         \hline
%         $C$ & Record of commutative canonical 2Q gate pairs within $L$ \\
%         \hline
%         $H$ & Unified heuristic cost function for evaluating $\SWAP$ options \\
%         \hline
%         $c_g$ & Minimal synthesis gate count increment of a $\SWAP$ insertion \\
%         \hline
%         $c_{\mathrm{swap}}$ & Na\"ive synthesis cost of an independent $\SWAP$ under the ISA \\
%         \hline
%         $\Delta_{\mathrm{depth}}$ & Estimated circuit depth cost increment incurred by a $\SWAP$ \\
%         \hline
%         $w_g$, $w_d$ & Weighting factors prioritizing gate count and depth metrics \\
%         \hline
%     \end{tabular}
% \end{table}

% exact, minimal synthesis cost quantification and holistic heuristic cost function to guide the co-optimization.

\note{

Instead of treating $\SWAP$ as an independent, fixed-cost insertion, we evaluate its cost based on how it interacts with the \dquote{last mapped layer} $L$, defined as the set of 2Q gates in the current DAG that have no succeeding interactions. When a candidate $\SWAP$ acts on the same physical qubit pair as a gate $U \in L$, it can be \dquote{absorbed} by consolidating them into a single composite unitary $U' = \SWAP \cdot U$, dubbed \dquote{$ \SWAP $ mirroring}, as detailed in Appendix \ref{appendix:swap_mirroring}. This $ \SWAP $ insertion cost is then defined as the marginal synthesis cost increment: $c_g = \textsc{cost}(\SWAP \cdot U) - \textsc{cost}(U)$. The cost component $c_g $ is typically lower than the na\"ive cost $ c_{\mathrm{swap}} $ of an independent $\SWAP$ gate, and it can even be negative when the composite unitary is cheaper to synthesize than $ U $. For instance, under $ \CX $ basis, if the absorption location is an $ \iSWAP $-equivalent gate, the composite $ \SWAP\cdot\iSWAP \sim \mathrm{Can}\bigl(\frac{1}{2},0,0\bigr)$ requires only one $ \CX $ gate to synthesize, leading to a negative gate count increment ($c_g = c_{\mathrm{cx}} - 2\,c_{\mathrm{cx}} = - c_{\mathrm{cx}}$); similarly, with $ \SQiSW $ basis, the resulting $ c_g $ is zero.

As illustrated in \Cref{fig:canonical_routing_a}, \canopus\ evaluates $ \SWAP $ insertions by regarding all 2Q gates/blocks as canonical gates and quantifying their synthesis costs based on the target ISA. Without loss of generality, this example considers only the overhead of $ \SWAP $ insertion, omitting topological distance and circuit depth heuristics. According to the synthesis cost reference table in \Cref{fig:canonical_routing_a}, an independent $\SWAP$ gate normally costs three 2Q gates under both the $\CX$ and $\SQiSW$ gate set. However, in the first $ \SWAP $ search step, absorbing a $\SWAP$ candidate into a preceding $\Can(0.5, 0, 0)$ gate (left selection) forms the mirror gate $\Can(0.5, 0.5, 0)$, merely yielding a marginal synthesis cost increment of $ c_g=1\times c_{\mathrm{cx}} $ or $c_g = 0\times c_{\mathrm{\sqrt{iswap}}}$. In the second step, both selections offer absorbable $\SWAP$ candidates with identical $c_g$ costs under the $\CX$ basis. Yet, targeting the $\SQiSW$ basis prioritizes the left selection ($c_g=0\times c_{\mathrm{\sqrt{iswap}}}$ vs. $1\times c_{\mathrm{\sqrt{iswap}}}$). This example demonstrates how ISA-aware cost evaluation steers routing to effectively exploit the specific synthesis capabilities of the underlying hardware.

% this cost increment is simply the na\"ive synthesis cost of an independent $ \SWAP $.

}

To optimize circuit execution time, we also evaluate the \dquote{circuit depth} cost increment ($\Delta_{\mathrm{depth}}$) by tracking the accumulated duration on each physical qubit wire via a data structure $D$. As \Cref{fig:canonical_routing_b} illustrates, different $\SWAP$ insertion choices yield varying trade-offs between gate count and circuit depth, necessitating a comprehensive consideration. Notably, we quantify circuit depth based on the predefined costs of the underlying basis gates which reflect their physical durations, through tracking the length of the weighted critical path on the mapped DAG. By integrating both gate count and depth costs into a unified heuristic, \canopus\ can make informed decisions that balance these two critical 
metrics. The detailed heuristic cost function is defined as:
\begin{align}
    H  =\,\,  & w_g\, c_{g}  +  w_d\, \Delta_{\mathrm{depth}} \notag\\
    & + (\Delta_{\mathrm{Avg}\{\mathrm{dist}[i,j]\}_F} + k_E \, \Delta_{\mathrm{Avg}\{\mathrm{dist}[i,j]\}_E})\, c_{\mathrm{swap}},\label{eq:heuristic_cost}
\end{align}
where $ w_g $ and $w_d$ weight the count and depth cost components. The final term adapts \sabre's topological heuristic ($H_{\sabre} = \mathrm{Avg}\{\mathrm{dist}[i,j]\}_F + k_E\,\mathrm{Avg}\{\mathrm{dist}[i,j]\}_E$), which relies on the average shortest-path distance between physical qubits mapped to demanded logical interactions in the front layer $F$ and the lookahead extended set $E$. Instead of using absolute topological distances, \canopus\ computes the \dquote{differential} average distance ($\Delta_{\mathrm{Avg}\{\mathrm{dist}\}}$) resulting from a candidate $\SWAP$, scaled by the ISA-specific $\SWAP$ cost ($c_{\mathrm{swap}}$). This securely translates topological distance reduction into a concrete basis-gate cost metric. Because $c_g$ and $\Delta_{\mathrm{depth}}$ provide highly accurate, hardware-aware feedback for count-depth co-optimization, the empirical decay factor originally required in \sabre\ is no longer needed. Ultimately, every term in \Cref{eq:heuristic_cost} represents a marginal cost increment, allowing the heuristic to holistically minimize routing overhead.

% Furthermore, the decay factor in \sabre\ is no longer needed, as the $ c_g $ and $\Delta_{\mathrm{depth}}$ guide more accurate count-depth co-optimization. Therefore, each cost component in \Cref{eq:heuristic_cost} refers to \dquote{cost increment}, composing a unified heuristic cost function to reconcile multifaced routing costs. 

% the routing process is proceeded similar to \sabre, but with the ISA-aware cost components.

% In contrast to the regular heuristic cost function used in \sabre:
% \begin{align}
%     H_{\sabre} &= \frac{1}{|F|} \sum\nolimits_{(i,j)\in F} \mathrm{dist}[i,j] + \frac{k_E}{|E|}\sum\nolimits_{(i, j)\in E}\mathrm{dist}[i,j]\notag\\
%     &= \mathrm{Avg}\{\mathrm{dist}[i,j]\}_F + k_E\,\mathrm{Avg}\{\mathrm{dist}[i,j]\}_E 
%     % &= \widebar{\mathrm{dist}(F)} + k_E\,\widebar{\mathrm{dist}(E)} \notag
% \end{align}
% which involves the basic (left term) and lookahead (right term) components.
% In practice, there is a $ w_{\mathrm{decay}} $ decay factor applied to $H$, which is not shown as it does not affect the composition of $H$.

% - Unified and highly-effective qubit routing approach in canonical form, with properties of quantum ISAs tailored to the routing process

\subsection{Enhanced optimization via commutation}\label{sec:gate_commutation_guided_optimization}

\begin{figure}[tbp]
    \centering

    \subfigure[\note{Efficient $\SWAP$ absorption via canonical commutation relations.}]{\includegraphics[width=\columnwidth]{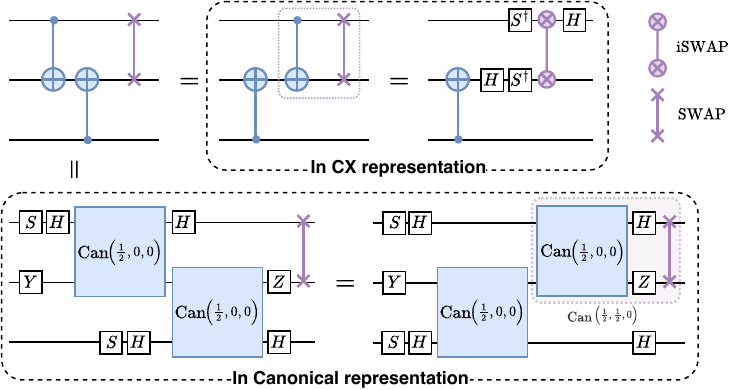}\label{fig:commutation_a}}
    \subfigure[\note{More commutation pattern examples captured by the canonical form.}]{\includegraphics[width=\columnwidth]{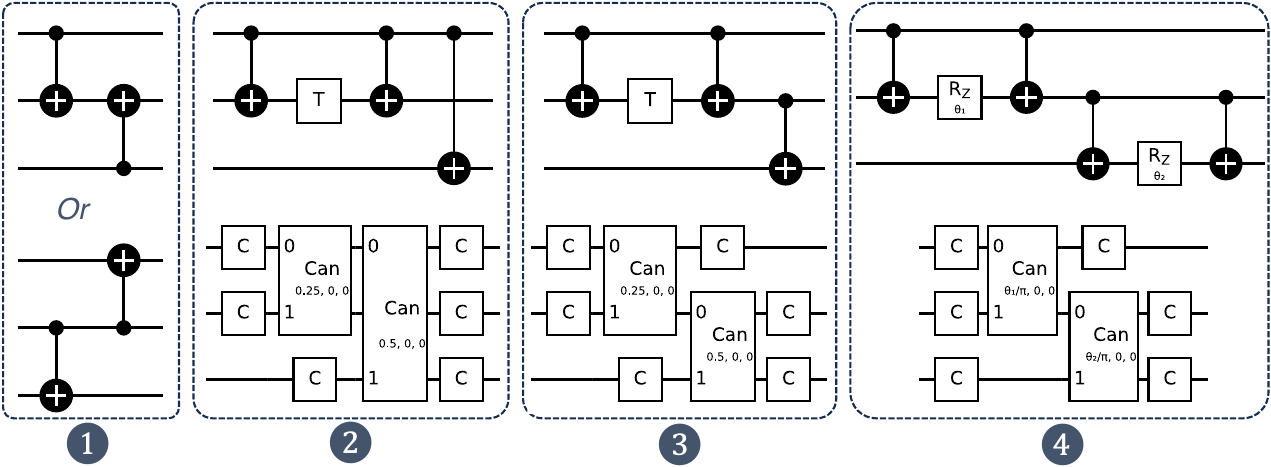}\label{fig:commutation_b}}
    \caption{\note{Canonical representation efficiently captures commutative relations in real-world quantum circuits. (a) The canonical commutation relation enhances $ \SWAP $ absorption opportunities in a formal and efficient manner. Herein commutativity within $\CX$ chain can be identified without tracking control and target qubit positions. (b) Additional commutation patterns captured in the canonical form. The first pattern is intuitive in the standard $\CX$ basis, while the subsequent three highlight complex equivalences obscured in the $\CX$ basis but clearly exposed in canonical form ($C$ denotes 1Q Clifford).}}
    \label{fig:commutation}
\end{figure}

Previous works have observed that employing the commutativity between $ \CX $ gates exposes more optimization opportunities for $ \SWAP $ insertion~\cite{liu2022not}. However, the commutation pattern they exploit is limited to a pair of $ \CX $ gates, where they either act on the same control qubit or target qubit. In our findings, the general 2Q gate commutativity can be captured through the canonical form:
\begin{theorem}[Canonical gate commutation]\label{thm:commutation}
    Let $\,\Can(a,b,c)_{q_0,q_1}$ and $\,\Can(a',b',c')_{q_1,q_2}$ denote canonical gates acting on qubits ($q_0, q_1$) and ($q_1, q_2$) respectively, with an overlapping qubit $ q_1 $. They are commutative if and only if
    \begin{align}
        b=b'=c=c'=0,
    \end{align}
    that is, when both consist solely of $\,\XX$ rotations.
\end{theorem}
% Detailed proof is in the \canopusGitHub.
Detailed proof is in Appendix \ref{appendix:commutation_proof}.
Through this formalized commutativity determination, the ordinary $\CX$ commutation pattern can be captured without tracking the control and target qubit positions, as shown in \Cref{fig:commutation_a}.
\note{Moreover, \Cref{fig:commutation_b} showcases additional commutation patterns that are captured in the canonical form but remain difficult to handle for $\CX$-based compilers. These patterns are commonly observed in real-world circuits (e.g., arithmetic, QFT, chemistry simulation) and the transformation to commutative canonical gates can be readily obtained using \tket.}

% \ZY{In practice, while the pure XX condition is restrictive, it occurs frequently in specific algorithms like QAOA or transverse-field Ising model simulations, where it provides an additional $\sim 3-5\%$ reduction in routing overhead.}

% - Capture optimization opportunities exposed by gate commutation; while commutation relations can be uniformly described in canonical form

% canonical gate representation also offers a generalized approach to capturing commutativity-based optimization opportunities.

% \subsection{Qubit dependencies guided optimization}\label{sec:qubit_dependencies_guided_optimization}
% - Capture optimization opportunities exposed by qubit dependencies, which implies optimization in a more global scope

% \documentclass{article}
% \usepackage[ruled,vlined,linesnumbered]{algorithm2e}
% \usepackage{amsmath}

% % Optional: Define a command for code-like text if you prefer a specific style
% \newcommand{\code}[1]{\texttt{#1}}

% \begin{document}

\begin{algorithm}[tbp]
    \caption{Update $L$ when adding a new 2Q gate}
    \label{alg:update_L_D_C}

    \SetKwInOut{KwIn}{Input}
    \SetKwInOut{KwOut}{Output}

    \KwIn{
        $G'$ (Routed DAG), 
        $\pi$ (current logic-to-physical mapping), 
        $L$ (last mapped layer), 
        $D$ (wire durations for each qubit), 
        $C$ (commutative pairs within $L$)
    }
    \KwOut{Updated $G'$, $L$, $D$, $C$}

    \BlankLine

    % \tcc{$g$: resolved logical gate;\, $g'$: routed gate}
    \tcc{\footnotesize $g$: resolved logical gate;\, $g'$: routed gate}
    $g' \leftarrow G'.\textsc{pushBack}(g, \pi[g.q_0], \pi[g.q_1])$;\, \tcp{\footnotesize $g'.q_i = \pi[g.q_i]$}

    $d \leftarrow \textsc{max}(D[g'.q_0],\, D[g'.q_1]) + \textsc{synthCost}(g)$\;
    $D[g'.q_0] \leftarrow d$; $D[g'.q_1] \leftarrow d$\;

    \For{$\mathrm{pred} \in G'.\textsc{predecessors}(g')$}{
        \If{$\textsc{is2QGate}(\mathrm{pred})$}{
            % \tcp{Distinguish if they are a pair of commutative canonical gates (Theorem 1)}
            \If{$\textsc{isCommutativeCanonicalPair}(g',\, \text{pred})$}{
                $C[(\mathrm{pred}.q_0,\, \text{pred}.q_1)] \leftarrow (g'.q_0,\, g'.q_1)$\;
            }
            \Else{
                $L.\textsc{pop}((\mathrm{pred}.q_0,\, \mathrm{pred}.q_1),\, \textsc{None})$\;
                $C.\textsc{pop}((\mathrm{pred}.q_0,\, \mathrm{pred}.q_1),\, \textsc{None})$\;
            }
        }
        \Else{ 
            \tcc{\footnotesize $ \mathrm{pred\_pred} $ must be None or a 2Q gate}
            $\mathrm{pred\_pred} \leftarrow \textsc{next}(G'.\textsc{predecessors}(\mathrm{pred}))$\;
            \If{$\mathrm{pred\_pred} \neq \textsc{None}$}{
                $L.\textsc{pop}((\mathrm{pred\_pred}.q_0,\, \mathrm{pred\_pred}.q_1),\, \textsc{None})$\;
                $C.\textsc{pop}((\mathrm{pred\_pred}.q_0,\, \mathrm{pred\_pred}.q_1),\, \textsc{None})$\;
            }
        }
    }
    $L[(g'.q_0,\, g'.q_1)] \leftarrow g'$\;

\end{algorithm}

% \end{document}

% \documentclass{article}
% \usepackage[ruled,vlined,linesnumbered]{algorithm2e}
% \usepackage{amsmath} % For \text
% \usepackage{amssymb} % For \Delta

% % Optional: Define a command for code-like text
% \newcommand{\code}[1]{\texttt{#1}}

% \begin{document}

    \begin{algorithm}[tbp]
    \caption{Update $D$ when adding a $ \SWAP $ gate}
    \label{alg:update_durations_swap}

    \SetKwInOut{KwIn}{Input}
    \SetKwInOut{KwOut}{Output}

    \KwIn{
        \code{swap} (encountered SWAP gate), 
        \code{can} (canonical gate within $ L $ on the same qubits as \code{swap}),
        % $D$ (wire durations), 
        % $C$ (commutative pairs within $ L $)
        $ D $, $ C $
    }
    \KwOut{Updated $D$}

    \BlankLine

    % \If{$(\code{swap}.q_0, \code{swap}.q_1) \in C$}{
    %     \tcc{Adjust $ D $ for commutative pair}
    %     $q'_0,\, q'_1 \leftarrow C[(\code{swap}.q_0,\, \code{swap}.q_1)]$\;
    %     \uIf{$\code{swap}.q_0 = q'_0$}{
    %         $D[\code{swap}.q_0] \leftarrow D[q'_0] + \textsc{synthCost}(\text{can})$\;
    %         $D[\code{swap}.q_1] \leftarrow D[\code{swap}.q_0]$\;
    %     }
    %     \ElseIf{$\code{swap}.q_0 = q'_1$}{
    %         $D[\code{swap}.q_0] \leftarrow D[q'_1] + \textsc{synthCost}(\text{can})$\;
    %         $D[\code{swap}.q_1] \leftarrow D[\code{swap}.q_0]$\;
    %     }
    %     \ElseIf{$\code{swap}.q_1 = q'_0$}{
    %         $D[\code{swap}.q_1] \leftarrow D[q'_0] + \textsc{synthCost}(\text{can})$\;
    %         $D[\code{swap}.q_0] \leftarrow D[\code{swap}.q_1]$\;
    %     }
    %     \ElseIf{$\code{swap}.q_1 = q'_1$}{
    %         $D[\code{swap}.q_1] \leftarrow D[q'_1] + \textsc{synthCost}(\text{can})$\;
    %         $D[\code{swap}.q_0] \leftarrow D[\code{swap}.q_1]$\;
    %     }
    % }

    \If{$(\code{swap}.q_0, \code{swap}.q_1) \in C$}{
        % \tcc{Adjust $ D $ for commutative pair}
        $q'_0,\, q'_1 \leftarrow C[(\code{swap}.q_0,\, \code{swap}.q_1)]$\;
        \tcc{\footnotesize Adjust $ D $ by finding matched qubits $q_i \in \{\code{swap}.q_0, \code{swap}.q_1\}$ and $q'_j \in \{q'_0, q'_1\}$}
        $D[q_i] \leftarrow D[q'_j] + \textsc{synthCost}(\text{can})$\;
        $D[\text{the other \code{swap} qubit}] \leftarrow D[q_i]$\;
    }
    % \BlankLine

    % \tcc{Calculate circuit depth increment}
    % $\Delta_{\text{depth}} \leftarrow \textsc{synthCost}(\code{can}.\textsc{mirror}()) - \textsc{synthCost}(\text{can})$\;
    $d \leftarrow \textsc{max}(D[\code{swap}.q_0],\, D[\code{swap}.q_1]) + \textsc{synthCost}(\code{can}.\textsc{mirror}()) - \textsc{synthCost}(\text{can})$\;
    $D[\code{swap}.q_0] \leftarrow d$; $D[\code{swap}.q_1] \leftarrow d$\;

\end{algorithm}

% \end{document}

\subsection{Scalability and implementation}\label{sec:canopus_implementation}

The overall algorithm framework to implement \canopus\ resembles \sabre. To efficiently implement the sophisticated $\SWAP$ insertion mechanism in \canopus, we develop specific core algorithms. \Cref{alg:update_L_D_C} specifies how the essential data structures---the last mapped layer $L$, commutative canonical gate pairs $C$ within $L$, wire duration record $D$---will be updated when adding an executable 2Q gate to the routed circuit DAG. \Cref{alg:update_durations_swap} shows how the wire durations $D$ should be correctly updated when encountering a $\SWAP$ insertion that can exploit the canonical gate commutativity optimization opportunity. That is also crucial to evaluate the total circuit cost after mapping. Notably, all the computation processes within these algorithms are based on conditional control and operations on hashed data structures, achieving $\mathcal{O}(1)$ time complexity. The synthesis cost of a target 2Q gate is quantified by identifying the convex polytope containing its canonical coordinate, for which the computation process is highly efficient with linear time complexity. \canopus\ also caches canonical gate costs it has computed to avoid repetitive computation. Consequently, the overall scalability of \canopus\ is on par with that of \sabre, ensuring its practical applicability to large-scale circuits.

\note{For the specific hyperparameter values, we set $k_E$ to $0.5$, consistent with \sabre. Both $w_g$ and $w_d$ are also set to $0.5$. This configuration ensures that the synthesis-aware optimization significantly influences routing decisions without overshadowing the primary objective of minimizing topological distance. The depth weight $w_d$ is further scaled by a topology-adaptive factor $\bar{d}/(2+\bar{d})$, where $\bar{d}$ is the average degree of the device coupling graph, reflecting that depth optimization is more impactful on denser topologies.
% \delbyzy{Specifically, in sparse topologies, prioritizing topological distance or gate count optimization within the current search window is more effective. Conversely, denser topologies inherently offer shorter routing distances, which diminishes the discriminative power of the distance penalty and provides greater opportunities for depth optimization.}
% This configuration effectively balances gate count and depth optimizations within the synthesis-routing co-optimization framework. 
The sensitivity of these choices is evaluated in \Cref{sec:sensitivity_analysis}.}

\note{
    The implementation of \canopus\ which is accessible on \href{https://github.com/Youngcius/canopus}{GitHub}~\cite{canopusGitHub} builds on \code{qiskit}, \code{monodromy}, and \code{pytket}, along with additional self-implemented utilities. The core routing algorithm is realized as a native \qiskit\ \code{TransformationPass}, allowing seamless integration into existing \qiskit\ transpilation pipelines without any refactoring. Extending \canopus\ to a new ISA requires only a simple configuration step---specifying the unit costs for the target basis gates---without requiring any algorithmic modification.
}

% \NOTE{
% \begin{itemize}
%     \item Architectural fit — it's a native Qiskit TransformationPass; no refactoring of existing pipelines needed.
%     \item ISA extensibility — adding a new gate set is a configuration step (monodromy polytope pre-computation), not algorithm modification.
%     \item Demonstrated breadth — 7 ISAs already supported with the same code.
%     \item Multi-SDK interop — converter utilities for pytket and BQSKit are already in the codebase.
%     \item Deployment simplicity — single pip install ..
% \end{itemize}
% }

% \section{Implementation}\label{sec:implementation}

% 这一章节暂时扔到附录中

% \input{algorithms/update_L_D_C.tex}

% \input{algorithms/depth_trans.tex}

% \subsection{Core algorithms}

% With key heuristic

% \subsection{Scalability}\label{sec:scalability}

\section{Case Studies}\label{sec:studies}

We validate the practical advantages of \canopus\ through two realistic case studies: the real-machine execution of quantum Fourier transform (QFT) circuits on IBM's QPU \code{ibm\_marrakesh}, and the end-to-end simulation of quantum low-density parity-check (qLDPC) stabilizer measurement circuits to assess its impact on the logical error rate.

\subsection{QFT kernel}\label{sec:qft_study}

\begin{figure}[tbp]
    \centering
    \subfigure[Mapping \code{qft\_6} by \toqm.]{
        \includegraphics[width=0.47\columnwidth]{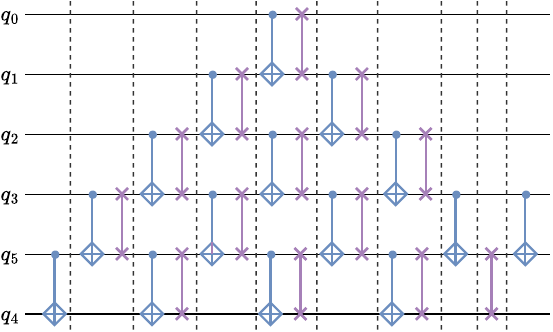}
        \label{fig:qft_mapping_a}
    }\subfigure[Mapping \code{qft\_6} by \canopus.]{
        \includegraphics[width=0.47\columnwidth]{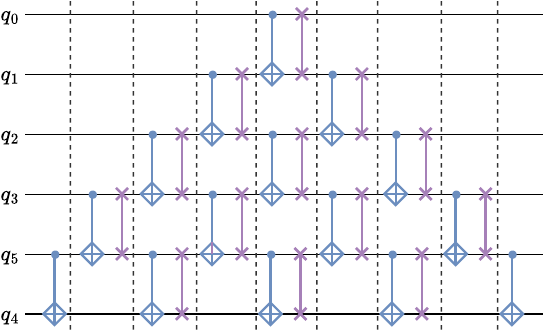} % ,trim={0.5cm 0.25cm 0 0},clip
        \label{fig:qft_mapping_b}
    }
    \caption{Mapping/routing comparison for the QFT kernel. For convenient visualization, only $ \CPhase $ and $ \SWAP $ gates are shown. (a) \toqm\ generates a sub-optimal mapping scheme, with 2Q depth of 10. (b) \canopus\ generates the optimal scheme in a perfect butterfly structure, with 2Q depth of 9.}
    \label{fig:qft_mapping}
\end{figure}

\begin{table}[tbp]
    \centering
    \caption{Qubit routing comparison for the QFT kernel.}
    \label{tab:qft_example}
    \begin{scriptsize}
        \begin{tabular}{|c|l|c|c|c|c|}
    \hline
    \multicolumn{2}{|c|}{\textbf{Benchmark}} & \multicolumn{2}{c|}{\textbf{\code{qft\_6}}} & \multicolumn{2}{c|}{\textbf{\code{qft\_12}}} \\
    \hline
    Topology & Method & \#Can & Depth2Q & \#Can & Depth2Q \\
    \hline
    \multirow{3}{*}{1D Chain} & \emph{Optimal} & \emph{15}\cellcolor{purple!30} & \emph{9}\cellcolor{purple!30} & \emph{66}\cellcolor{purple!30} & \emph{21}\cellcolor{purple!30} \\
    \cline{2-6}
     & \toqm & 16 & 10 & 67 & 22 \\
    \cline{2-6}
     & \canopus & 15\cellcolor{purple!30} & 9\cellcolor{purple!30} & 66\cellcolor{purple!30} & 21\cellcolor{purple!30} \\
    \hline
    \multirow{2}{*}{2D Square} & \toqm & 21 & 13 & 100 & 39 \\
    \cline{2-6}
     & \canopus & 15\cellcolor{purple!30} & 9\cellcolor{purple!30} & 75 \scriptsize{(±10\%)} & 33 \scriptsize{(±10\%)}\\
    \hline
\end{tabular}
    
    \end{scriptsize}
\end{table}

QFT is a fundamental subroutine in many promising quantum algorithms like Shor's algorithm~\cite{shor1994algorithms} and quantum phase estimation~\cite{kitaev1995quantum}. Amid extensive research on dedicated QFT compilers~\cite{zhang2021time,jin2024optimizing,maslov2007linear}, we select the specialized SOTA \toqm~\cite{zhang2021time} as our primary baseline.

A key finding is that \canopus\ always achieves the optimal QFT routing scheme on the 1D chain topology, while \toqm\ does not. It can be proven that the minimal number of $ \SWAP $ insertions to route an $ n $-qubit QFT is $ \frac{n(n-1)}{2} - 2 $, that is, 2 fewer than the original $ \CPhase $ count. This results in a perfect, symmetric butterfly circuit structure, as exemplified in \Cref{fig:qft_mapping_b}, with minimal \#$\Can$ and 2Q circuit depth. Notably, this result is indeed optimal, surpassing the manually designed scheme previously reported as optimal by Maslov~\cite{maslov2007linear} where 2 more $ \SWAP $ gates are required. This optimal scheme is irrespective of the target ISA. In contrast, our experiments show that \toqm\, despite claiming to realize the scheme from \cite{maslov2007linear}, fails to reproduce it and consistently yields inferior results to \canopus, as illustrated in \Cref{fig:qft_mapping}.

We compare compilation performance for both $ 6 $- and $ 12 $-qubit QFT kernels on both 1D chain and 2D square topologies, with results summarized in \Cref{tab:qft_example}. On the 1D chain, \canopus\ always produces the theoretically optimal routing result, while \toqm\ does not. For the small-scale \code{qft\_6} kernel on the 2D square, \canopus\ also achieves the optimal routing, superior to \toqm\ in both \#$\Can$ and 2Q depth. For the large-scale \code{qft\_12} kernel, \canopus\ consistently outperforms \toqm\ in both metrics.

\begin{figure}[tbp]
    \centering
    \includegraphics[width=\columnwidth]{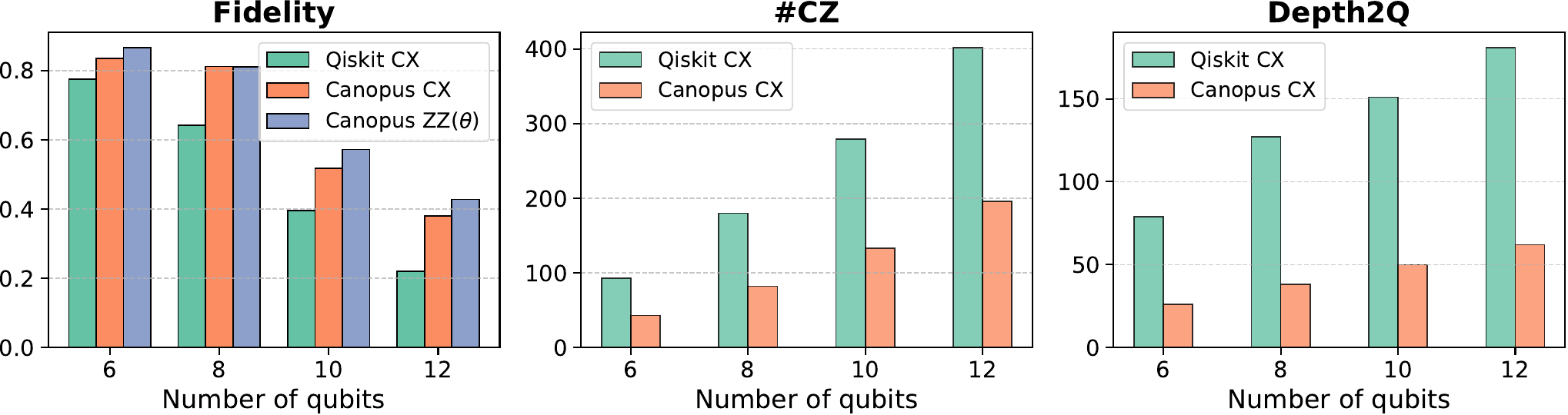}
    \caption{QFT kernel fidelity comparison benchmarked on IBM\textsuperscript{®} Quantum Platform (\code{ibm\_marrakesh}). \code{ibm\_marrakesh} is a Heron-R2 QPU with native gate set $ \bigl\{\mathrm{CZ},\, \sqrt{\mathrm{X}},\, \mathrm{Z}(\theta), \mathrm{ZZ}(\theta)\bigr\} $.}
    \label{fig:qft_cloud}
\end{figure}

To further validate these results, we performed real-machine experiments on IBM's \code{ibm\_marrakesh} QPU. We compiled QFT circuits of sizes $n \in \{6, 8, 10, 12\}$ for a 1D chain topology using both \canopus\ and the default \qiskit\ compiler. Although \code{ibm\_marrakesh} has a heavy-hex topology, it contains linear chains of sufficient size for these benchmarks. Fidelity was measured using the Hellinger fidelity between the experimental and ideal output distributions, with the number of shots set to $\textsc{max}\{4096, 2^n \times 10\}$. A layer of Hadamard gates is appended to each circuit execution so that the ideal final state will be $\lvert 0\rangle^{\otimes n}$. In \Cref{fig:qft_cloud}, circuits compiled with \canopus\ achieve, on average, a 52.9\% reduction in $\CZ$ gate count, a 66.4\% reduction in 2Q-gate depth, and a 26.89\% error reduction for the $\CZ/\CX$ and 34.98\% for the $\ZZ(\theta)$ gate set, respectively, compared to \qiskit\ with default settings. These results unequivocally demonstrate the practical advantages of \canopus\ for QFT kernel compilation.

\subsection{qLDPC stabilizer circuit}

% \begin{figure}[tbp]
%     \centering
%     \includegraphics[width=\linewidth]{figures/better_check.pdf}
%     \caption{Stabilizer check circuit with less routing overhead. \ZY{delete this figure}}
%     \label{fig:stabilizer}
% \end{figure}

% \begin{figure}[tbp]
%     \centering
%     \includegraphics[width=\linewidth]{figures/ler_topo_square.pdf}
%     \includegraphics[width=\linewidth]{figures/ler_topo_hhex.pdf}
%     \caption{Logical error rate of QLDPC stabilizer circuits compiled for square (top) and heavy-hex (bottom) topologies. The y-axis shows the relative logical error rate normalized by the error rate of an ideal baseline that assumes all-to-all connectivity. \ZY{This figure is now confusing}}
%     \label{fig:stabilizer_result}
% \end{figure}

\iffalse
\begin{figure}[tbp]
    \centering
    \includegraphics[width=\linewidth]{figures/qldpc_routing.pdf}
    \caption{Conceptual schematic of the compilation benefits of qLDPC code stabilizer measurement circuit in Canopus. Here we illustrate the $[[18, 4, 4]]$ BB code~\cite{wang2026demonstration} with generating polynomial as $A = 1 + y + xy$ and $B = 1 + x + xy$ on the grid topology, showing a representative Z-stabilizer check as an example. The orange and blue circles denote the right and left data qubits, while the pink and green squares denote the Z- and X-stabilizer ancilla qubits, respectively. The illustrated Z-ancilla qubit needs to interact with data qubits 1-6 to perform the stabilizer measurement.}
    \label{fig:conceptual_schematic_qldpc}
\end{figure}
\fi

\begin{figure}[tbp]
    \centering
    \includegraphics[width=0.49\linewidth]{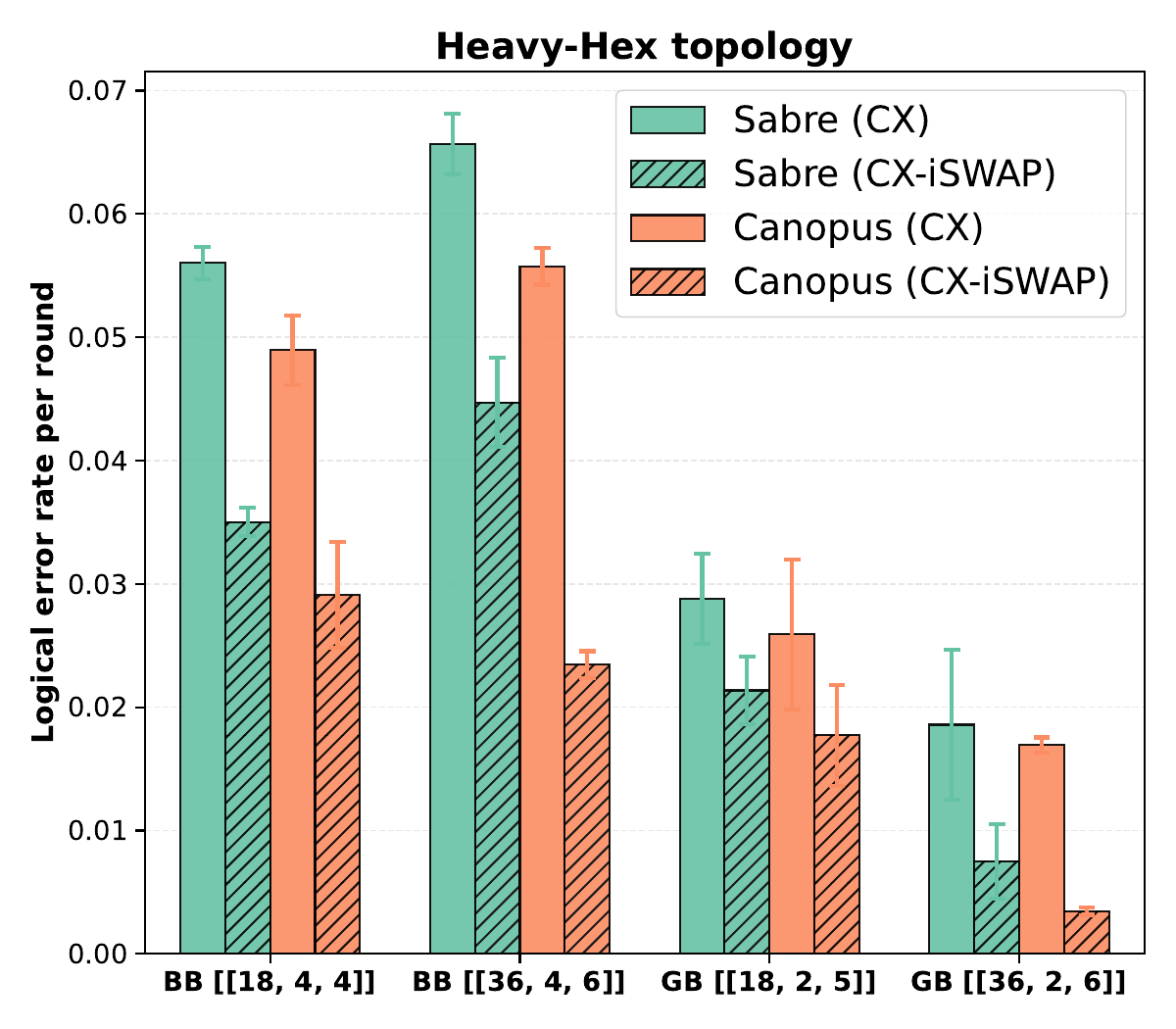}
    \includegraphics[width=0.49\linewidth]{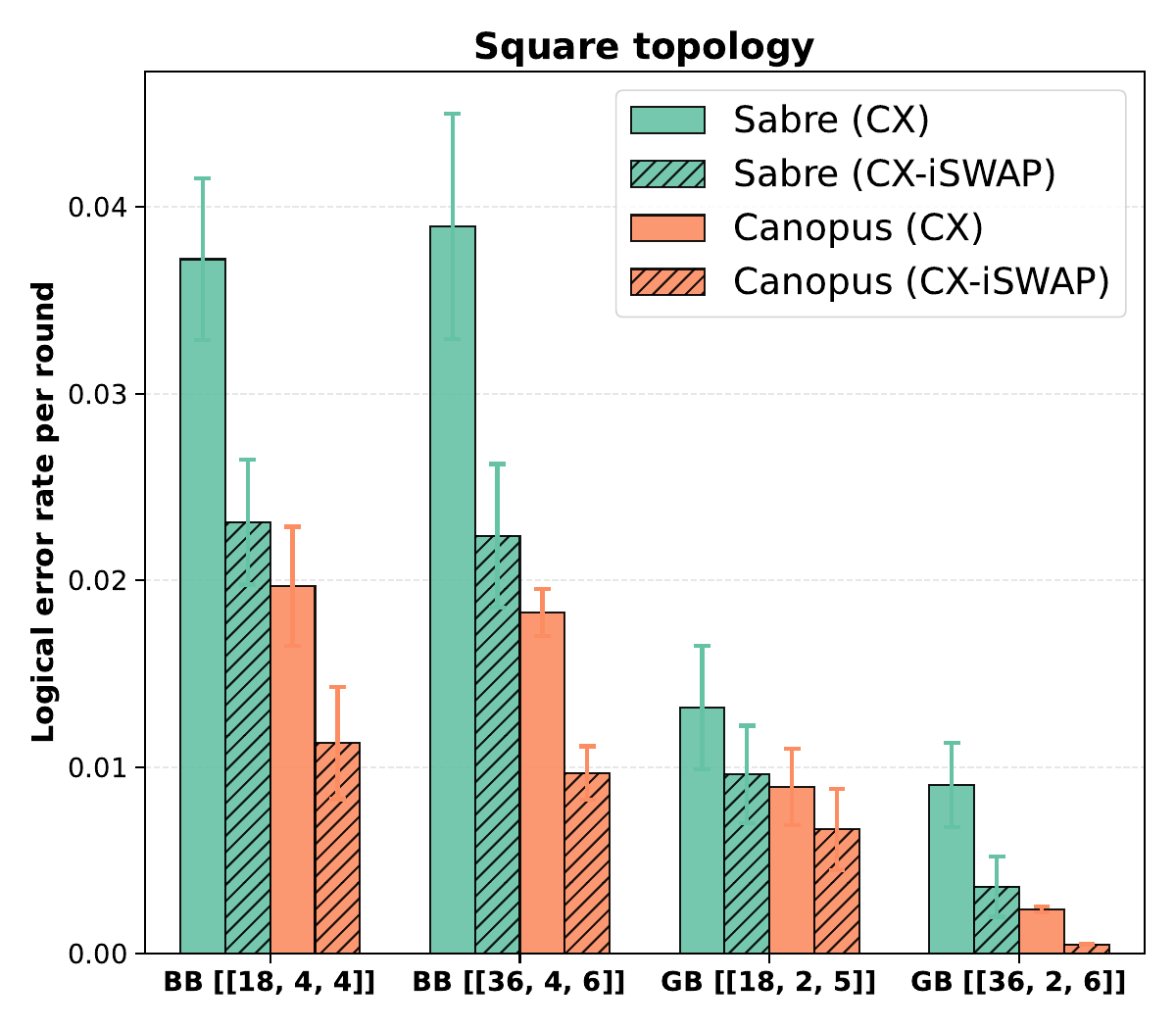}
    \caption{Logical error rates with error correction via qLDPC stabilizer circuits compiled for 2D heavy-hex (left) and square (right) topologies.}
    \label{fig:stabilizer_result}
\end{figure}

For our second case study, we shift to the fault-tolerant quantum computing (FTQC) context by looking at an important class of quantum error correction circuit---the stabilizer measurement circuit for qLDPC codes. qLDPC codes are rapidly moving from a topic of theoretical interest to a cornerstone of experimental FTQC research, mainly because of their superior encoding efficiency~\cite{bravyi2024high, breuckmann2021quantum}. However, due to their frequent long-range interactions for stabilizer measurement~\cite{breuckmann2021quantum, panteleev2021degenerate}, realizing qLDPC codes on superconducting processors with fixed, local connectivity is still hampered by significant routing overhead~\cite{wang2026demonstration}.

% Furthermore, we highlight that the $ \CX $-$ \iSWAP $ combinatorial native gate set particularly offers significant superiority than the only $ \CX $ basis gate for QEC. 

% We observe that stabilizer circuits for QLDPC code implementations are composed largely of $ \CX $ (or $ \CZ $) gates~\cite{peham2025automated, litinski2019game}, while $ \CX $ $ \iSWAP $ is a well-known pair of mirror~\cite{cross2019validating} gates differing by a $ \SWAP $ gate. Consequently, an ISA incorporating both $ \iSWAP $ and $ \CX $ leads to significant opportunities to \dquote{piggyback} a $ \SWAP $ insertion on a $ \CX $ without incurring extra 2Q gate count. A similar observation was also employed in \citet{zhou2024halma} to handle defects in the surface code, but it relied heavily on manual design and experience.

We demonstrate that the ISA-aware optimization mechanism of \canopus\ is crucial to mitigating the routing overhead across a diverse set of qLDPC codes. Here we attempt to compile the stabilizer measurement circuits with two ISAs: (1) \CXISA\ ISA with $ \CX $ as the 2Q basis gate; (2) \StabISA\ ISA with both $ \CX $ and $ \iSWAP $ as basis gates, assumed to have an identical cost.
Particularly, the \StabISA\ ISA aligns with practical hardware realities, e.g., both $ \CZ $ and $ \iSWAP $ can be natively supported by mainstream superconducting platforms~\cite{krantz2019quantum,wei2024native,arute2019quantum}. In addition, an ISA incorporating both $ \iSWAP $ and $ \CX $ leads to significant opportunities to \dquote{piggyback} a $ \SWAP $ insertion on a $ \CX $ without incurring extra 2Q gate count, as the composite block is equivalent to an $\iSWAP$, enabling the possibility of optimizing qubit routing overhead during the execution of stabilizer measurements. 
% \ZK{For example, as schematically illustrated in Fig.~\ref{fig:conceptual_schematic_qldpc}, merging the $ \CX $ and $ \SWAP $ gates between the Z-stabilizer ancilla and data qubit 4 into a single $\iSWAP$ gate can reduce the two-qubit gate overhead during the routing process.}
% which was also employed in recent works~\cite{zhou2024halma, yoshida2025low} for specific qLDPC codes.

We further build an end-to-end evaluation pipeline with qLDPC code examples from~\cite{wang2026demonstration, panteleev2021degenerate}, including the generalized bicycle (GB) and bivariate bicycle (BB) codes. We simulate the standard memory experiments using \code{stim}~\cite{Gidney2021StimAF} to evaluate the fault-tolerant performance of our compiled stabilizer measurement circuits, under the same circuit-level noise model as described in~\cite{bravyi2024high}. Finally, all syndromes are decoded using the BP-OSD decoder~\cite{panteleev2021degenerate, hillmann2024localized} to determine the logical qubit error rate.

As shown in \Cref{fig:stabilizer_result}, \canopus\ consistently achieves lower logical error rates than \sabre, as the ISA-aware approach of \canopus\ results in compiled circuits with less $ \CX/\iSWAP $ gate count and circuit depth. Under the \CXISA\ ISA, \canopus\ yields an average logical error suppression of 49.4\% on the square topology and 11.4\% on the heavy-hex topology compared to \sabre. The advantage becomes even more pronounced with the \StabISA\ combinatorial ISA, where \canopus\ achieves a 52.6\% (square) and 29.3\% (heavy-hex) error suppression, resulting from that there are many opportunities for $ \SWAP $ insertions piggybacked on $ \CX $ gates without incurring extra 2Q gate count. These results highlight two key findings: first, the ISA-aware mechanism in \canopus\ is highly effective for compiling QEC circuits, and second, the dedicated use of a hybrid $\CX$-$\iSWAP$ gate set offers a significant practical advantage for qLDPC code demonstrations on superconducting hardware.

\section{Evaluation}\label{sec:evaluation}

We further holistically evaluate \canopus\ compared to other leading methods, across representative ISAs and hardware topologies. The evaluation provides both cross-compiler and cross-ISA comparisons under the coherent settings for basis gate cost and routing overhead metric.

\begin{table}[tbp]
    \centering
    \caption{Selected quantum ISAs.}
    \setlength{\tabcolsep}{2pt}
    \begin{footnotesize}
        \begin{tabular}{|l|l|m{3cm}|} %l  l
    \hline
    \textbf{ISA} & \textbf{2Q basis gates} & \textbf{Description} \\
    \hline
    \CXISA &
    $\left\{ \text{CX} \right\}$ &
    Conventional CX gate \\
    \hline % Adds a bit of space between rows for readability
    \ZZPhaseISA &
    $\bigl\{ \ZZ_{\frac{\pi}{6}}, \ZZ_{\frac{\pi}{4}}, \ZZ_{\frac{\pi}{2}} \bigr\}$ &
    Discrete CX-family gates, i.e., $\bigl\{\sqrt[3]{\CX}, \sqrt{\CX}, \CX \bigr\}$~\cite{peterson2022optimal} \\
    \hline
    \SQiSWISA &
    $\bigl\{ \SQiSW, \iSWAP \bigr\}$ &
    Half evolution of $\iSWAP$ and $\iSWAP$~\cite{huang2023quantum} \\
    \hline
    \ZZPhaseWithMirrorISA &
    \ZZPhaseISA\, $+$ 
    $\bigl\{\pSWAP_{\frac{\pi}{6},\, \frac{\pi}{4},\, \frac{\pi}{2}}\bigr\}$ &
    \ZZPhaseISA\ ISA with the mirror gates \\
    \hline
    \SQiSWWithMirrorISA &
    \SQiSWISA\, $+$ $\left\{ \ECP, \CX \right\}$ &
    \SQiSWISA\ ISA with the mirror gates~\cite{mckinney2024mirage} \\
    \hline
    \HetISA &
    \ZZPhaseISA\, $+$ \SQiSWISA &
    Heterogeneous CX-family and iSWAP-family gates \\
    \hline
\end{tabular}

% \begin{tblr}{
%   width = \linewidth, % Set the table to the full text width
%   colspec = {|Q[l,m]|Q[c,m,mode=math]|X[l,m]|}, % Define column specifications
%   hlines, % Add horizontal lines for all rows
%   row{1} = {font=\bfseries, c}, % Style the header row (bold and centered)
% }
%     % Header
%     ISA & \text{2Q basis gates} & \text{Description} \\
    
%     % Content
%     \CXISA &
%     \left\{ \mathrm{CX} \right\} &
%     Conventional CX \\

%     \ZZPhaseISA &
%     \bigl\{ \ZZ_{\frac{\pi}{6}}, \ZZ_{\frac{\pi}{4}}, \ZZ_{\frac{\pi}{2}} \bigr\} &
%     Discrete CX-family gates, i.e., $\bigl\{\sqrt[3]{\CX}, \sqrt{\CX}, \CX \bigr\}$\\

%     \SQiSWISA &
%     \bigl\{ \SQiSW, \iSWAP \bigr\} &
%     $\iSWAP$ and its half evolution\\

%     \ZZPhaseWithMirrorISA &
%     \ZZPhaseISA + 
%     \bigl\{\pSWAP_{\frac{\pi}{6},\, \frac{\pi}{4},\, \frac{\pi}{2}}\bigr\} &
%     \ZZPhaseISA\ with the mirror gates \\
    
%     \SQiSWWithMirrorISA &
%     \SQiSWISA + \left\{ \ECP, \CX \right\} &
%     \SQiSWISA\ with the mirror gates \\

%     \HetISA &
%     \ZZPhaseISA + \SQiSWISA &
%     Heterogeneous CX-family and iSWAP-family \\
% \end{tblr}

    \end{footnotesize}
    \label{tab:isa_setting}
\end{table}

\begin{table}[tbp]
    \centering
    \caption{Benchmarks information. These metrics are collected from \tket-optimized logical circuits with only $\Can$ and $\Uthree$ gates. Circuit cost ($ \countCost $ and $ \depthCost $) is calculated in \CXISA\ ISA.}
    \label{tab:benchmark}
    \begin{footnotesize}
        \begin{tabular}{|l|r|r|r|r|r|}
\hline
\textbf{Program} & \textbf{\#Qubit} & \textbf{\#Can} & \textbf{Depth2Q} & $ C_{\mathrm{count}} $ & $ C_{\mathrm{depth}} $ \\
\hline
bigadder~\cite{li2023qasmbench} & 18 & 114 & 79 & 130.0 & 88.0 \\
\hline
bv~\cite{li2023qasmbench} & 19 & 18 & 18 & 18.0 & 18.0 \\
\hline
ising~\cite{li2023qasmbench} & 26 & 25 & 2 & 50.0 & 4.0 \\
\hline
knn~\cite{li2023qasmbench} & 25 & 72 & 50 & 84.0 & 62.0 \\
\hline
multiplier~\cite{li2023qasmbench} & 15 & 198 & 122 & 222.0 & 133.0 \\
\hline
qec9xz~\cite{li2023qasmbench} & 17 & 32 & 12 & 32.0 & 12.0 \\
\hline
qft~\cite{quetschlich2023mqt} & 18 & 153 & 33 & 306.0 & 66.0 \\
\hline
qpeexact~\cite{quetschlich2023mqt} & 16 & 127 & 43 & 260.0 & 86.0 \\
\hline
qram~\cite{li2023qasmbench} & 20 & 110 & 70 & 130.0 & 78.0 \\
\hline
sat~\cite{li2023qasmbench} & 11 & 210 & 182 & 252.0 & 204.0 \\
\hline
swap\_test~\cite{li2023qasmbench} & 25 & 72 & 50 & 84.0 & 62.0 \\
\hline
wstate~\cite{li2023qasmbench} & 27 & 52 & 28 & 52.0 & 28.0 \\
\hline
\end{tabular}

    \end{footnotesize}
\end{table}

% \begin{table}[tbp]
%     \centering
%     \caption{Routing overhead in terms of $ \countCost $ for different compilers across different topologies and quantum ISAs.}
%     \label{tab:count_routing_overhead_avg}
%     % \setlength{\tabcolsep}{1.8pt}
%     % \fontsize{4.6}{5.0}\selectfont
%     \begin{small}
%         \input{tables/routing_overhead_count_avg.tex}
%     \end{small}
% \end{table}

% \begin{table}[tbp]
%     \centering
%     \caption{Routing overhead in terms of $ \depthCost $ for different compilers across different topologies and quantum ISAs.}
%     \label{tab:depth_routing_overhead_avg}
%     % \setlength{\tabcolsep}{1.8pt}
%     % \fontsize{4.6}{5.0}\selectfont
%     \begin{small}
%         \input{tables/routing_overhead_depth_avg.tex}
%     \end{small}
% \end{table}

\subsection{Experimental settings}

\subsubsection{ISAs and basis gate costs}
We consider six different ISAs (including the conventional \CXISA\ ISA) listed in \Cref{tab:isa_setting}. These cover a wide range of basis gates from individual $ \CX $-family or $ \iSWAP $-family gates to combinatorial ones. Particularly, \SQiSWISA~\cite{huang2023quantum} proves to be a powerful ISA option and has been adopted by recent software projects~\cite{mckinney2024mirage,cirqSQiSWDecomposer}. \ZZPhaseISA\ ISA containing three fractional $ \ZZ(\theta) $ rotation gates (equivalently, $ \bigl\{\sqrt[3]{\CX}, \sqrt{\CX}, \CX \bigr\} $) is adopted by \qiskit's latest synthesis functionalities~\cite{peterson2022optimal,qiskitXXDecomposer}. For \ZZPhaseISA\ and \SQiSWISA, we also consider the mirror-enhanced version by incorporating the mirrored basis gates~\cite{cross2019validating,mckinney2024mirage} into the ISAs. We also include the \HetISA\ ISA that is the composition of \ZZPhaseISA\ and \SQiSWISA.
Their synthesis capabilities are visualized as coverage sets within Weyl chamber, respectively, as demonstrated in \Cref{fig:coverage_cx,fig:coverage_sqisw,fig:coverage_sqisw_with_mirror,fig:coverage_zzphase,fig:coverage_zzphase_with_mirror,fig:coverage_het} in Appendix.

To conduct a coherent cross-ISA performance comparison, we use a consistent basis gate cost setting:
\begin{align}
    \left\{
    \begin{array}{c}
        \CX: 1,\, \ZZ(\frac{\pi}{t}): \frac{2}{t},\, \SQiSW: 0.75,\\
        \iSWAP: 1.5,\, \ECP: 1.25,\, \pSWAP(\frac{\pi}{t}): 2-\frac{1}{t}
    \end{array}
    \right\},\label{eq:cost}
\end{align}
where $ \CX $ gate is the unit cost. Such a setting ensures the continuity of gate costs along the critical edges in the Weyl chamber. For example, $\pSWAP(\pi/2)$ is equivalent to $\iSWAP$ and they have the same cost of $ 1.5 $. With a specific gate family, basis gates with larger canonical coefficients usually requires proportionally longer interaction time on physical devices, which was reflected in the cost setting. Note that this setting is a comprehensive consideration for current gate schemes and hardware-implemented gate fidelities in superconducting~\cite{chen2025efficient,arute2019quantum,wei2024native,nguyen2024programmable,acharya2024quantum}. It is neither limited to a specific gate scheme nor a specific hardware platform.

% For example, $ \SQiSW $ can be implemented with $ 2\sqrt{2} $x faster than $ \CZ $ on flux-tunable transmons~\cite{krantz2019quantum}, the current hardware data does not report a fidelity improvement in such a high level~\cite{huang2023quantum,arute2019quantum,chen2025efficient}.

\subsubsection{Metrics}
With the consistent basis gate cost settings above, we can evaluate cross-ISA circuit cost comparison, in terms of both gate count ($ C_{\mathrm{count}} $) and circuit depth ($ C_{\mathrm{depth}} $). Specifically, $ C_{\mathrm{count}} $ refers to the sum of all 2Q gate costs according to the basis gate setting in \Cref{eq:cost}. $ C_{\mathrm{depth}} $ refers to the length of the cost-weighted critical path within the circuit DAG. $ C_{\mathrm{count}} $ and $ C_{\mathrm{depth}} $ are naturally the generalized metrics for 2Q gate count and circuit depth. To quantify the routing effects across ISAs and topologies, we define the routing overhead as the ratio of routed circuit cost to the pre-routed circuit cost, for which the pre-routed logical-level circuit cost is uniformly computed in the \CXISA\ ISA.

\subsubsection{Benchmarks}
We select medium-size benchmarks from QASMBench~\cite{li2023qasmbench} and MQTBench~\cite{quetschlich2023mqt} spanning various categories of quantum programs. These benchmarks first go through logical-level optimization by \tket\ and are rebased to $ \left\{ \Can,\, \Uthree \right\} $ as the input of the evaluated compilers, with their detailed characteristics summarized in \Cref{tab:benchmark}.
% where $ \countCost $ and $ \depthCost $ denote costs of the total gate count and circuit duration, respectively, assuming each canonical gate will be finally rebased to $ \CX $ ISA and the cost (duration) of each $ \CX $ is set to $ 1 $.

\subsubsection{Baselines}
The leading methods \sabre, \toqm, and \bqskit\ are selected as our baselines, as they represent the most practical, scalable qubit routing approaches currently available. We implement \sabre\ and \canopus\ in the Python-based \qiskit\ framework, that is, we do not use the Rust-accelerated \sabre\ in the latest \qiskit\ version, for fair runtime comparison. \toqm\ is the SOTA circuit depth driven qubit routing method~\cite{zhang2021time}. We also select \bqskit\ as a baseline as it represents another different cross-ISA compilation paradigm~\cite{bqskit}. Given a target gate set and coupling graph, \bqskit\ performs end-to-end compilation via numerical optimization, that is, finally the rebased circuit is generated.

% by \bqskit~\cite{kalloor2024quantum}. 

%  (\code{max\_iterations} is 5, both \code{trials} and \code{layout\_trials} are 10). 

Hyperparameters for \sabre\ and \canopus\ are of the same settings. Each performs 10 times layout procedure, within which 8-round bidirectional passes are proceeded and each pass performs 10 trials. The best result across all attempts is selected. \toqm\ can obtain the deterministic routing result in one go. Compiled circuits by \bqskit, although in terms of only the 2Q gate arrangement, is also random. Thus we perform 3 trials for each input case and report the best result. 

% Specifically, each algorithm first explores 10 different initial layouts (layout_trials=10). For the best layout found, the main routing procedure is executed 10 independent times (trials=10), with each trial consisting of up to 5 internal optimization iterations (max_iterations=5). The best result across all trials is selected as the final output.

\subsection{Suppression of routing overhead}

\begin{table}[tbp]
    \centering
    \caption{Average (geometric-mean) routing overhead.}% across different topologies and backend ISAs.}
    \label{tab:routing_overhead_avg}
    \setlength{\tabcolsep}{2pt}
    \begin{footnotesize}
        \begin{tabular}{|l|l||r|r|r|r||r|r|r|r|}
\hline
\multicolumn{2}{|c||}{Routing overhead} & \multicolumn{4}{c||}{In terms of $ C_{\mathrm{count}} $} & \multicolumn{4}{c|}{In terms of $ C_{\mathrm{depth}} $} \\
\hline
{Topo} & {ISA Type} & \emph{sabre} & \emph{toqm} & \emph{bqskit} & \emph{canop} & \emph{sabre} & \emph{toqm} & \emph{bqskit} & \emph{canop} \\
\hline
\multirow{6}{*}{{Chain}} 
& {\scriptsize\CXISA} & 2.26\cellcolor{SkyBlue!30} & 3.07\cellcolor{SkyBlue!60} & 2.27\cellcolor{SkyBlue!45} & 1.88\cellcolor{SkyBlue!15} & 2.57\cellcolor{SkyBlue!60} & 2.38\cellcolor{SkyBlue!45} & 2.18\cellcolor{SkyBlue!30} & 1.81\cellcolor{SkyBlue!15} \\
& {\scriptsize\ZZPhaseISA} & 1.97\cellcolor{SkyBlue!45} & 2.75\cellcolor{SkyBlue!60} & 1.92\cellcolor{SkyBlue!30} & 1.7\cellcolor{SkyBlue!15} & 2.22\cellcolor{SkyBlue!60} & 2.15\cellcolor{SkyBlue!45} & 1.91\cellcolor{SkyBlue!30} & 1.63\cellcolor{SkyBlue!15} \\
& {\scriptsize\SQiSWISA} & 2.06\cellcolor{SkyBlue!45} & 2.63\cellcolor{SkyBlue!60} & 1.85\cellcolor{SkyBlue!30} & 1.73\cellcolor{SkyBlue!15} & 2.32\cellcolor{SkyBlue!60} & 2.08\cellcolor{SkyBlue!45} & 1.84\cellcolor{SkyBlue!30} & 1.68\cellcolor{SkyBlue!15} \\
& {\scriptsize\ZZPhaseWithMirrorISA} & 1.61\cellcolor{SkyBlue!30} & 2.18\cellcolor{SkyBlue!60} & 1.69\cellcolor{SkyBlue!45} & 1.39\cellcolor{SkyBlue!15} & 1.82\cellcolor{SkyBlue!60} & 1.72\cellcolor{SkyBlue!45} & 1.66\cellcolor{SkyBlue!30} & 1.35\cellcolor{SkyBlue!15} \\
& {\scriptsize\SQiSWWithMirrorISA} & 1.72\cellcolor{SkyBlue!45} & 2.25\cellcolor{SkyBlue!60} & 1.68\cellcolor{SkyBlue!30} & 1.45\cellcolor{SkyBlue!15} & 1.95\cellcolor{SkyBlue!60} & 1.76\cellcolor{SkyBlue!45} & 1.66\cellcolor{SkyBlue!30} & 1.4\cellcolor{SkyBlue!15} \\
& {\scriptsize\HetISA} & 1.65\cellcolor{SkyBlue!45} & 2.23\cellcolor{SkyBlue!60} & 1.58\cellcolor{SkyBlue!30} & 1.43\cellcolor{SkyBlue!15} & 1.86\cellcolor{SkyBlue!60} & 1.76\cellcolor{SkyBlue!45} & 1.56\cellcolor{SkyBlue!30} & 1.36\cellcolor{SkyBlue!15} \\
\hline
\multirow{6}{*}{{HHex}} 
& {\scriptsize\CXISA} & 2.37\cellcolor{SkyBlue!30} & 2.82\cellcolor{SkyBlue!60} & 2.59\cellcolor{SkyBlue!45} & 1.93\cellcolor{SkyBlue!15} & 3.05\cellcolor{SkyBlue!60} & 2.68\cellcolor{SkyBlue!45} & 2.66\cellcolor{SkyBlue!30} & 2.08\cellcolor{SkyBlue!15} \\
& {\scriptsize\ZZPhaseISA} & 2.12\cellcolor{SkyBlue!30} & 2.65\cellcolor{SkyBlue!60} & 2.25\cellcolor{SkyBlue!45} & 1.74\cellcolor{SkyBlue!15} & 2.77\cellcolor{SkyBlue!60} & 2.52\cellcolor{SkyBlue!45} & 2.26\cellcolor{SkyBlue!30} & 1.91\cellcolor{SkyBlue!15} \\
& {\scriptsize\SQiSWISA} & 2.14\cellcolor{SkyBlue!30} & 2.48\cellcolor{SkyBlue!60} & 2.17\cellcolor{SkyBlue!45} & 1.72\cellcolor{SkyBlue!15} & 2.71\cellcolor{SkyBlue!60} & 2.43\cellcolor{SkyBlue!45} & 2.28\cellcolor{SkyBlue!30} & 1.96\cellcolor{SkyBlue!15} \\
& {\scriptsize\ZZPhaseWithMirrorISA} & 1.7\cellcolor{SkyBlue!30} & 2.08\cellcolor{SkyBlue!60} & 1.88\cellcolor{SkyBlue!45} & 1.4\cellcolor{SkyBlue!15} & 2.2\cellcolor{SkyBlue!60} & 2.0\cellcolor{SkyBlue!45} & 1.96\cellcolor{SkyBlue!30} & 1.56\cellcolor{SkyBlue!15} \\
& {\scriptsize\SQiSWWithMirrorISA} & 1.78\cellcolor{SkyBlue!30} & 2.09\cellcolor{SkyBlue!60} & 1.98\cellcolor{SkyBlue!45} & 1.46\cellcolor{SkyBlue!15} & 2.27\cellcolor{SkyBlue!60} & 2.02\cellcolor{SkyBlue!30} & 2.1\cellcolor{SkyBlue!45} & 1.66\cellcolor{SkyBlue!15} \\
& {\scriptsize\HetISA} & 1.74\cellcolor{SkyBlue!30} & 2.13\cellcolor{SkyBlue!60} & 1.86\cellcolor{SkyBlue!45} & 1.43\cellcolor{SkyBlue!15} & 2.25\cellcolor{SkyBlue!60} & 2.05\cellcolor{SkyBlue!45} & 1.98\cellcolor{SkyBlue!30} & 1.58\cellcolor{SkyBlue!15} \\
\hline
\multirow{6}{*}{{Square}} 
& {\scriptsize\CXISA} & 1.64\cellcolor{SkyBlue!30} & 2.18\cellcolor{SkyBlue!60} & 2.06\cellcolor{SkyBlue!45} & 1.38\cellcolor{SkyBlue!15} & 1.94\cellcolor{SkyBlue!45} & 1.87\cellcolor{SkyBlue!30} & 2.47\cellcolor{SkyBlue!60} & 1.49\cellcolor{SkyBlue!15} \\
& {\scriptsize\ZZPhaseISA} & 1.35\cellcolor{SkyBlue!30} & 1.87\cellcolor{SkyBlue!60} & 1.61\cellcolor{SkyBlue!45} & 1.16\cellcolor{SkyBlue!15} & 1.63\cellcolor{SkyBlue!45} & 1.61\cellcolor{SkyBlue!30} & 1.94\cellcolor{SkyBlue!60} & 1.24\cellcolor{SkyBlue!15} \\
& {\scriptsize\SQiSWISA} & 1.63\cellcolor{SkyBlue!30} & 2.05\cellcolor{SkyBlue!60} & 1.74\cellcolor{SkyBlue!45} & 1.34\cellcolor{SkyBlue!15} & 1.89\cellcolor{SkyBlue!45} & 1.81\cellcolor{SkyBlue!30} & 2.02\cellcolor{SkyBlue!60} & 1.42\cellcolor{SkyBlue!15} \\
& {\scriptsize\ZZPhaseWithMirrorISA} & 1.16\cellcolor{SkyBlue!30} & 1.55\cellcolor{SkyBlue!60} & 1.43\cellcolor{SkyBlue!45} & 0.99\cellcolor{SkyBlue!15} & 1.39\cellcolor{SkyBlue!45} & 1.36\cellcolor{SkyBlue!30} & 1.65\cellcolor{SkyBlue!60} & 1.09\cellcolor{SkyBlue!15} \\
& {\scriptsize\SQiSWWithMirrorISA} & 1.31\cellcolor{SkyBlue!30} & 1.69\cellcolor{SkyBlue!60} & 1.56\cellcolor{SkyBlue!45} & 1.11\cellcolor{SkyBlue!15} & 1.54\cellcolor{SkyBlue!45} & 1.47\cellcolor{SkyBlue!30} & 1.83\cellcolor{SkyBlue!60} & 1.2\cellcolor{SkyBlue!15} \\
& {\scriptsize\HetISA} & 1.18\cellcolor{SkyBlue!30} & 1.58\cellcolor{SkyBlue!60} & 1.36\cellcolor{SkyBlue!45} & 1.0\cellcolor{SkyBlue!15} & 1.41\cellcolor{SkyBlue!45} & 1.38\cellcolor{SkyBlue!30} & 1.56\cellcolor{SkyBlue!60} & 1.09\cellcolor{SkyBlue!15} \\
\hline
\end{tabular}

    \end{footnotesize}
\end{table}

\Cref{tab:routing_overhead_avg} lists the geometric-mean routing overhead across all 216 cases (3 topologies $\times$ 6 ISAs $\times$ 12 programs) for each compiler, with per-benchmark details shown in \Cref{fig:routing_overhead}. \canopus\ achieves the lowest routing overhead for every ISA-topology combination. Specifically, \canopus\ reduces average routing overhead by 16.06\% in $ C_{\mathrm{count}} $ and 26.44\% in $ C_{\mathrm{depth}} $ compared to \sabre, by 34.70\% and 21.25\% compared to \toqm, and by 19.89\% and 20.72\% compared to \bqskit.

\note{Notably, \canopus\ uniquely leverages the synthesis capabilities of more expressive ISAs. With \canopus, transitioning from \CXISA\ to more powerful ISAs yields substantial routing overhead reductions---e.g., from $1.88\times$ to $1.39\times$ ($-26\%$) on 1D chain, and from $1.38\times$ to $0.99\times$ ($-28\%$) on 2D square for $ C_{\mathrm{count}} $ when equipped with \ZZPhaseWithMirrorISA\ ISA---while baseline methods exhibit much less pronounced improvements. This confirms that \canopus\ does not merely benefit from ISA rebase but actively exploits ISA expressiveness during routing.

The advantage of \canopus\ also lies in its unified optimization of both gate count and circuit depth. In contrast, \sabre\ and \bqskit\ are primarily gate-count-driven, while \toqm\ specializes in optimizing depth. This bias manifests in measurable weaknesses. \toqm\ incurs the worst count overhead across nearly all configurations. For instance, on 1D chain with \CXISA, \toqm's reaches $3.07\times$ routing in terms of $ C_{\mathrm{count}} $, more than 63\% above that of \canopus\ ($1.88\times$). Conversely, \bqskit\ suffers severe depth overhead on 2D square topology, where its $ C_{\mathrm{depth}} $-related routing overhead consistently exceeds those of all other compilers (e.g., $2.47\times$ for \CXISA\ versus $1.49\times$ for \canopus).}

Additionally, \canopus\ maintains consistently low overhead across all benchmarks, whereas every baseline fails on specific circuits. For instance, \toqm\ and \bqskit\ cannot effectively manage the routing overhead for some structurally challenging circuits like \code{qec9} and \code{qram}; \bqskit\ struggles with \code{bv} even under expressive ISAs.

% Notably, \sabre\ and \bqskit\ focus on gate count driven optimization, \toqm\ specializes in lowering circuit depth, while our \canopus\ involves both count and depth related optimization. Among them, \toqm\ leads to the worst count-related routing overhead. Even for the routing overhead in terms of depth, \toqm\ consistently underperforms \canopus, and it only outperforms \bqskit\ for 2D square topology. In addition, \canopus\ maintains routing overhead---in terms of both count and depth costs---at a consistently low level; while every baseline fails to optimize some specific circuits. For example, \toqm/\bqskit\ cannot manage the routing overhead for \code{qec9} circuit; \bqskit\ cannot for \code{bv} circuit, even though with more expressive ISAs.
% with more connected topologies and more expressive ISAs. 

% Count-cost reduction v.s. Sabre: 16.06%
% Count-cost reduction v.s. Toqm: 34.70%
% Count-cost reduction v.s. Bqskit: 19.89%

% Depth-cost reduction v.s. Sabre: 26.44%
% Depth-cost reduction v.s. Toqm: 21.25%
% Depth-cost reduction v.s. Bqskit: 20.72%

\subsection{Program-ISA-Topology co-exploration}

\begin{figure*}[tbp]
    \centering
    \subfigure[Routing overhead in terms of $ C_{\mathrm{count}} $.]{\includegraphics[width=\linewidth]{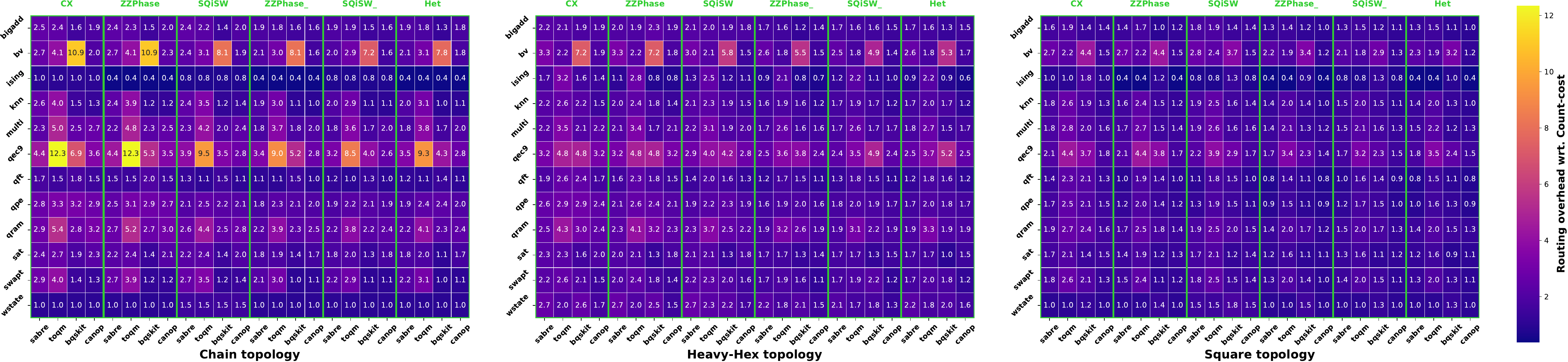}\label{fig:routing_overhead_count}}
    \subfigure[Routing overhead in terms of $ C_{\mathrm{depth}} $.]{\includegraphics[width=\linewidth]{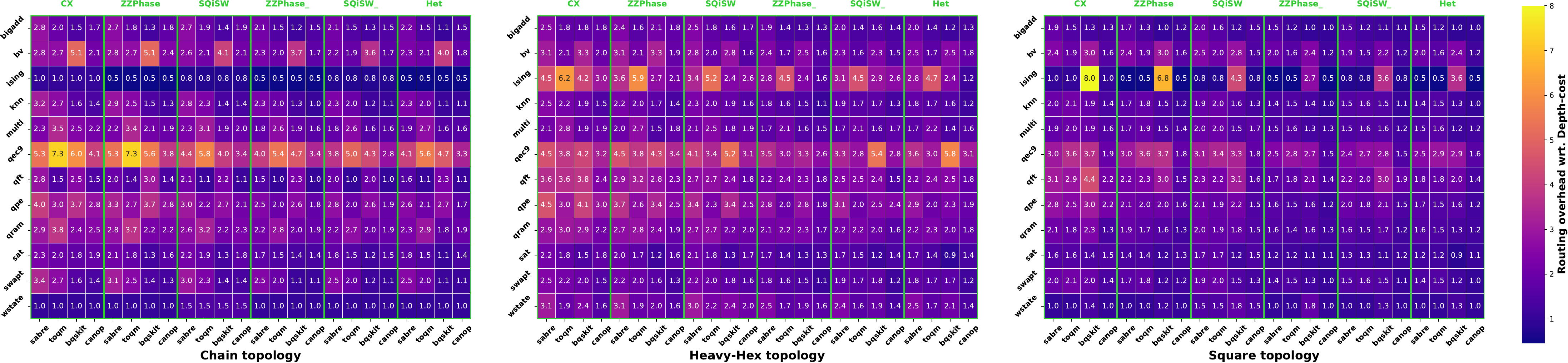}\label{fig:routing_overhead_depth}}
    \caption{Routing overhead in terms of (a) $ C_{\mathrm{count}} $ and (b) $C_{\mathrm{depth}}$ for different compilers across various device topologies and quantum ISAs.}
    \label{fig:routing_overhead}
\end{figure*}

Our evaluation also systematically explores how program patterns, ISA selection, and hardware topologies impact each other. We highlight some co-design guidelines particularly according to results achieved by \canopus\ (\Cref{tab:routing_overhead_avg}, \Cref{fig:routing_overhead}): 
\begin{itemize}[leftmargin=*, topsep=2pt, itemsep=2pt, parsep=2pt]
    \item \note{\emph{Topology-program affinity matters more than raw connectivity:}} Heavy-hex topology consistently incurs higher routing overhead across all ISAs, despite having higher average connectivity. This is because most quantum algorithms are constructed in a subroutine-unrolling approach, naturally more friendly to chain topology. The QFT kernel detailed in \Cref{sec:qft_study} is a thorough good example.
    \item \note{\emph{Heterogeneous ISAs yield disproportionate gains:}} Combining $ \CX $-family and $\iSWAP$-family gates into \HetISA\ provides substantially greater routing overhead reduction than either family alone. \note{On 1D chain under \canopus, \ZZPhaseISA\ reduces count overhead by 9.6\% and \SQiSWISA\ by 7.9\% relative to \CXISA, while \HetISA\ achieves a 23.9\% reduction. The same amplified effect holds across other topologies, indicating that the two gate families address complementary routing scenarios, enabling \canopus\ to select the most efficient decomposition in each $ \SWAP $ insertion context.} This benefit is more pronounced for circuits largely containing $ \CX/\CZ $ as 2Q blocks, such as \code{qec9}.
    % The combinatorial ISA involving both $ \CX $-family and $\iSWAP$-family gates is much superior to either sole one. Specifically, either \ZZPhaseISA\ or \SQiSWISA\ could lead to no more than 11\% routing overhead reduction compared to \CXISA, while \HetISA\ could result in more than 25\% reduction than \CXISA. This benefit is more significant for circuits largely containing $ \CX/\CZ $ as 2Q blocks, such as \code{qec9}.
    \item \note{\emph{Gate mirroring is another approach to designing powerful quantum ISAs:}} Both \ZZPhaseISA\ and \SQiSWISA\ achieve comparable results to \HetISA, since mirror gates naturally enable low-overhead $ \SWAP $ absorption, that is, $ \SWAP $ mirroring.
    \item \note{\emph{ISA selection should be program-aware:}} For Hamiltonian simulation programs like \code{ising}, \ZZPhaseISA\ ISA is essential to improve execution performance. Therein multiple Ising gates (i.e., 2-local Pauli rotations equivalent to $ \XX(\theta) $) are included. As a discrete fractional $ \XX(\theta) $ basis gate set, \ZZPhaseISA\ ISA inherently aligns better with these workloads than other gate families, significantly boosting execution performance. \note{Besides, the commutation patterns (the fourth pattern in \Cref{fig:commutation_b}) occurring in \code{ising} can be effectively identified in the canonical form and the commutativity-optimization mechanism plays a critical role in routing (see \Cref{fig:commutation_stats}, \Cref{tab:ablation} and \Cref{sec:ablation_study} for further discussion).} \note{While circuits dominated by $ \CX/\CZ $ blocks (e.g., \code{qec9}) benefit more from heterogeneous ISAs in which both $ \CX/\CZ $ and $ \iSWAP $ gates are included.}
    % \item \note{\textbf{Near-zero routing overhead is achievable.}} On 2D square topology, \canopus\ with \ZZPhaseWithMirrorISA\ achieves a count overhead of $0.99\times$---meaning the routed circuit contains \emph{fewer} two-qubit gates than the logical circuit. With \HetISA, the count overhead is $1.00\times$ and depth overhead is $1.09\times$. These results demonstrate that when topology connectivity is sufficient and the ISA is expressive enough, \canopus\ can fully compensate for routing cost through synthesis-aware gate absorption.
\end{itemize}

The real-machine experiment in \Cref{sec:qft_study} showcases how our method can help achieve superior compilation results and thus higher program fidelities for QFT kernels using the \CXISA\ and \ZZPhaseISA\ ISAs via IBM Quantum Cloud. However, there are current practical hurdles to extending this real-machine validation to alternative ISAs—ones that arise primarily from the continued scarcity of quantum processors with well-calibrated heterogeneous gate sets. 
% For instance, while IBM has proposed fractional gates, that is, the continuous $ \ZZ(\theta) $ gate set~\cite{ibmFractionalGates}, their implementation details and calibration procedures are not publicly disclosed. To our knowledge, the $ \ZZ(\theta) $ gates have the same duration as $ \CZ $ on IBM's Heron QPUs regardless of the rotation angle, and their error rates are consistently 1x-3x that of $ \CZ $. This performance is far from the ideal assumptions of \ZZPhaseISA.
Fortunately, a path forward is emerging with the recently proposed AshN gate scheme~\cite{chen2024one} and its extended generalization~\cite{yang2026reconfigurable} that enable directly implementing any basis gates with the optimal gate durations. It is also experimentally demonstrated on transmon qubits by Chen et al.~\cite{chen2025efficient}, where multiple basis gates are calibrated with high fidelity, which aligns with our cost model as well. This development may enable comprehensive, real-machine program-ISA-topology co-exploration in the near future.

\subsection{Diverse-ISA compilation paradigms}

Prior to this work, there are two major compilation paradigms targeting diverse ISAs: (1) Use the conventional compiler that operates entirely on the $ \CX $-based circuit representation before ISA rebase. The final-stage rebase pass can usually be completed via optimal synthesis in efficient analytical or numerical computation~\cite{tucci2005introduction,peterson2022optimal,huang2023quantum,mckinney2025two}. (2) Use brute-force approximate synthesis to perform structural search and numerical optimization to determine the synthesized circuit with minimal gate count~\cite{davis2019heuristics,kukliansky2023qfactor,patel2022quest}. \sabre/\toqm\ and \bqskit\ are representative of these two paradigms, respectively. In our evaluation, \bqskit\ even underperforms the industrial-standard \sabre\ in most cases. As an exception, in terms of the circuit depth, \bqskit\ leads to better results than other baselines on sparse topologies (chain, heavy-hex), as its A*-based search for 2Q gate arrangement could exhibit advantages over long-range qubit routing, but this advantage does not hold for more connected topologies. Besides, the second numerical optimization based paradigm is of exponential computational complexity. For benchmarking the 216 medium-size cases, the Rust-backend \bqskit\ requires on average 18 minutes to process each circuit with an Apple M3 Max CPU; in contrast, the Python-implemented \sabre\ requires only 17 seconds. Consequently, this second paradigm is ill-suited for compiling real-world programs, proving both ineffective and inefficient when targeting diverse ISAs (at least for discrete gate sets). Instead, although there is a gap between the conventional routing model and backend ISA properties, by means of the routing-synthesis co-optimization mechanism of \canopus, the first paradigm is enhanced to bridge the gap between the routing model and backend ISA properties and thus provides a more viable path.

% BQSKit: 1 hour for 72*3=216 cases --- 平均每个电路 18 min, Rust backend implementation
% Sabre: 17 sec per case

\iffalse

\subsection{Breakdown analysis}

In this section, we analyze individual factors in the improvement brought by \canopus, mainly about the commutative optimization mechanism. % and the heuristic depth-cost weight factor.

\begin{table}[tbp]
    \centering
    \caption{Routing overhead improvement analysis for \canopus\ relative to the routing process without commutative optimization (\code{no\_comm}). \emph{Avg.} in the table indicates the relative reduction of geometric-mean $ \countCost $ or $ \depthCost $ across all benchmarks; \emph{Max.} indicates the maximum reduction achieved on one of benchmarks.}
    \setlength{\tabcolsep}{3pt}
    \label{tab:breakdown}
    \begin{footnotesize}
        \input{tables/breakdown_table.tex}        
    \end{footnotesize}
    
\end{table}

Note that the ...

\fi

\subsection{Runtime analysis}

% \begin{figure}[tbp]
%     \centering
%     \includegraphics[width=\columnwidth]{figures/runtime.pdf}
%     \caption{Compilation latency comparison. \ZY{Redraw the figure}}
%     \label{fig:runtime}
% \end{figure}

\begin{figure}
    \centering
    % width=0.475\columnwidth
    \includegraphics[width=0.43\columnwidth]{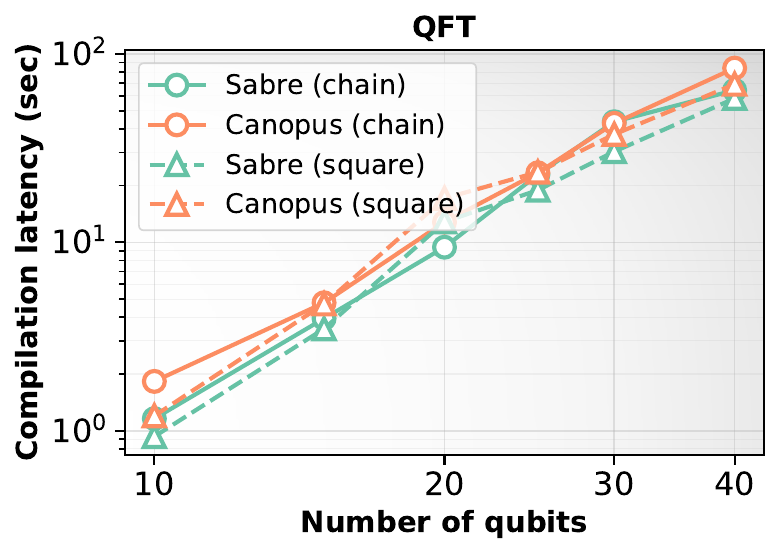}\hfill
    \includegraphics[width=0.43\columnwidth]{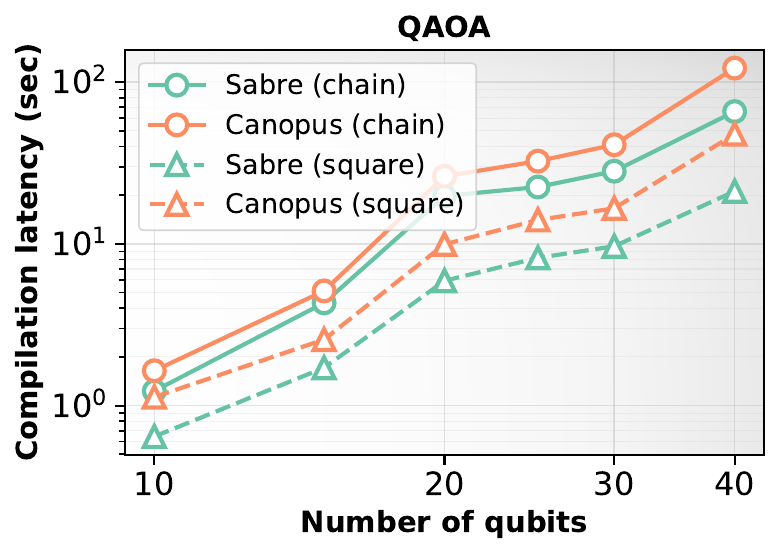}\\\quad

    \includegraphics[width=0.43\columnwidth]{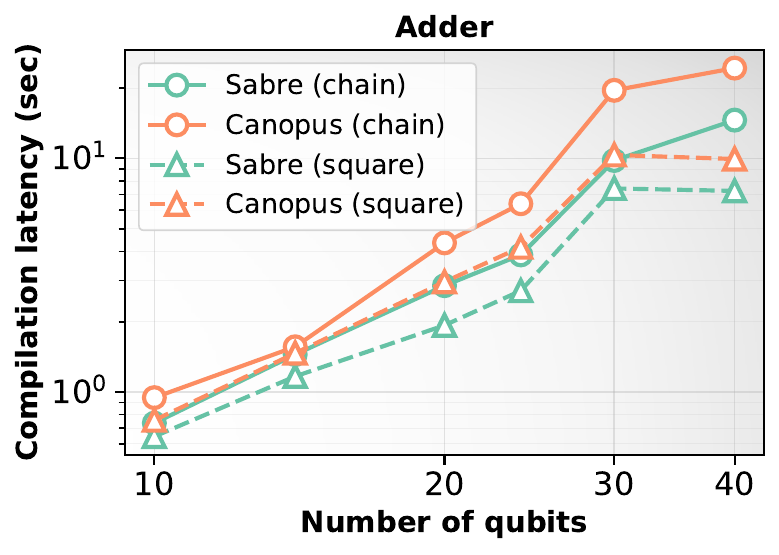}\hfill
    \includegraphics[width=0.43\columnwidth]{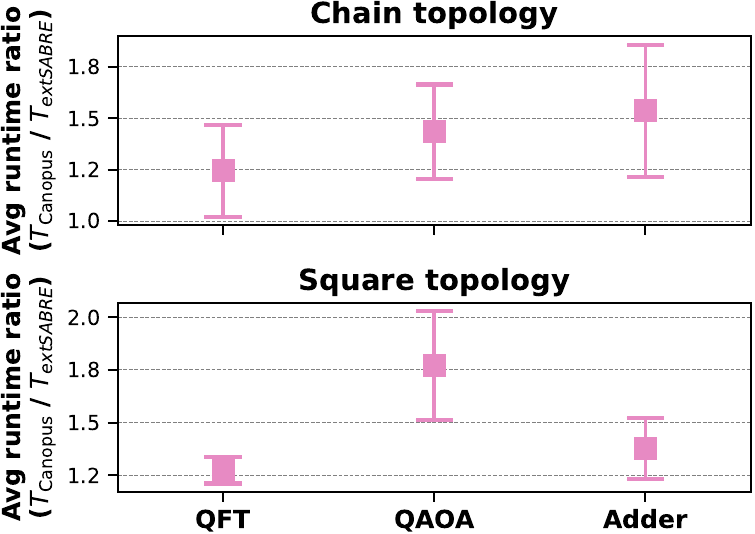}
    \caption{\note{Compilation latency comparison.}}
    \label{fig:runtime}
\end{figure}

\note{
  In our field tests for the 216 cases above, \canopus\ consistently exhibits 1--2$\times$ runtime latency of \sabre. To further evaluate runtime scalability, we benchmark \canopus\ against \sabre\ on three representative quantum algorithms including QFT~\cite{shor1994algorithms}, QAOA (MaxCut on random 3-regular graphs)~\cite{farhi2014quantum}, and CDKM ripple-carry adder~\cite{cuccaro2004new}, across 1D chain and 2D square topologies, with qubit counts ranging from 10 to 40 or 50. As shown in \Cref{fig:runtime}, both methods exhibit polynomial scaling (linear trends on log-log axes), confirming that \canopus\ preserves the asymptotic complexity of the underlying \sabre\ routing procedure. The near-constant-factor overhead arises from (1) evaluating the heuristic cost for each $ \SWAP $ candidate, including computing the cost component $ c_g $ and depth overhead $ \Delta_{\mathrm{depth}} $ and (2) updating the state tracking variables ($ L $, $ D $, and $ C $ in \Cref{alg:update_L_D_C} and \Cref{alg:update_durations_swap}) after inserting the best $ \SWAP $ or mapping an executable gate. All these operations do not involve matrix-level numerical computations and full data structure rebuilds, thus \canopus\ does not change the asymptotic scaling.
  Across all sizes of benchmarks, the average runtime ratio $T_{\canopus}/T_{\sabre}$ remains within 1--2$\times$ (\Cref{fig:runtime}), even for circuits at scale.
%   , with lower ratios observed on the square topology (2.3--3.4$\times$).
%    where higher connectivity requires fewer $\SWAP$ candidate evaluations, and moderately higher on the sparser chain (2.6--4.8$\times$).
   Overall, despite its sophisticated data structures and computation mechanisms, \canopus\ achieves practical compilation scalability comparable to the industry-standard \sabre\ algorithm.

}

\note{

\subsection{Ablation study}\label{sec:ablation_study}

To isolate the contribution of the canonical commutativity optimization (\Cref{sec:gate_commutation_guided_optimization}), we run \canopus\ with and without it across all 216 test cases. A prerequisite for this optimization is the availability of commutative 2Q gate pairs in the circuit DAG. As shown in \Cref{fig:commutation_stats}, canonical-basis circuits exhibit significantly higher commutative pair ratio (always near 100\%) among all successive 2Q gates, meaning that almost all consecutive 2Q canonical gates in these real-world applications are commutative, whose commutation patterns are partly illustrated by \Cref{fig:commutation_b}. This substantially larger pool of reorderable gates translates into concrete routing improvements. 
Table~\ref{tab:ablation} reports the gate count and circuit depth reductions achieved by enabling commutativity optimization. 
Across all ISA-topology combinations in our field tests, enabling commutativity yields an average gate count reduction of 2--11\% and depth reduction of 2--10\%, with peak improvements reaching 37--48\% on individual benchmarks. The gains are most pronounced on the chain topology under \CXISA\ ISA,
% (10.56\% average count reduction, 9.15\% average depth reduction),
where the limited connectivity forces more $ \SWAP $ insertions and thus creates more opportunities for commutation-based reordering to find lower-cost $ \SWAP $ insertions. 

Specifically, circuits with dense, non-local interaction patterns benefit most: \code{knn} and \code{swap\_test} on chain achieve 31--37\% count reduction, as their high density of overlapping two-qubit gates provides abundant commutation opportunities for the router to exploit. Arithmetic circuits such as \code{bigadder} and \code{multiplier} also see consistent improvements across topologies, since their structured but non-trivial qubit interaction graphs produce many reorderable gate pairs. In contrast, circuits with inherently local connectivity (e.g., \code{ising}, \code{wstate}) show negligible change, as their routing overhead is already minimal and leaves little room for commutation-based improvement. On denser topologies such as heavy-hex and square, the absolute reductions are smaller but remain consistent. Notably, ISAs with mirror-enhanced gate sets such as \SQiSWWithMirrorISA\ show smaller marginal gains from commutativity, as their richer native gate repertoire already reduces the baseline routing overhead, leaving less room for further optimization through gate reordering.

\begin{table}[tbp]
    \centering
    \caption{\note{Routing overhead improvement of \canopus\ vs. routing without commutative optimization (\code{no\_comm}).}}
    \label{tab:ablation}
    %  \emph{Avg.} in the table indicates the relative reduction of geometric-mean $ \countCost $ or $ \depthCost $ across all benchmarks; \emph{Max.} indicates the maximum reduction achieved on one of benchmarks.}}
    \setlength{\tabcolsep}{3.5pt}
    \begin{scriptsize}
        \begin{tabular}{|l|r|r|r|r|r|r|}
    \hline
    \multirow{2}{*}{\parbox{1.8cm}{\centering $ \countCost $ improv. \\ vs. \code{no\_comm}}} & \multicolumn{2}{c|}{\textbf{Chain}} & \multicolumn{2}{c|}{\textbf{HHex}} & \multicolumn{2}{c|}{\textbf{Square}} \\
    \cline{2-7}
     & \emph{Avg.} & \emph{Max.} & \emph{Avg.} & \emph{Max.} & \emph{Avg.} & \emph{Max.} \\
    \hline
    \CXISA & -10.56\% & -37.57\% & -0.77\% & -12.35\% & -4.1\% & -20.59\% \\
    \hline
    \ZZPhaseISA & -4.31\% & -34.81\% & -8.44\% & -35.29\% & -2.51\% & -15.62\% \\
    \hline
    \SQiSWISA & -5.81\% & -30.97\% & -6.13\% & -42.86\% & -4.82\% & -20.0\% \\
    \hline
    \ZZPhaseWithMirrorISA & 0.04\% & -5.38\% & -5.44\% & -26.58\% & -2.56\% & -8.0\% \\
    \hline
    \SQiSWWithMirrorISA & -2.88\% & -12.12\% & -5.9\% & -27.14\% & -2.86\% & -11.86\% \\
    \hline
    \HetISA & -3.59\% & -26.67\% & -8.74\% & -47.92\% & -3.59\% & -18.52\% \\
    \hline
    \hline

    \multirow{2}{*}{\parbox{1.8cm}{\centering $ \depthCost $ improv. \\ vs. \code{no\_comm}}} & \multicolumn{2}{c|}{\textbf{Chain}} & \multicolumn{2}{c|}{\textbf{HHex}} & \multicolumn{2}{c|}{\textbf{Square}} \\
    \cline{2-7}
     & \emph{Avg.} & \emph{Max.} & \emph{Avg.} & \emph{Max.} & \emph{Avg.} & \emph{Max.} \\
    \hline
    \CXISA & -9.15\% & -38.57\% & -1.99\% & -20.0\% & -1.88\% & -10.13\% \\
    \hline
    \ZZPhaseISA & -4.76\% & -40.44\% & -10.61\% & -31.08\% & 0.16\% & -12.89\% \\
    \hline
    \SQiSWISA & -3.26\% & -31.71\% & 1.14\% & -29.63\% & -2.16\% & -13.75\% \\
    \hline
    \ZZPhaseWithMirrorISA & 0.94\% & -6.4\% & -4.04\% & -25.96\% & -2.81\% & -27.81\% \\
    \hline
    \SQiSWWithMirrorISA & -1.5\% & -11.43\% & -1.45\% & -17.91\% & -1.82\% & -7.69\% \\
    \hline
    \HetISA & -5.12\% & -32.43\% & -4.84\% & -48.65\% & -1.61\% & -14.87\% \\
    \hline
\end{tabular}
        
    \end{scriptsize}
    
\end{table}

\begin{figure}[tbp]
    \centering
    \includegraphics[width=\columnwidth]{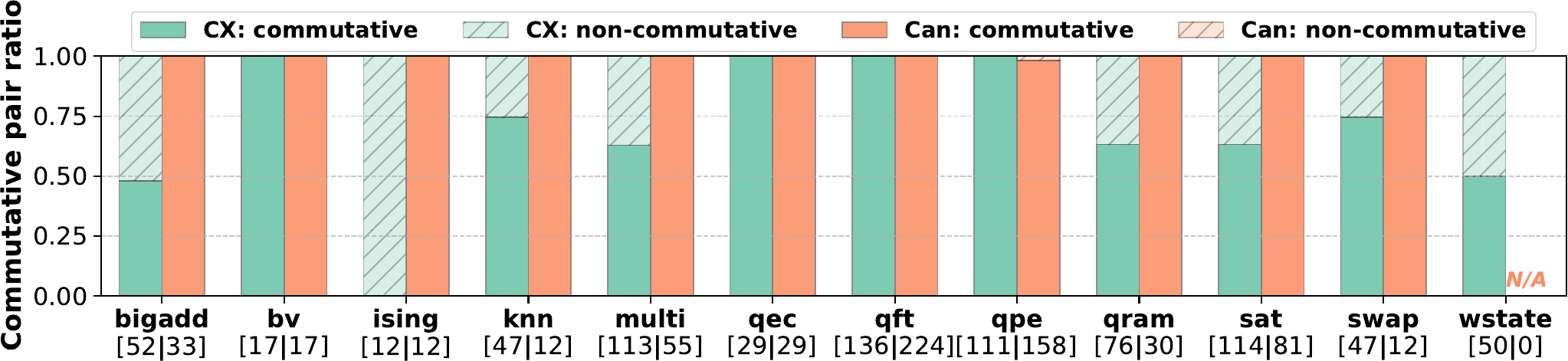}
    \caption{\note{Commutative pairs within successive 2Q Gates. X-axis ticks indicate the benchmark name and the total number of successively occurring 2Q gate pairs, i.e., $ [\text{\#}\code{CX\_pairs}\, |\, \text{\#}\code{Can\_pairs} ]$ shown below the circuit name; Y-axis indicates the ratio of commutative pairs among all successive 2Q gates.}}
    \label{fig:commutation_stats}
\end{figure}

}

\note{
\subsection{Sensitivity and trade-off analysis}\label{sec:sensitivity_analysis}

\begin{figure}[tbp]
    \centering
    \subfigure[\note{Average routing overhead wrt. varying weight on 1D chain.}]{\includegraphics[width=0.95\columnwidth]{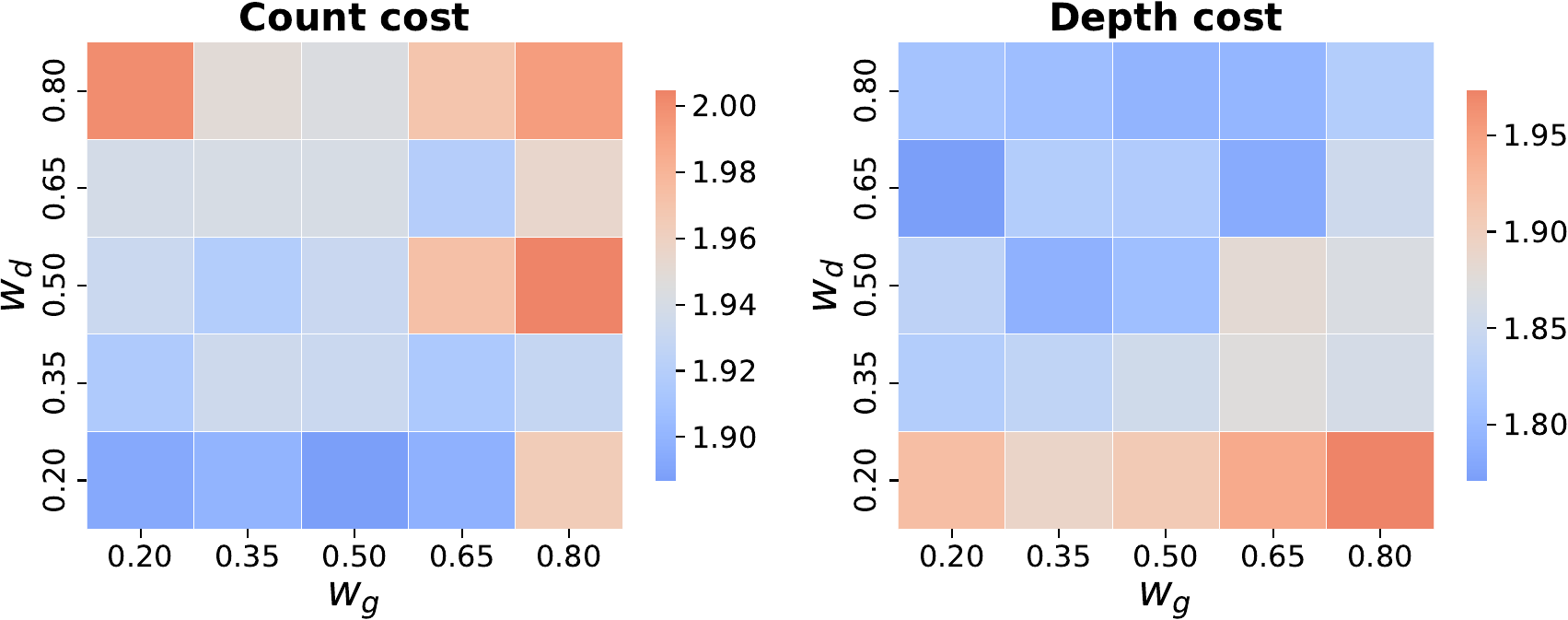}}
    \subfigure[\note{Average routing overhead wrt. varying weight on 2D square.}]{\includegraphics[width=0.95\columnwidth]{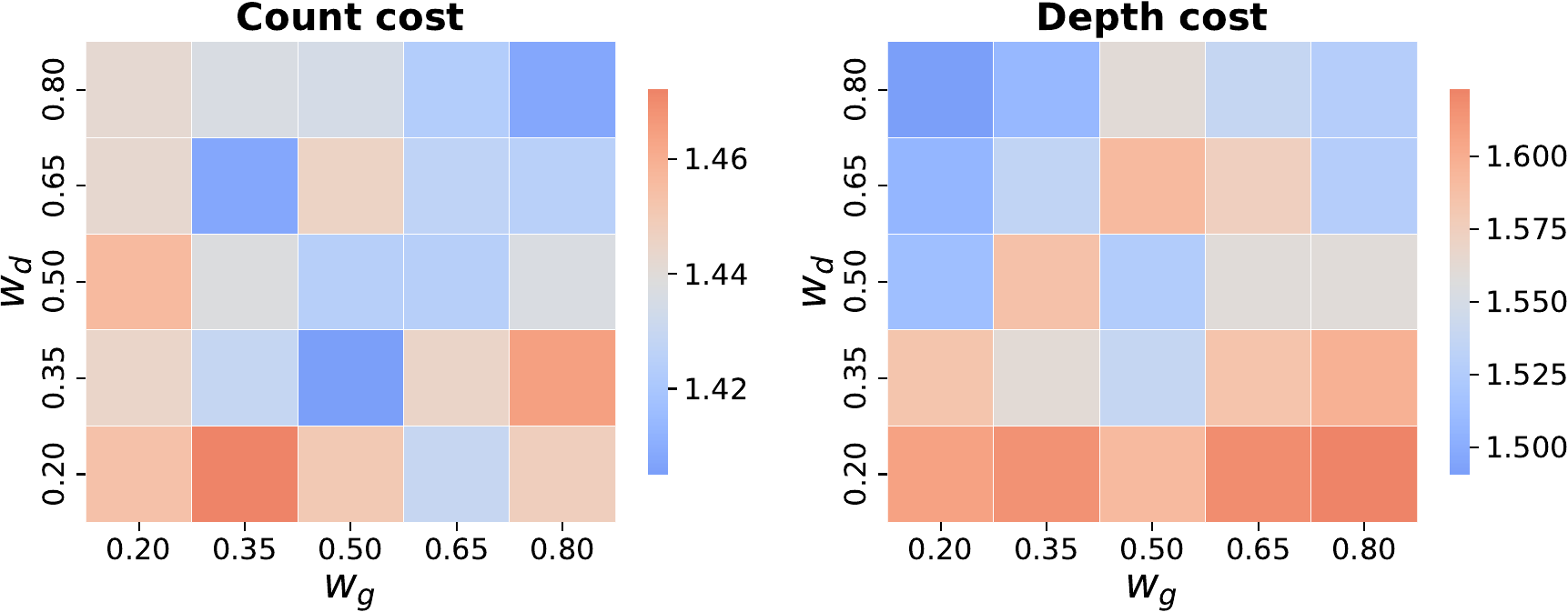}}
    \caption{\note{Sensitivity analysis for the weight factors $w_{g}$ and $w_{d}$ in the heuristic cost function. Routing overhead is geometric-mean average across all 12 benchmarks under \CXISA\ ISA on (a) 1D chain or (b) 2D square topologies.}}
    \label{fig:sensitivity}
\end{figure}

% To assess the sensitivity of \canopus\ to the weight factors $w_{g}$ and $w_{d}$ introduced in \Cref{eq:heuristic_cost}, we sweep them from $ 0.2 $ to $ 0.8 $ with step $ 0.15 $ (totally 25 parameter combinations) and evaluate the average routing overhead across all 12 benchmarks under \CXISA\ ISA on 1D chain. Each data point is obtained by setting fewer  bidirectional routing iterations to reduce the overall evaluation time. The same sensitivity analysis is also performed for 2D square topology.

To assess the sensitivity to the weight factors $w_g$ and $w_d$ (\Cref{eq:heuristic_cost}), we sweep both parameters from $0.2$ to $0.8$ in steps of $0.15$ (totally 25 combinations).
For each configuration, we evaluate the average routing overhead across all 12 benchmarks using the \CXISA\ ISA on both 1D chain and 2D square topologies.  To accelerate this extensive evaluation, all data points are generated using a reduced number of bidirectional routing iterations. \Cref{fig:sensitivity} demonstrates that the heuristic is robust to the choice of weights: across the central region ($w_g, w_d \in [0.35, 0.65]$), the average count overhead varies by less than 3.5\% on both topologies (chain: $1.92$--$1.97$; square: $1.41$--$1.45$), and the default setting ($w_{g} = w_{d} = 0.5$) lies within 2.5\% of the global optimum in all four heatmaps. The worst overhead consistently occurs at low $w_{d}$ values (bottom rows), where the synthesis-aware depth optimization is effectively disabled and the heuristic degrades toward vanilla gate count driven routing; increasing $w_d$ progressively improves depth quality, an effect especially pronounced on 2D square. The default $w_g = w_d = 0.5$ setting therefore provides a well-balanced operating point that sits close to the Pareto front between gate count and circuit depth.

}

\section{Related Work}\label{sec:related}
Qubit mapping/routing is one of the most well-explored topics of quantum compiler research~\cite{zhu2025quantum}, as it shares similar methodologies with instruction scheduling~\cite{codina2001unified,hennessy1983postpass} and register allocation~\cite{chaitin1982register,poletto1999linear} in classical computing. 

To perform scalable qubit routing, Zulehner et al.~\cite{zulehner2018efficient} introduces an A*-based algorithm to minimize $ \SWAP $ gate overhead for concurrent $\CX$ gate layers. The approach partitions the circuit into layers and solves the mapping problem subsequently. Li et al.~\cite{li2019tackling} also utilizes the circuit DAG layering thought and proposes a bidirectional routing procedure \sabre\ to find better initial mappings thus with lower $\SWAP$ insertion count.
% desired to result in \#$ \SWAP $ inserted minimization as expected.
It also briefly discusses the trade-off between the inserted $ \SWAP $ count and the circuit depth but does not prioritize optimizing circuit depth.
Subsequent works have aimed to improve circuit depth and parallelism, either by using \sabre-like heuristics~\cite{lao2021timing,annechini2025ddroute,zou2024lightsabre} or graph matching techniques~\cite{childs2019circuit}. Zhang et al.~\cite{zhang2021time} systematically investigates the depth-optimality of qubit mapping and proposes an A*-based method \toqm\ that reported superior performance over existing solver-based depth-driven approaches~\cite{tan2020optimal}. However, holistic optimality of qubit routing is contingent on the specific ISA, device topology, and circuit cost model, and is rarely guaranteed by theoretical bounds. Indeed, our evaluation reveals that \toqm\ does not always produce depth-optimal results compared to our heuristic, \canopus. For instance, the case study in \Cref{sec:qft_study} demonstrates that the mapping scheme for the QFT kernel, purported to be optimal in their analysis, can be further improved. Besides, while several studies have explored merging $ \SWAP $ gates with preceding operations and reordering commutative gates during routing to enhance performance~\cite{lao20222qan,tan2021optimal,liu2022not,mckinney2024mirage}, these approaches remain largely restricted to specific routing models, program patterns, or basis gate sets.

% \ZY{Cover more about scalable ML/RL-based, noise-aware qubit maping}

% \ZY{
%     Several studies have proposed merging SWAP gates with preceding operations or using "aggregated instructions" for specific gate sets~\cite{shi2019optimized, gokhale2019partial, lao20222qan}. Others have explored instruction set design and its impact on performance~\cite{lao2021designing, huang2023quantum, chen2024one}. In particular, reordering commuting gates has been leveraged for specific algorithms like QAOA~\cite{alam2020circuit}.

%     \canopus\ distinguishes itself from these works in three key ways. First, unlike previous methods that target specific architectures or instruction subsets (e.g., [1, 2, 7] from the reviewer's list), \canopus\ provides a \textbf{unified framework} based on the Monodromy Polytope~\cite{peterson2020fixed, peterson2022optimal} that is applicable to \textit{any} 2Q basis gate set. Second, while ISA design works~\cite{lao2021designing, huang2023quantum} discuss hardware capabilities, \canopus\ is the first to integrate this ISA-awareness directly into a \textbf{scalable routing heuristic}, bridging the gap between routing and synthesis. Third, our \textbf{generalized commutativity} theorem (Theorem 1) provides an ISA-independent condition in the Weyl chamber, extending the algorithmic-specific rules used in prior work~\cite{alam2020circuit, gokhale2019partial}.
% }

With the recent development of advanced quantum ISAs such as superconducting fractional gates~\cite{ibmFractionalGates} and fSim~\cite{foxen2020demonstrating,lao2021designing} or XY~\cite{abrams2020implementation} family gates, ion-trapped partial entangling gates~\cite{ionqPartialGates,yale2025realization}, and the AshN gates~\cite{chen2024one,chen2025efficient,yang2026reconfigurable}, some works have begun exploring how to efficiently utilize these ISAs to make compiler optimizations closer to hardware characteristics. McKinney et al.~\cite{mckinney2024mirage} investigates the practical performance of \SQiSWISA\ ISA proposed by Huang et al.~\cite{huang2023quantum} and the synthesis capability when incorporating the basis gates' mirrors into the ISA. Their modified \sabre\ algorithm offers a preliminary attempt at the collaborative gate decomposition and qubit routing approach, while the optimization opportunities considered therein are limited and the algorithmic techniques are not sophisticated. \bqskit~\cite{bqskit} and the series of works behind it~\cite{davis2019heuristics,wu2020qgo,kukliansky2023qfactor,younis2021qfast} provide a toolkit to rebase arbitrary 2Q unitaries to specific ISAs through approximate synthesis (structural search and numerical optimization) which is not computationally efficient. Approximate synthesis by \bqskit\ does not ensure optimal schemes for two-qubit and multi-qubit circuit synthesis. In addition, due to the lack of native compilation strategies and a rational synthesis cost model, Kalloor et al.~\cite{kalloor2024quantum} claims that alternative ISAs are hardly comparable to \CXISA\ when evaluating quantum hardware roofline by \bqskit. As for the applicability of expanded ISAs to QEC, Google's latest theoretical~\cite{mcewen2023relaxing} and experimental~\cite{eickbusch2024demonstrating} works demonstrate that the \CXISA-\iSWAPISA\ combination ISA could help suppress the fault-tolerant threshold. Zhou et al.~\cite{zhou2024halma} proposes a routing-based method enhanced by \CXISA-\iSWAPISA\ for overcoming ancilla defects among surface code blocks while preserving encoded logical information, but it relies on manual design and experience.

% - expanded ISAs \& systematic utilization
%     - Noise-aware ...
%     - Mirage ... not sophisticated algorithm, ... a subset of our approach
%     - BQSKit ... quantum hardware roofline .. they claim that ... 
%     - The last-step \dquote{synthesizer} .... most based on \dquote{approximate synthesis} (roofline, heterogeneous, ZZ(theta))

% (((((((None of them make deep co-optimization tailored to various quantum ISAs in a systematic and efficient approach)))))))

% With respect to the heuristic cost for $ \SWAP $ search, our routing algorithm involves the duration (generalization metric of circuit depth) driven goal by taking the canonical gate synthesis cost into the circuit duration increment. 

\section{Conclusion}\label{sec:conclusion}

In our work, we introduce \canopus, the first unified, ISA-aware qubit routing framework designed to operate across diverse quantum hardware. By leveraging the canonical two-qubit gate representation and a formal cost model derived from monodromy polytope theory, \canopus\ achieves deep co-optimization of routing and synthesis. It not only demonstrates the practical superiority of emerging quantum ISAs but also enables systematic co-exploration of how different ISAs, program patterns, and hardware topologies interact, providing a powerful new tool for quantum computing system design.

% We introduce \canopus, the first unified, ISA-aware qubit routing framework for diverse quantum hardware. By modeling gate costs with monodromy polytopes, \canopus\ deeply co-optimizes routing and synthesis. This approach demonstrates the superiority of emerging ISAs and provides a systematic tool for co-exploring the interplay between hardware topology, instruction sets, and program structure, paving the way for true quantum co-design.

% It is promising to explore novel Clifford circuit optimization techniques drawing on the canonical gate representation.

\section*{Acknowledgements}
% We thank all anonymous reviewers for their thoughtful feedbacks.
This research was partially conducted by AI Chip Center for Emerging Smart Systems (ACCESS), supported by the InnoHK initiative of the Innovation and Technology Commission of the Hong Kong Special Administrative Region Government. It was also supported partially by Research Grants Council of Hong Kong SAR (\#16213824 \& \#16212825).
This research was also funded by the Shanghai Institute of Mathematics and Interdisciplinary Sciences under grant number \textsc{simis-id-2025-qt}.
Z. Y. would like to thank Xueci Zhang for her valuable suggestions on the paper's visual presentation, especially the figure and color styles.
D. D. would like to thank God for all of His provisions.

%%%%%%% -- PAPER CONTENT ENDS -- %%%%%%%%

%%%%%%%%% -- BIB STYLE AND FILE -- %%%%%%%%
\bibliographystyle{IEEEtranS}
\balance

\bibliography{reference}

@article{chen2025efficient,
author={Chen, Zhen
and Liu, Weiyang
and Ma, Yanjun
and Sun, Weijie
and Wang, Ruixia
and Wang, He
and Xu, Huikai
and Xue, Guangming
and Yan, Haisheng
and Yang, Zhen
and Ding, Jiayu
and Gao, Yang
and Li, Feiyu
and Zhang, Yujia
and Zhang, Zikang
and Jin, Yirong
and Yu, Haifeng
and Chen, Jianxin
and Yan, Fei},
title={Efficient implementation of arbitrary two-qubit gates using unified control},
journal={Nature Physics},
year={2025},
month={Aug},
day={15},
issn={1745-2481},
doi={10.1038/s41567-025-02990-x},
url={https://doi.org/10.1038/s41567-025-02990-x}
}

@inproceedings{chen2024one,
  title={One Gate Scheme to Rule Them All: Introducing a Complex Yet Reduced Instruction Set for Quantum Computing},
  author={Chen, Jianxin and Ding, Dawei and Gong, Weiyuan and Huang, Cupjin and Ye, Qi},
  booktitle={Proceedings of the 29th ACM International Conference on Architectural Support for Programming Languages and Operating Systems, Volume 2},
  pages={779--796},
  year={2024},
  publisher={ACM},
  address={La Jolla, CA, USA}
}

@article{huang2023quantum,
  title = {Quantum Instruction Set Design for Performance},
  author = {Huang, Cupjin and Wang, Tenghui and Wu, Feng and Ding, Dawei and Ye, Qi and Kong, Linghang and Zhang, Fang and Ni, Xiaotong and Song, Zhijun and Shi, Yaoyun and Zhao, Hui-Hai and Deng, Chunqing and Chen, Jianxin},
  journal = {Physical Review Letters},
  volume = {130},
  issue = {7},
  pages = {070601},
  numpages = {7},
  year = {2023},
  month = {Feb},
  publisher = {APS},
  doi = {10.1103/PhysRevLett.130.070601},
  url = {https://link.aps.org/doi/10.1103/PhysRevLett.130.070601}
}

@misc{ibmFractionalGates,
    author = {{IBM Quantum}},
    title = {New fractional gates reduce circuit depth for utility-scale workloads},
    year = {2024},
    howpublished = {\url{https://www.ibm.com/quantum/blog/fractional-gates}},
    note = {Accessed: Nov. 18, 2024}
}

@misc{ionqPartialGates,
    author = {{IonQ}},
    title = {Getting started with IonQ's hardware-native gateset},
    year = {2023},
    howpublished = {\url{https://docs.ionq.com/guides/getting-started-with-native-gates}}
}

@misc{quantinuumArbitraryAngleGates,
  author = {{Quantinuum}},
  title = {Native Parameterized Angle Hardware Gates},
  year = {2024},
  howpublished = {\url{https://docs.quantinuum.com/systems/trainings/helios/getting_started/parameterized_angle_2_qubit_gates.html}}
}

@inproceedings{zhang2021time,
  title={Time-optimal qubit mapping},
  author={Zhang, Chi and Hayes, Ari B and Qiu, Longfei and Jin, Yuwei and Chen, Yanhao and Zhang, Eddy Z},
  booktitle={Proceedings of the 26th ACM International Conference on Architectural Support for Programming Languages and Operating Systems},
  pages={360--374},
  year={2021}
}

@inproceedings{li2019tackling,
  title={Tackling the qubit mapping problem for NISQ-era quantum devices},
  author={Li, Gushu and Ding, Yufei and Xie, Yuan},
  booktitle={Proceedings of the twenty-fourth international conference on architectural support for programming languages and operating systems},
  pages={1001--1014},
  year={2019}
}

@inproceedings{tan2020optimal,
  title={Optimal layout synthesis for quantum computing},
  author={Tan, Bochen and Cong, Jason},
  booktitle={Proceedings of the 39th International Conference on Computer-Aided Design},
  pages={1--9},
  year={2020}
}

@article{li2023qasmbench,
  title={Qasmbench: A low-level quantum benchmark suite for nisq evaluation and simulation},
  author={Li, Ang and Stein, Samuel and Krishnamoorthy, Sriram and Ang, James},
  journal={ACM Transactions on Quantum Computing},
  volume={4},
  number={2},
  pages={1--26},
  year={2023},
  publisher={ACM New York, NY}
}

@article{quetschlich2023mqt,
  title={MQT Bench: Benchmarking software and design automation tools for quantum computing},
  author={Quetschlich, Nils and Burgholzer, Lukas and Wille, Robert},
  journal={Quantum},
  volume={7},
  pages={1062},
  year={2023},
  publisher={Verein zur F{\"o}rderung des Open Access Publizierens in den Quantenwissenschaften}
}

@article{peterson2022optimal,
  title={Optimal synthesis into fixed xx interactions},
  author={Peterson, Eric C and Bishop, Lev S and Javadi-Abhari, Ali},
  journal={Quantum},
  volume={6},
  pages={696},
  year={2022},
  publisher={Verein zur F{\"o}rderung des Open Access Publizierens in den Quantenwissenschaften}
}

@inproceedings{mckinney2024mirage,
  title={Mirage: Quantum circuit decomposition and routing collaborative design using mirror gates},
  author={McKinney, Evan and Hatridge, Michael and Jones, Alex K},
  booktitle={2024 IEEE International Symposium on High-Performance Computer Architecture (HPCA)},
  pages={704--718},
  year={2024},
  organization={IEEE}
}

@misc{bqskit,
    author = {Younis, Ed and Iancu, Costin C and Lavrijsen, Wim and Davis, Marc and Smith, Ethan},
    title = {{Berkeley Quantum Synthesis Toolkit (BQSKit)}},
    howpublished = {GitHub},
    month = {4},
    year = {2021},
    doi = {10.11578/dc.20210603.2},
    osti_id = {1785933},
    code_id = {58510}
}

@article{yale2024noise,
  title={Noise-Aware Circuit Compilations for a Continuously Parameterized Two-Qubit Gateset},
  author={Yale, Christopher G and Rines, Rich and Omole, Victory and Thotakura, Bharath and Burch, Ashlyn D and Chow, Matthew NH and Ivory, Megan and Lobser, Daniel and McFarland, Brian K and Revelle, Melissa C and Clark, Susan M and Gokhale, Pranav},
  journal={arXiv preprint arXiv:2411.01094},
  year={2024}
}

@article{arute2019quantum,
  title={Quantum supremacy using a programmable superconducting processor},
  author={Frank Arute and Kunal Arya and  Ryan Babbush and  Dave Bacon  and  Joseph C. Bardin  and  Rami Barends and  Rupak Biswas  and  Sergio Boixo  and  Fernando G. S. L. Brandao  and  David A. Buell  and  Brian Burkett  and  Yu Chen  and  Zijun Chen  and  Ben Chiaro  and  Roberto Collins  and  William Courtney  and  Andrew Dunsworth  and  Edward Farhi  and  Brooks Foxen  and  Austin Fowler  and  Craig Gidney  and  Marissa Giustina and  Rob Graff  and  Keith Guerin  and  Steve Habegger  and  Matthew P. Harrigan  and  Michael J. Hartmann  and  Alan Ho  and  Markus Hoffmann  and  Trent Huang  and  Travis S. Humble  and  Sergei V. Isakov  and  Evan Jeffrey  and  Zhang Jiang  and  Dvir Kafri  and  Kostyantyn Kechedzhi  and  Julian Kelly  and  Paul V. Klimov  and  Sergey Knysh  and  Alexander Korotkov  and  Fedor Kostritsa  and  David Landhuis  and  Mike Lindmark  and Erik Lucero  and  Dmitry Lyakh  and  Salvatore Mandr\`{a}  and  Jarrod R. McClean  and  Matthew McEwen  and  Anthony Megrant  and  Xiao Mi  and  Kristel Michielsen  and  Masoud Mohseni  and  Josh Mutus  and   Ofer Naaman  and  Matthew Neeley  and  Charles Neill  and  Murphy Yuezhen Niu  and  Eric Ostby  and  Andre Petukhov  and  John C. Platt  and  Chris Quintana  and  Eleanor G. Rieffel  and  Pedram Roushan  and  Nicholas C. Rubin  and  Daniel Sank  and  Kevin J. Satzinger  and  Vadim Smelyanskiy  and  Kevin J. Sung  and  Matthew D. Trevithick  and  Amit Vainsencher  and  Benjamin Villalonga and  Theodore White  and  Z. Jamie Yao  and  Ping Yeh  and  Adam Zalcman  and  Hartmut Neven and John M. Martinis},
  journal={Nature},
  volume={574},
  number={7779},
  pages={505--510},
  year={2019},
  publisher={Nature Publishing Group}
}

@article{acharya2024quantum,
  title={Quantum error correction below the surface code threshold},
  author = {Acharya, Rajeev and Abanin, Dmitry A. and Aghababaie-Beni, Laleh and Aleiner, Igor and Andersen, Trond I. and Ansmann, Markus and Arute, Frank and Arya, Kunal and Asfaw, Abraham and Astrakhantsev, Nikita and Atalaya, Juan and Babbush, Ryan and Bacon, Dave and Ballard, Brian and Bardin, Joseph C. and Bausch, Johannes and Bengtsson, Andreas and Bilmes, Alexander and Blackwell, Sam and Boixo, Sergio and Bortoli, Gina and Boussass, Alexander and Bovaird, Jenna and Brill, Leon and Broughton, Michael and Browne, David A. and Buchea, Brett and Buckley, Bob B. and Buell, David A. and Burger, Tim and Burkett, Brian and Bushnell, Nicholas and Cabrera, Anthony and Campero, Juan and Chang, Hung-Sheng and Chen, Yu and Chen, Zijun and Chiaro, Ben and Chik, Desmond and Chou, Charina and Claes, Jalan and Clambaneanu, Amenta Y. and Cong, Josh and Collins, Roberto and Conner, Paul and Cournier, William and Crook, Alexander L. and Curtin, Ben and Das, Sayan and Davies, Alex and De Lorezzo, Laura and Debry, Dristo M. and Denver, Sean and Devoret, Michael and Di Paolo, Augustin and Donoho, Paul and Drozdov, Illy and Dunsworth, Andrew and Eark, Clint and Elich, Thanes and Eickbusch, Alec and Elbag, Aviv Moshe and Elzouka, Mahmoud and Erickson, Catherine and Faoro, Lara and Farhi, Edward and Ferreira, Vincicus S. and Burgos, Leslie Fores and Forati, Ebrahim and Fowler, Austin G. and Foxen, Brooks and Ganjam, Subas and Garcia, Gonzalo and Gasca, Robert and Genois, Elie and Giang, William and Gidney, Craig and Gilboa, Dar and Gokhale, Rajan and Daul, Alejandro Grajales and Grauman, Dietrich and Greene, Alex and Gross, Jonathan A. and Habegger, Steve and Hall, John and Hamilton, Michael C. and Hansen, Monica and Harrigan, Matthew and Harrington, Sean D. and Heras, Francisco J.H. and Hincks, Stephen and Hoel, Paula and Higgott, Oscar and Hill, Gordon and Hilton, Jeremy and Holland, George and Hong, Sabrina and Huang, Hsin-Yuan and Huff, Ashley and Huggins, William J. and Ioffe, Lev B. and Isakov, Sergei V. and Iveland, J. Justin and Jeffrey, Evan and Jiang, Zhang and Jones, Cody and Jordan, Stephen and John, Chitatil and Juhas, Pavol and Kafri, Dvir and Kang, Hui and Karamlou, Amir H. and Kechedzhi, Kostantyn and Kelly, Julian and Khaire, Trupt and Khattar, Tanuj H. and Kim, Seon and Klimov, Paul V. and Klots, Andrey R. and Kobrin, Bryce and Kohli, Pushmeet and Korotkov, Alexander N. and Kostritsa, Fedor and Kothari, Robin and Kozlovskii, Borislav and Kreikebaum, John Mark and Kurilovich, Vladislav D. and Lacroix, Nathan and Landhuis, David and Lange-Dei, Tiano and Langley, Brandon W. and Laptev, Pavel and Lau, Kim-Ming and Le Guevel, Loick and Ledford, Justin and Lee, Joonho and Lee, Kennley and Lensky, Yuri D. and Leon, Shannon and Lester, Brian J. and Li, Wing Yan and Li, Yin and Lili, Alexander T. and Liu, Wayee and Livingston, William P. and Locharla, Aditya and Lucero, Erik and Lundahl, Daniel and Luni, Aaron and Madhuk, Sid and Malone, Finnon D. and Maloney, Ashley and Mandra, Salvatore and Manyika, James and Martin, Leigh S. and Martin, Orion and Martin, Steven and Marfield, Cameron and McClean, Jarrod R. and McEwen, Matt and Meeks, Seneca and Megrant, Anthony and Mi, Xiao and Miao, Kevin C. and Mieszala, Amanda and Mola, Reza and Molina, Sebastian and Montazeri, Shirin and Morvan, Alexis and Moussa, Ramis and Muczkiewicz, Wojciech and Naaman, Ofer and Neeley, Matthew and Neil, Charles and Nersisyan, Ani and Neven, Hartmut and Newman, Michael and Ng, Jun How and Nguyen, Anthony and Nguyen, Murray and Ni, Chia-Hung and Niu, Murphy Yuezhen and O'Brien, Thomas E. and Oliver, William D. and Opremcak, Alex and Ottosson, Kristoffer and Petukhov, Andre and Pizzito, Alex and Platt, John and Potter, Rebecca and Pritchard, Orion and Pryadko, Leonid P. and Quintana, Chris and Ramachandran, Ganesh and Reagor, Matthew J. and Redding, John and Rados, Dadvi M. and Roberts, Gabrielle and Rosenberg, Elliott and Rosenfeld, Emma and Roushan, Pedram and Rubin, Nicholas C. and Saei, New Year and Sank, Daniel and Sankaragomathi, Kannan and Satzinger, Kevin J. and Schurkus, Henry F. and Schuster, Christopher and Senior, Andrew W. and Shearn, Michael J. and Shorter, Aaron and Shutty, Noah and Shvarts, Vladimir and Singh, Shraddha and Sivak, Volodymyr and Skruzny, Jindra and Small, Spencer and Smelyanskiy, Vadim and Smith, W. Clarke and Somma, Rolando and Springer, Sofia and Sterling, George and Strain, Doug and Suchard, Jordan and Szasz, Aaron and Sztein, Alex and Thor, Douglas and Torres, Alfredo and Torubaldi, M. Mert and Vishnav, Abeer and Vargas, Justin and Vdovichev, Sergey and Vidal, Guifre and Villalonga, Benjamin and Heidweiller, Catherine Vollgraff and Waltman, Steven and Wang, Shannon X. and Ware, Brayden and Weber, Kate and Weidel, Travis and White, Theodore and Wong, Kristi and Woo, Bryan W.K. and Xing, Cheng and Yao, Z. Jamie and Yeh, Ping and Ying, Bicheng and Yoo, Juhwan and Yost, Nourelin and Young, Grayson and Zalcman, Adam and Zhang, Yaxing and Zhu, Ningfeng and Zobrist, Nicholas},
  journal={Nature},
  volume={638},
  number={8052},
  pages={920},
  year={2024}
}

@article{yale2025realization,
  title={Realization and calibration of continuously parameterized two-qubit gates on a trapped-ion quantum processor},
  author={Yale, Christopher G and Burch, Ashlyn D and Chow, Matthew NH and Ruzic, Brandon P and Lobser, Daniel S and McFarland, Brian K and Revelle, Melissa C and Clark, Susan M},
  journal={arXiv preprint arXiv:2504.06259},
  year={2025}
}

@article{wei2024native,
  title={Native two-qubit gates in fixed-coupling, fixed-frequency transmons beyond cross-resonance interaction},
  author={Wei, Ken Xuan and Lauer, Isaac and Pritchett, Emily and Shanks, William and McKay, David C and Javadi-Abhari, Ali},
  journal={PRX Quantum},
  volume={5},
  number={2},
  pages={020338},
  year={2024},
  publisher={APS}
}

@article{krantz2019quantum,
  title={A quantum engineer's guide to superconducting qubits},
  author={Krantz, Philip and Kjaergaard, Morten and Yan, Fei and Orlando, Terry P and Gustavsson, Simon and Oliver, William D},
  journal={Applied Physics Reviews},
  volume={6},
  number={2},
  pages = {021318},
  year={2019},
  publisher={AIP Publishing}
}

@article{bruzewicz2019trapped,
  title={Trapped-ion quantum computing: Progress and challenges},
  author={Bruzewicz, Colin D and Chiaverini, John and McConnell, Robert and Sage, Jeremy M},
  journal={Applied Physics Reviews},
  volume={6},
  number={2},
  pages = {021314},
  year={2019},
  publisher={AIP Publishing}
}

@misc{qiskitXXDecomposer,
    author = {{IBM Quantum}},
    title = {Qiskit API},
    year = {2025},
    howpublished = {\url{https://quantum.cloud.ibm.com/docs/en/api/qiskit/qiskit.synthesis.XXDecomposer}},
}

@misc{cirqSQiSWDecomposer,
  author       = {{Google Quantum AI}},
  title        = {Cirq API},
  year         = {2025},
  howpublished = {\url{https://quantumai.google/reference/python/cirq/two_qubit_matrix_to_sqrt_iswap_operations}},
}

@article{zou2024lightsabre,
  title={Lightsabre: A lightweight and enhanced sabre algorithm},
  author={Zou, Henry and Treinish, Matthew and Hartman, Kevin and Ivrii, Alexander and Lishman, Jake},
  journal={arXiv preprint arXiv:2409.08368},
  year={2024}
}

@inproceedings{codina2001unified,
  title={A unified modulo scheduling and register allocation technique for clustered processors},
  author={Codina, Josep M and S{\'a}nchez, Jes{\'u}s and Gonz{\'a}lez, Antonio},
  booktitle={Proceedings 2001 International Conference on Parallel Architectures and Compilation Techniques},
  pages={175--184},
  year={2001},
  organization={IEEE}
}

@article{hennessy1983postpass,
  title={Postpass code optimization of pipeline constraints},
  author={Hennessy, John L and Gross, Thomas},
  journal={ACM Transactions on Programming Languages and Systems (TOPLAS)},
  volume={5},
  number={3},
  pages={422--448},
  year={1983},
  publisher={ACM New York, NY, USA}
}

@article{chaitin1982register,
  title={Register allocation \& spilling via graph coloring},
  author={Chaitin, Gregory J},
  journal={ACM Sigplan Notices},
  volume={17},
  number={6},
  pages={98--101},
  year={1982},
  publisher={ACM New York, NY, USA}
}

@article{poletto1999linear,
  title={Linear scan register allocation},
  author={Poletto, Massimiliano and Sarkar, Vivek},
  journal={ACM Transactions on Programming Languages and Systems (TOPLAS)},
  volume={21},
  number={5},
  pages={895--913},
  year={1999},
  publisher={ACM New York, NY, USA}
}

@article{cross2019validating,
  title={Validating quantum computers using randomized model circuits},
  author={Cross, Andrew W and Bishop, Lev S and Sheldon, Sarah and Nation, Paul D and Gambetta, Jay M},
  journal={Physical Review A},
  volume={100},
  number={3},
  pages={032328},
  year={2019},
  publisher={APS}
}

@article{zulehner2018efficient,
  title={An efficient methodology for mapping quantum circuits to the IBM QX architectures},
  author={Zulehner, Alwin and Paler, Alexandru and Wille, Robert},
  journal={IEEE Transactions on Computer-Aided Design of Integrated Circuits and Systems},
  volume={38},
  number={7},
  pages={1226--1236},
  year={2018},
  publisher={IEEE}
}

@article{lao2021timing,
  title={Timing and resource-aware mapping of quantum circuits to superconducting processors},
  author={Lao, Lingling and Van Someren, Hans and Ashraf, Imran and Almudever, Carmen G},
  journal={IEEE Transactions on Computer-Aided Design of Integrated Circuits and Systems},
  volume={41},
  number={2},
  pages={359--371},
  year={2021},
  publisher={IEEE}
}

@inproceedings{annechini2025ddroute,
  title={Ddroute: A novel depth-driven approach to the qubit routing problem},
  author={Annechini, Alessandro and Venere, Marco and Sciuto, Donatella and Santambrogio, Marco D},
  booktitle={2025 62nd ACM/IEEE Design Automation Conference (DAC)},
  pages={1--7},
  year={2025},
  organization={IEEE}
}

@article{childs2019circuit,
  title={Circuit transformations for quantum architectures},
  author={Childs, Andrew M and Schoute, Eddie and Unsal, Cem M},
  journal={arXiv preprint arXiv:1902.09102},
  year={2019}
}

@inproceedings{kalloor2024quantum,
  title={Quantum hardware roofline: Evaluating the impact of gate expressivity on quantum processor design},
  author={Kalloor, Justin and Weiden, Mathias and Younis, Ed and Kubiatowicz, John and De Jong, Bert and Iancu, Costin},
  booktitle={2024 IEEE International Conference on Quantum Computing and Engineering (QCE)},
  volume={1},
  pages={805--816},
  year={2024},
  organization={IEEE}
}

@misc{davis2019heuristics,
  title={Heuristics for quantum compiling with a continuous gate set},
  author={Davis, Marc Grau and Smith, Ethan and Tudor, Ana and Sen, Koushik and Siddiqi, Irfan and Iancu, Costin},
  note={arXiv preprint arXiv:1912.02727},
  year={2019},
  numpages={12}
}

@article{wu2020qgo,
  title={QGo: Scalable quantum circuit optimization using automated synthesis},
  author={Wu, Xin-Chuan and Davis, Marc Grau and Chong, Frederic T and Iancu, Costin},
  journal={arXiv preprint arXiv:2012.09835},
  year={2020}
}

@inproceedings{kukliansky2023qfactor,
  title={QFactor: A Domain-Specific Optimizer for Quantum Circuit Instantiation},
  author={Kukliansky, Alon and Younis, Ed and Cincio, Lukasz and Iancu, Costin},
  booktitle={2023 IEEE International Conference on Quantum Computing and Engineering (QCE)},
  volume={1},
  pages={814--824},
  year={2023},
  organization={IEEE}
}

@inproceedings{younis2021qfast,
  title={Qfast: Conflating search and numerical optimization for scalable quantum circuit synthesis},
  author={Younis, Ed and Sen, Koushik and Yelick, Katherine and Iancu, Costin},
  booktitle={2021 IEEE International Conference on Quantum Computing and Engineering (QCE)},
  pages={232--243},
  year={2021},
  organization={IEEE}
}

@article{eickbusch2024demonstrating,
  title={Demonstration of dynamic surface codes},
  author = {Eickbusch, Alec and McEwen, Matt and Sivak, Volodymyr and Bourassa, Alexandre and Atalaya, Juan and Claes, Jahan and Kafri, Dvir and Gidney, Craig and Warren, Christopher W. and Gross, Jonathan and Opremcak, Alex and Zobrist, Nicholas and Miao, Kevin C. and Roberts, Gabrielle and Satzinger, Kevin J. and Bengtsson, Andreas and Neeley, Matthew and Livingston, William P. and Greene, Alex and Acharya, Rajeev and Beni, Laleh Aghababaie and Aigeldinger, Georg and Alcaraz, Ross and Andersen, Trond I. and Ansmann, Markus and Arute, Frank and Arya, Kunal and Asfaw, Abraham and Babbush, Ryan and Ballard, Brian and Bardin, Joseph C. and Bilmes, Alexander and Bovaird, Jenna and Bowers, Dylan and Brill, Leon and Broughton, Michael and Browne, David A. and Buchea, Brett and Buckley, Bob B. and Burger, Tim and Burkett, Brian and Bushnell, Nicholas and Cabrera, Anthony and Campero, Juan and Chang, Hung-Shen and Chiaro, Ben and Chih, Liang-Ying and Cleland, Agnetta Y. and Cogan, Josh and Collins, Roberto and Conner, Paul and Courtney, William and Crook, Alexander L. and Curtin, Ben and Das, Sayan and Del Toro Barba, Alexander and Demura, Sean and De Lorenzo, Laura and Di Paolo, Agustin and Donohoe, Paul and Drozdov, Ilya K. and Dunsworth, Andrew and Elbag, Aviv Moshe and Elzouka, Mahmoud and Erickson, Catherine and Ferreira, Vinicius S. and Flores Burgos, Leslie and Forati, Ebrahim and Fowler, Austin G. and Foxen, Brooks and Ganjam, Suhas and Garcia, Gonzalo and Gasca, Robert and Genois, {\'E}lie and Giang, William and Gilboa, Dar and Gosula, Raja and Grajales Dau, Alejandro and Graumann, Dietrich and Ha, Tan and Habegger, Steve and Hansen, Monica and Harrigan, Matthew P. and Harrington, Sean D. and Heslin, Stephen and Heu, Paula and Higgott, Oscar and Hiltermann, Reno and Hilton, Jeremy and Huang, Hsin-Yuan and Huff, Ashley and Huggins, William J. and Jeffrey, Evan and Jiang, Zhang and Jin, Xiaoxuan and Jones, Cody and Joshi, Chaitali and Juhas, Pavol and Kabel, Andreas and Kang, Hui and Karamlou, Amir H. and Kechedzhi, Kostyantyn and Khaire, Trupti and Khattar, Tanuj and Khezri, Mostafa and Kim, Seon and Kobrin, Bryce and Korotkov, Alexander N. and Kostritsa, Fedor and Kreikebaum, John Mark and Kurilovich, Vladislav D. and Landhuis, David and Lange-Dei, Tiano and Langley, Brandon W. and Lau, Kim-Ming and Ledford, Justin and Lee, Kenny and Lester, Brian J. and Le Guevel, Lo{\"\i}ck and Li, Wing Yan and Lill, Alexander T. and Locharla, Aditya and Lucero, Erik and Lundahl, Daniel and Lunt, Aaron and Madhuk, Sid and Maloney, Ashley and Mandr\`{a}, Salvatore and Martin, Leigh S. and Martin, Orion and Maxfield, Cameron and McClean, Jarrod R. and Meeks, Seneca and Megrant, Anthony and Molavi, Reza and Molina, Sebastian and Montazeri, Shirin and Movassagh, Ramis and Newman, Michael and Nguyen, Anthony and Nguyen, Murray and Ni, Chia-Hung and Oas, Logan and Orosco, Raymond and Ottosson, Kristoffer and Pizzuto, Alex and Potter, Rebecca and Pritchard, Orion and Quintana, Chris and Ramachandran, Ganesh and Reagor, Matthew J. and Rhodes, David M. and Rosenberg, Eliott and Rossi, Elizabeth and Sankaragomathi, Kannan and Schurkus, Henry F. and Shearn, Michael J. and Shorter, Aaron and Shutty, Noah and Shvarts, Vladimir and Small, Spencer and Smith, W. Clarke and Springer, Sofia and Sterling, George and Suchard, Jordan and Szasz, Aaron and Sztein, Alex and Thor, Douglas and Tomita, Eifu and Torres, Alfredo and Torunbalci, M. Mert and Vaishnav, Abeer and Vargas, Justin and Vdovichev, Sergey and Vidal, Guifre and Vollgraff Heidweiller, Catherine and Waltman, Steven and Waltz, Jonathan and Wang, Shannon X. and Ware, Brayden and Weidel, Travis and White, Theodore and Wong, Kristi and Woo, Bryan W. K. and Woodson, Maddy and Xing, Cheng and Yao, Z. Jamie and Yeh, Ping and Ying, Bicheng and Yoo, Juhwan and Yosri, Noureldin and Young, Grayson and Zalcman, Adam and Zhang, Yaxing and Zhu, Ningfeng and Boixo, Sergio and Kelly, Julian and Smelyanskiy, Vadim and Neven, Hartmut and Bacon, Dave and Chen, Zijun and Klimov, Paul V. and Roushan, Pedram and Neill, Charles and Chen, Yu and Morvan, Alexis},
  journal={Nature Physics},
  pages={1--8},
  year={2025},
  publisher={Nature Publishing Group UK London}
}

@article{mcewen2023relaxing,
  title={Relaxing hardware requirements for surface code circuits using time-dynamics},
  author={McEwen, Matt and Bacon, Dave and Gidney, Craig},
  journal={Quantum},
  volume={7},
  pages={1172},
  year={2023},
  publisher={Verein zur F{\"o}rderung des Open Access Publizierens in den Quantenwissenschaften}
}

@article{zhou2024halma,
  title={Routing-based technique for defect mitigation in quantum error correction},
  author={Zhou, Runshi and Zhang, Fang and Kong, Linghang and Wu, Feng and Zhao, Hui-Hai and Chen, Jianxin},
  journal={Physical Review A},
  volume={113},
  number={3},
  pages={032401},
  year={2026},
  publisher={APS}
}

@article{zhang2003geometric,
  title={Geometric theory of nonlocal two-qubit operations},
  author={Zhang, Jun and Vala, Jiri and Sastry, Shankar and Whaley, K Birgitta},
  journal={Physical Review A},
  volume={67},
  number={4},
  pages={042313},
  year={2003},
  publisher={APS}
}

@misc{tucci2005introduction,
  title={An introduction to Cartan's KAK decomposition for QC programmers},
  author={Tucci, Robert R},
  note={arXiv preprint quant-ph/0507171},
  year={2005}
}

@inproceedings{zulehner2019compiling,
  title={Compiling SU (4) quantum circuits to IBM QX architectures},
  author={Zulehner, Alwin and Wille, Robert},
  booktitle={Proceedings of the 24th Asia and South Pacific Design Automation Conference},
  pages={185--190},
  year={2019},
  publisher={ACM New York, NY, USA},
  address={Tokyo, Japan}
}

@inproceedings{bullock2003arbitrary,
  title={An arbitrary two-qubit computation in 23 elementary gates or less},
  author={Bullock, Stephen S and Markov, Igor L},
  booktitle={Proceedings of the 40th Annual Design Automation Conference},
  pages={324--329},
  year={2003},
  address={Anaheim, CA, USA},
  publisher={IEEE}
}

@misc{crooks2020gates,
  title={Gates, states, and circuits},
  author={Crooks, Gavin E},
  year={2020},
  note = {Available at \url{https://threeplusone.com/pubs/on-gates-v0-5/}}
}

@misc{weylchamber,
  author       = {Goerz, Michael and McKinney, Evan},
  title        = {weylchamber: Python package for analyzing two-qubit gates in the Weyl chamber},
  year         = {2024},
  publisher    = {PyPI},
  version      = {0.6.0},
  howpublished = {\url{https://pypi.org/project/weylchamber/}},
  note         = {Python package},
  license      = {BSD}
}

@article{peterson2020fixed,
  title={Fixed-depth two-qubit circuits and the monodromy polytope},
  author={Peterson, Eric C and Crooks, Gavin E and Smith, Robert S},
  journal={Quantum},
  volume={4},
  pages={247},
  year={2020},
  publisher={Verein zur F{\"o}rderung des Open Access Publizierens in den Quantenwissenschaften}
}

@article{proctor2022measuring,
  title={Measuring the capabilities of quantum computers},
  author={Proctor, Timothy and Rudinger, Kenneth and Young, Kevin and Nielsen, Erik and Blume-Kohout, Robin},
  journal={Nature Physics},
  volume={18},
  number={1},
  pages={75--79},
  year={2022},
  publisher={Nature Publishing Group UK London}
}

@article{rigetti2010fully,
  title={Fully microwave-tunable universal gates in superconducting qubits with linear couplings and fixed transition frequencies},
  author={Rigetti, Chad and Devoret, Michel},
  journal={Physical Review B—Condensed Matter and Materials Physics},
  volume={81},
  number={13},
  pages={134507},
  year={2010},
  publisher={APS}
}

@book{nielsen2010quantum,
  title={Quantum computation and quantum information},
  author={Nielsen, Michael A and Chuang, Isaac L},
  year={2010},
  publisher={Cambridge university press}
}

@inproceedings{shor1994algorithms,
  title={Algorithms for quantum computation: discrete logarithms and factoring},
  author={Shor, Peter W},
  booktitle={Proceedings 35th annual symposium on foundations of computer science},
  pages={124--134},
  year={1994},
  organization={Ieee}
}

@article{harrow2009quantum,
  title={Quantum algorithm for linear systems of equations},
  author={Harrow, Aram W and Hassidim, Avinatan and Lloyd, Seth},
  journal={Physical review letters},
  volume={103},
  number={15},
  pages={150502},
  year={2009},
  publisher={APS}
}

@article{lloyd1996universal,
  title={Universal quantum simulators},
  author={Lloyd, Seth},
  journal={Science},
  volume={273},
  number={5278},
  pages={1073--1078},
  year={1996},
  publisher={American Association for the Advancement of Science}
}

@article{chamberland2020topological,
  title={Topological and subsystem codes on low-degree graphs with flag qubits},
  author={Chamberland, Christopher and Zhu, Guanyu and Yoder, Theodore J and Hertzberg, Jared B and Cross, Andrew W},
  journal={Physical Review X},
  volume={10},
  number={1},
  pages={011022},
  year={2020},
  publisher={APS}
}

@misc{canopusGitHub,
    year = {2025},
    title = {{\textsc{Canopus} GitHub repo}},
    howpublished = {\url{https://github.com/Youngcius/canopus}}
}

@article{wang2026demonstration,
  title={Demonstration of low-overhead quantum error correction codes},
  author={Wang, Ke and Lu, Zhide and Zhang, Chuanyu and Liu, Gongyu and Chen, Jiachen and Wang, Yanzhe and Wu, Yaozu and Xu, Shibo and Zhu, Xuhao and Jin, Feitong and others},
  journal={Nature Physics},
  pages={1--7},
  year={2026},
  publisher={Nature Publishing Group UK London}
}

@article{breuckmann2021quantum,
  title={Quantum low-density parity-check codes},
  author={Breuckmann, Nikolas P and Eberhardt, Jens Niklas},
  journal={PRX Quantum},
  volume={2},
  number={4},
  pages={040101},
  year={2021},
  publisher={APS}
}

@article{bravyi2024high,
  title={High-threshold and low-overhead fault-tolerant quantum memory},
  author={Bravyi, Sergey and Cross, Andrew W and Gambetta, Jay M and Maslov, Dmitri and Rall, Patrick and Yoder, Theodore J},
  journal={Nature},
  volume={627},
  number={8005},
  pages={778--782},
  year={2024},
  publisher={Nature Publishing Group UK London}
}

@article{panteleev2021degenerate,
  title={Degenerate quantum LDPC codes with good finite length performance},
  author={Panteleev, Pavel and Kalachev, Gleb},
  journal={Quantum},
  volume={5},
  pages={585},
  year={2021},
  publisher={Verein zur F{\"o}rderung des Open Access Publizierens in den Quantenwissenschaften}
}

@article{Gidney2021StimAF,
  title={Stim: a fast stabilizer circuit simulator},
  author={Craig Gidney},
  journal={Quantum},
  year={2021},
  volume={5},
  pages={497},
  url={https://api.semanticscholar.org/CorpusID:232104816}
}

@article{hillmann2024localized,
  title={Localized statistics decoding: A parallel decoding algorithm for quantum low-density parity-check codes},
  author={Hillmann, Timo and Berent, Lucas and Quintavalle, Armanda O and Eisert, Jens and Wille, Robert and Roffe, Joschka},
  journal={arXiv preprint arXiv:2406.18655},
  year={2024}
}

@inproceedings{liu2023tackling,
  title={Tackling the qubit mapping problem with permutation-aware synthesis},
  author={Liu, Ji and Younis, Ed and Weiden, Mathias and Hovland, Paul and Kubiatowicz, John and Iancu, Costin},
  booktitle={2023 IEEE International Conference on Quantum Computing and Engineering (QCE)},
  volume={1},
  pages={745--756},
  year={2023},
  organization={IEEE}
}

@article{nguyen2024programmable,
  title={Programmable Heisenberg interactions between Floquet qubits},
  author={Nguyen, Long B and Kim, Yosep and Hashim, Akel and Goss, Noah and Marinelli, Brian and Bhandari, Bibek and Das, Debmalya and Naik, Ravi K and Kreikebaum, John Mark and Jordan, Andrew N and Santiago, David I. and Siddiqi, Irfan},
  journal={Nature Physics},
  volume={20},
  number={2},
  pages={240--246},
  year={2024},
  publisher={Nature Publishing Group UK London}
}

@inproceedings{liu2022not,
  title={Not all swaps have the same cost: A case for optimization-aware qubit routing},
  author={Liu, Ji and Li, Peiyi and Zhou, Huiyang},
  booktitle={2022 IEEE International Symposium on High-Performance Computer Architecture (HPCA)},
  pages={709--725},
  year={2022},
  organization={IEEE}
}

@article{tang2024quantum,
  title={Quantum circuit synthesis with SQiSW},
  author={Tang, Jialiang and Zhang, Jialin and Sun, Xiaoming},
  journal={arXiv preprint arXiv:2412.14828},
  year={2024}
}

@article{kitaev1995quantum,
  title={Quantum measurements and the Abelian stabilizer problem},
  author={Kitaev, A Yu},
  journal={arXiv preprint quant-ph/9511026},
  year={1995}
}

@inproceedings{jin2024optimizing,
  title={Optimizing quantum fourier transformation (qft) kernels for modern nisq and ft architectures},
  author={Jin, Yuwei and Gao, Xiangyu and Guo, Minghao and Chen, Henry and Hua, Fei and Zhang, Chi and Zhang, Eddy Z},
  booktitle={SC24: International Conference for High Performance Computing, Networking, Storage and Analysis},
  pages={1--15},
  year={2024},
  organization={IEEE}
}

@article{maslov2007linear,
  title={Linear depth stabilizer and quantum Fourier transformation circuits with no auxiliary qubits in finite-neighbor quantum architectures},
  author={Maslov, Dmitri},
  journal={Physical Review A—Atomic, Molecular, and Optical Physics},
  volume={76},
  number={5},
  pages={052310},
  year={2007},
  publisher={APS}
}

@inproceedings{Lattner2004,
    author    = {Lattner, Chris and Adve, Vikram},
    title     = {LLVM: A Compilation Framework for Lifelong Program Analysis \& Transformation},
    booktitle = {Proceedings of the International Symposium on Code Generation and Optimization (CGO)},
    year      = {2004},
    pages     = {75--86},
    publisher = {IEEE Computer Society},
    address   = {Washington, DC, USA},
    doi       = {10.1109/CGO.2004.1281665},
    isbn      = {0769521029}
}

@article{mckinney2025two,
  title={Two-Qubit Gate Synthesis via Linear Programming for Heterogeneous Instruction Sets},
  author={McKinney, Evan and Bishop, Lev S},
  journal={arXiv preprint arXiv:2505.00543},
  year={2025}
}

@inproceedings{patel2022quest,
  title={Quest: systematically approximating quantum circuits for higher output fidelity},
  author={Patel, Tirthak and Younis, Ed and Iancu, Costin and de Jong, Wibe and Tiwari, Devesh},
  booktitle={Proceedings of the 27th ACM International Conference on Architectural Support for Programming Languages and Operating Systems},
  pages={514--528},
  year={2022}
}

@inproceedings{yang2026reconfigurable,
  title={Reconfigurable Quantum Instruction Set Computers for High Performance Attainable on Hardware},
  author={Yang, Zhaohui and Ding, Dawei and Ye, Qi and Huang, Cupjin and Chen, Jianxin and Xie, Yuan},
  booktitle={Proceedings of the 31st ACM International Conference on Architectural Support for Programming Languages and Operating Systems, Volume 2},
  pages={1523--1546},
  year={2026}
}

@article{zhu2025quantum,
  title={Quantum Compiler Design for Qubit Mapping and Routing: A Cross-Architectural Survey of Superconducting, Trapped-Ion, and Neutral Atom Systems},
  author={Zhu, Chenghong and Wu, Xian and Yang, Zhaohui and Wang, Jingbo and Wu, Anbang and Zheng, Shenggen and Wang, Xin},
  journal={arXiv preprint arXiv:2505.16891},
  year={2025}
}

@inproceedings{tan2021optimal,
  title={Optimal qubit mapping with simultaneous gate absorption},
  author={Tan, Bochen and Cong, Jason},
  booktitle={2021 IEEE/ACM International Conference On Computer Aided Design (ICCAD)},
  pages={1--8},
  year={2021},
  organization={IEEE}
}

@inproceedings{lao2021designing,
  title={Designing calibration and expressivity-efficient instruction sets for quantum computing},
  author={Lao, Lingling and Murali, Prakash and Martonosi, Margaret and Browne, Dan},
  booktitle={Proceedings of the 48th Annual International Symposium on Computer Architecture},
  pages={363--376},
  year={2021}
}

@inproceedings{lao20222qan,
  title={2QAN: a quantum compiler for 2-local qubit Hamiltonian simulation algorithms},
  author={Lao, Lingling and Browne, Dan E},
  booktitle={Proceedings of the 49th Annual International Symposium on Computer Architecture},
  pages={430--445},
  year={2022}
}

@misc{cuccaro2004new,
      title={A new quantum ripple-carry addition circuit}, 
      author={Steven A. Cuccaro and Thomas G. Draper and Samuel A. Kutin and David Petrie Moulton},
      year={2004},
      eprint={quant-ph/0410184},
      archivePrefix={arXiv},
      primaryClass={quant-ph},
      url={https://arxiv.org/abs/quant-ph/0410184}
}

@article{farhi2014quantum,
  title={A quantum approximate optimization algorithm},
  author={Farhi, Edward and Goldstone, Jeffrey and Gutmann, Sam},
  journal={arXiv preprint arXiv:1411.4028},
  year={2014}
}

@article{foxen2020demonstrating,
  title={Demonstrating a continuous set of two-qubit gates for near-term quantum algorithms},
  author={B. Foxen and C. Neill and A. Dunsworth and P. Roushan and B. Chiaro and A. Megrant and J. Kelly and Zijun Chen and K. Satzinger and R. Barends and F. Arute and K. Arya and R. Babbush and D. Bacon and J.C. Bardin and S. Boixo and D. Buell and B. Burkett and Yu Chen and R. Collins and E. Farhi and A. Fowler and C. Gidney and M. Giustina and R. Graff and M. Harrigan and T. Huang and S.V. Isakov and E. Jeffrey and Z. Jiang and D. Kafri and K. Kechedzhi and P. Klimov and A. Korotkov and F. Kostritsa and D. Landhuis and E. Lucero and J. McClean and M. McEwen and X. Mi and M. Mohseni and J.Y. Mutus and O. Naaman and M. Neeley and M. Niu and A. Petukhov and C. Quintana and N. Rubin and D. Sank and V. Smelyanskiy and A. Vainsencher and T.C. White and Z. Yao and P. Yeh and A. Zalcman and H. Neven and John M. Martinis},
  journal={Physical Review Letters},
  volume={125},
  number={12},
  pages={120504},
  year={2020},
  publisher={APS}
}

@article{abrams2020implementation,
  title={Implementation of XY entangling gates with a single calibrated pulse},
  author={Abrams, Deanna M and Didier, Nicolas and Johnson, Blake R and Silva, Marcus P da and Ryan, Colm A},
  journal={Nature Electronics},
  volume={3},
  number={12},
  pages={744--750},
  year={2020},
  publisher={Nature Publishing Group UK London}
}
%%%%%%%%%%%%%%%%%%%%%%%%%%%%%%%%%%%%

\onecolumn
\appendix
\section{Canonical gate and 2Q circuit synthesis}\label{sec:appendix_A}

% In this section we show the basic mathematical properties of the canonical form of 2Q unitary and then discuss the synthesis capability of some 2Q basis gates.

\subsection{Canonical decomposition}\label{appendix:canonical_form}

$\mathbf{SU}(N)$ is a real manifold with dimension $N^2 - 1$, within which any element is a \emph{special unitary} matrix with determinant equal to 1. Since the global phase does not affect quantum computation processes, it is sufficient to focus on the mathematical properties of special unitaries in the area of circuit synthesis. A generic 2Q gate, despite having 15 real parameters, can have its nonlocal behavior fully characterized by only 3 real parameters. This method, known as \emph{Canonical decomposition} or \emph{KAK decomposition} from Lie algebra theory, is widely adopted in quantum computing~\cite{zhang2003geometric,tucci2005introduction,bullock2003arbitrary,zulehner2019compiling}. Specifically, for any $U \in \mathbf{SU}(4)$, there exists a unique $\vec{\eta} = (x, y, z) \in W \subseteq \mathbb{R}^3$, along with $V_1, V_2, V_3, V_4 \in \mathbf{SU}(2)$ and a global phase, such that
\begin{align}
U = g \cdot (V_1 \otimes V_2) e^{-i\vec{\eta} \cdot \vec{\Sigma}} (V_3 \otimes V_4),\, g \in \{1, i\}
\end{align}
where $\vec{\Sigma} \equiv (XX, YY, ZZ)$~\cite{tucci2005introduction}. The set
\begin{align}
W := \left\{(x, y, z) \in \mathbb{R}^3 \,\vert\, \frac{\pi}{4} \geq x \geq y \geq |z|,\, z \geq 0 \text{ if } x = \frac{\pi}{4}\right\}
\end{align}
is known as the \emph{Weyl chamber}~\cite{zhang2003geometric}, and  $\vec{\eta} \in W$ is known as the \emph{Weyl coordinate} of $ U $. We also refer to a gate of the form 
\begin{align*}
    \Can(a,b,c):= e^{-i\frac{\pi}{2}(a\,XX+b\,YY+c\,ZZ)} = \begin{pmatrix}
        e^{-i \frac{c\pi}{2}} \cos{\frac{(a-b)\pi}{2}} & 0 & 0 & -i e^{-i \frac{c\pi}{2}} \sin{\frac{(a-b)\pi}{2}} \\
        0 & e^{i \frac{c\pi}{2}} \cos{\frac{(a+b)\pi}{2}} & -i e^{i \frac{c\pi}{2}} \sin{\frac{(a+b)\pi}{2}} & 0 \\
        0 & -i e^{i \frac{c\pi}{2}} \sin{\frac{(a+b)\pi}{2}} & e^{i \frac{c\pi}{2}} \cos{\frac{(a+b)\pi}{2}} & 0 \\
        -i e^{-i \frac{c\pi}{2}} \sin{\frac{(a-b)\pi}{2}} & 0 & 0 & e^{-i \frac{c\pi}{2}} \cos{\frac{(a-b)\pi}{2}}
    \end{pmatrix}
\end{align*}
as a \emph{canonical} gate. Two 2Q gates $ U $ and $ V $ are considered \emph{locally equivalent} if they differ only by 1Q gates, meaning their canonical coefficients can be transformed into one another via the equivalence rules~\cite{crooks2020gates}:
\begin{enumerate}
    \item $(a,b,c)\sim (b,a,c)$ or $(a,b,c)\sim (c,b,a)$, i.e., any permutation of the coefficients;
    \item $(a,b,c)\sim (-a, -b, c)$;
    \item $(a,b,c)\sim (a-1, b, c)$;
    \item $(1/2, b, c) \sim (1/2, b, -c)$.
\end{enumerate}
Note that we align the conventional that canonical coefficient $ (a,b,c) $ differs from Weyl coordinate $ (x,y,z) $ by a $ \frac{\pi}{2} $ factor. Unless otherwise specified, the canonical coefficients of gates in quantum ISAs and circuits are confined to $ \frac{1}{2}\geq a \geq b \geq \lvert c\rvert $. While for the Weyl chamber visualization by means of \code{weylchamber}~\cite{weylchamber}, we assume the Weyl coordinates are confined to $\left\{\frac{\pi}{4}\geq x \geq y \geq z\geq 0\right\} \cup \left\{\frac{\pi}{4} \geq \frac{\pi}{2}-x \geq y \geq z \geq 0\right\}$, as illustrated by \Cref{fig:weyl_chamber}. Conversion of Weyl coordinates for different conventions is simple according to the equivalence rules above.

\subsection{Quantum ISA and the synthesis capability}\label{appendix:isa_analysis}

\begin{figure}[bp]
    \centering
    \begin{minipage}[t]{0.48\textwidth}
        \centering
        \includegraphics[width=\textwidth]{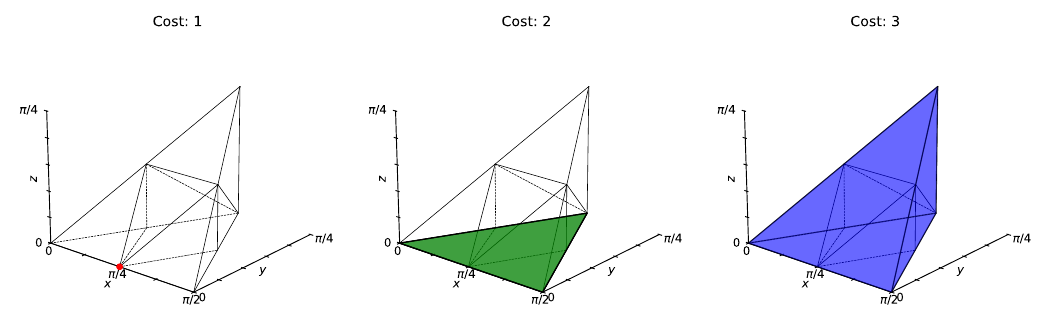}
        \caption{Coverage set for \CXISA\ ISA.}
        \label{fig:coverage_cx}
    \end{minipage}
    \hfill
    \begin{minipage}[t]{0.48\textwidth}
        \centering
        \includegraphics[width=\textwidth]{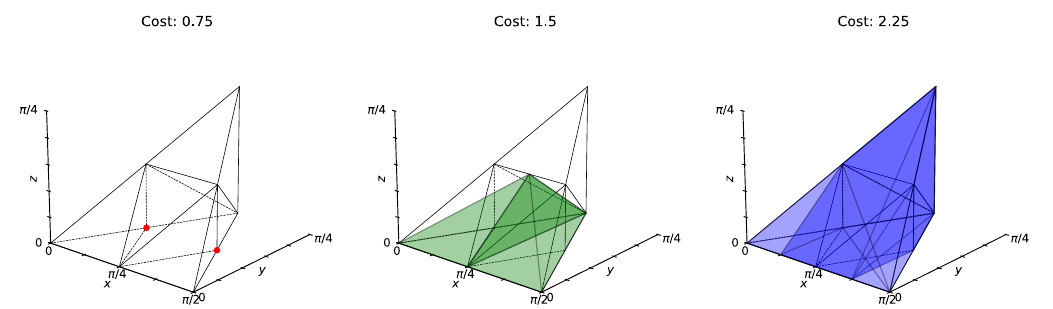}
        \caption{Coverage set for \SQiSWISA\ ISA.}
        \label{fig:coverage_sqisw}
    \end{minipage}
\end{figure}

\begin{figure}[tbp]
    \centering
    \includegraphics[width=\textwidth]{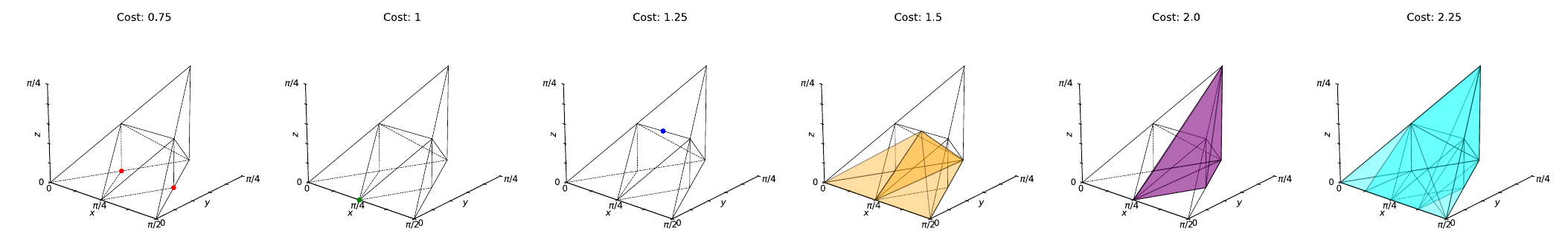}
    \caption{Coverage set for \SQiSWWithMirrorISA\ ISA.}
    \label{fig:coverage_sqisw_with_mirror}
\end{figure}

\begin{figure}[tbp]
    \centering
    \includegraphics[width=\textwidth]{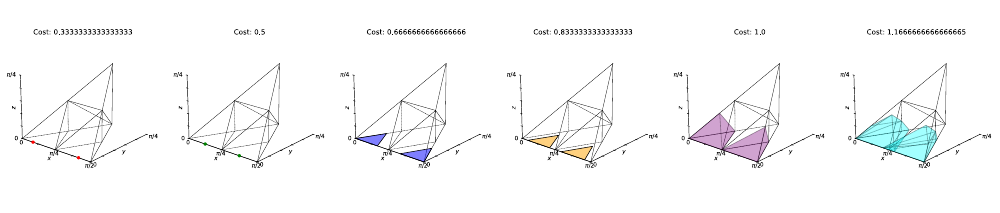}\vspace{-1.5em}
    \includegraphics[width=\textwidth]{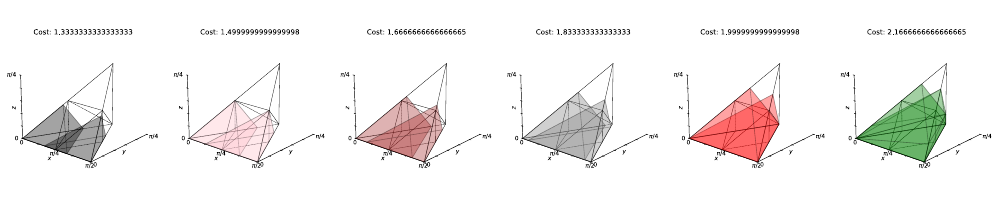}\vspace{-1.5em}
    \includegraphics[width=0.83\textwidth]{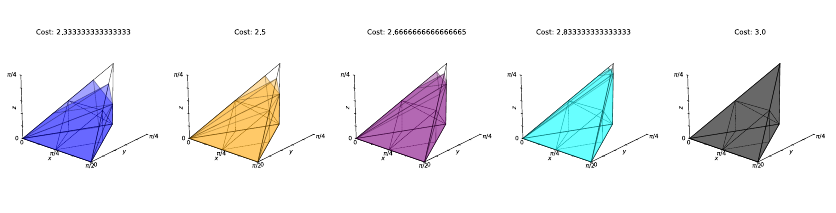}\vspace{-2em}
    \caption{Coverage set for \ZZPhaseISA\ ISA.}
    \label{fig:coverage_zzphase}
\end{figure}

\begin{figure}[tbp]
    \centering
    \includegraphics[width=\textwidth]{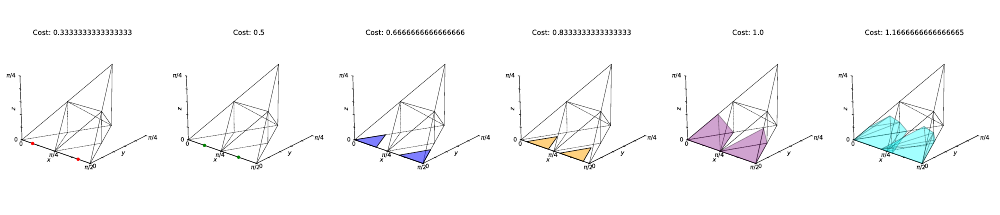}\vspace{-1.5em}
    \includegraphics[width=\textwidth]{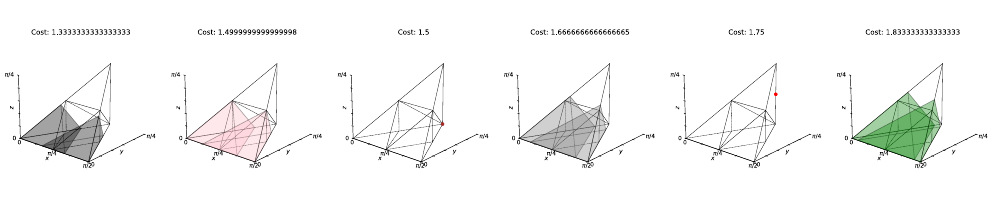}\vspace{-1.5em}
    \includegraphics[width=\textwidth]{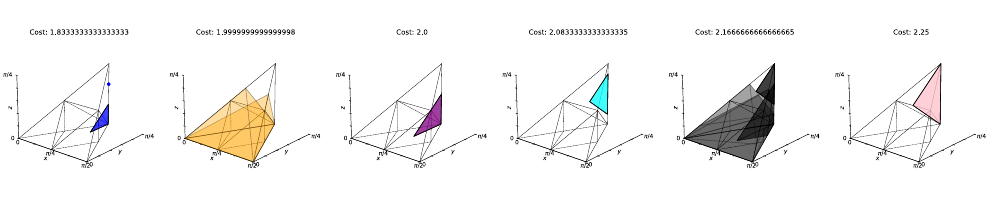}\vspace{-1.5em}
    \includegraphics[width=0.5\textwidth]{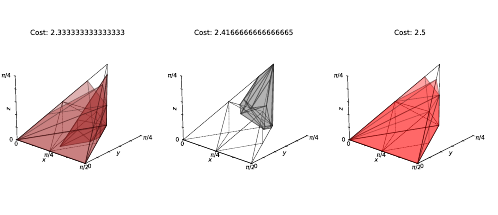}\vspace{-2em}
    \caption{Coverage set for \ZZPhaseWithMirrorISA\ ISA.}
    \label{fig:coverage_zzphase_with_mirror}
\end{figure}

\begin{figure}[tbp]
    \centering
    \includegraphics[width=\textwidth]{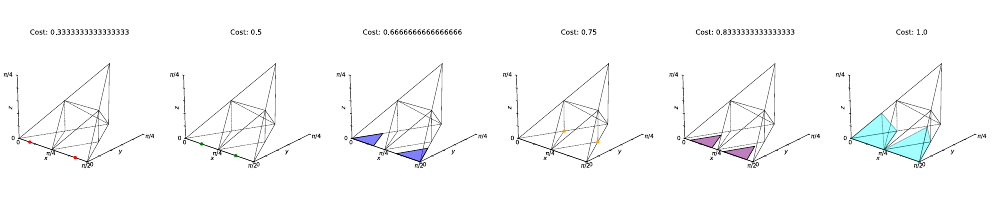}\vspace{-1.5em}
    \includegraphics[width=\textwidth]{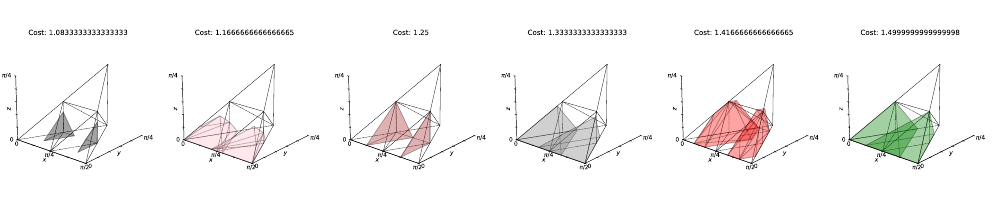}\vspace{-1.5em}
    \includegraphics[width=\textwidth]{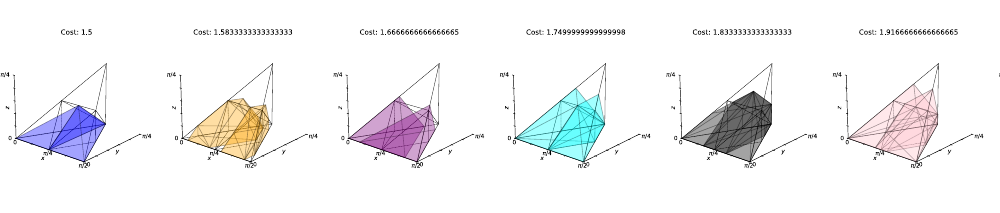}\vspace{-1.5em}
    \includegraphics[width=0.66\textwidth]{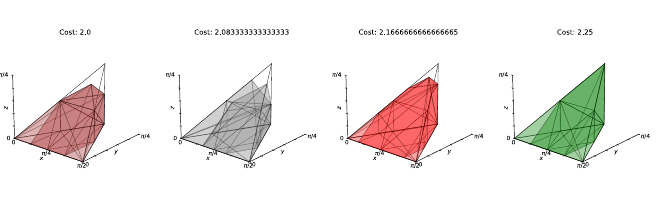}\vspace{-2em}
    \caption{Coverage set for \HetISA\ ISA.}
    \label{fig:coverage_het}
\end{figure}

A quantum ISA typically includes qubit initialization, a universal gate set, and measurement. It serves as an interface between software and hardware by mapping high-level semantics of quantum programs to low-level native quantum operations or pulse sequences on hardware. The universal gate set, especially specified by its 2Q basis gates, is the key component of a quantum ISA that dominates its hardware-implementation accuracy and cost, as well as software-expressivity sufficiency.

$ \mathrm{CX} $ or $ \mathrm{CNOT} $ is the most popular basis gate provides by hardware vendors and considered by various quantum compiler optimization methods. The superconducting Cross-Resonance gate~\cite{rigetti2010fully} and ion-trapped Mølmer-Sørensen gate~\cite{bruzewicz2019trapped} are both $ \CX $-equivalent gates with the same canonical form $ \Can\bigl(\frac{1}{2},0,0\bigr) $. In the superconducting platforms with $ XY $-coupled Hamiltonian like Google's Sycamore~\cite{arute2019quantum}, $ \iSWAP\sim\Can\bigl(\frac{1}{2},\frac{1}{2},0\bigr) $ is another representative native 2Q basis gate and could be less sensitive to leakage error than the native $ \mathrm{CZ} $ gate. Recent experimental advances demonstrate that more basis gates could be implemented natively and calibrated in high precision~\cite{chen2025efficient,wei2024native,yale2025realization}. Particularly, some basis gates like $ \SQiSW\sim\Can\bigl( \frac{1}{4},\frac{1}{4},0 \bigr) $ and fractional $ \ZZ(\theta)\sim\Can\bigl(a,0,0\bigr) $ gates offers more promising ISA selections as they exhibit shorter gate duration, higher gate accuracy, and stronger synthesis capability.

The synthesis capability or computational power of basis gates can be geometrically illustrated by monodromy polytopes within the Weyl chamber. The coverage set for \CXISA\ depicted in \Cref{fig:coverage_cx} implies that
\begin{enumerate}
    \item One $ \CX $ gate is required to synthesize 2Q gates $\sim \Can\bigl(\frac{1}{2},0,0\bigr) $, i.e., $ \CX $-equivalent gates $ (V_1\otimes V_2) \CX (V_3\otimes V_4) $;
    \item Two $ \CX $ gates are required to synthesize 2Q gates $\sim \Can(a, b, 0) $, i.e., $ (V_1\otimes V_2) \CX (V_3\otimes V_4) \CX (V_5\otimes V_6) $;
    \item Three $ \CX $ gates are required to synthesize 2Q gates $ \sim \Can(a,b,c) $, i.e., $ (V_1\otimes V_2) \CX (V_3\otimes V_4) \CX (V_5\otimes V_6) \CX (V_7\otimes V_8) $.
\end{enumerate}
We assume the cost of one $ \CX $ gate is $ 1.0 $. Polytopes in different colors denotes the minimal circuit cost (duration) for the coverage set if synthesized by $ \CX $ and arbitrary 1Q gates. That is, on average, the number of $ \CX $ gates required to synthesize arbitrary 2Q gates is $ 3 $. In contrast, the number for \SQiSWISA\ ISA is $ 2.21 $~\cite{huang2023quantum}.

Monodromy polytope theory~\cite{peterson2020fixed} provides a framework for determining the synthesis coverage set and circuit cost (in 2Q depth) for any set of basis gates with specified costs, while the specific gate decomposition process is left to the synthesizer to complete. For the selected ISAs in \Cref{tab:isa_setting} with the basis gate costs assumed in \Cref{eq:cost}, \Cref{fig:coverage_cx,fig:coverage_sqisw,fig:coverage_sqisw_with_mirror,fig:coverage_zzphase,fig:coverage_zzphase_with_mirror,fig:coverage_het} describes their coverage sets, respectively. With the enrichment of quantum ISA (e.g., combining gate families, involving mirror gates) and heterogeneous basis gate cost settings, the coverage set reveals a richer variety of convex polyhedra. That implies more optimization effects for the ISA-ware routing mechanism in \canopus.

\subsection{2Q gate mirroring}\label{appendix:swap_mirroring}

The mirror symmetry of a 2Q gate $ U $ is defined as the composition of the original gate and a $ \SWAP $ gate~\cite{proctor2022measuring}, i.e., $ \mathrm{SWAP} \cdot U $. For example, $ \CX $ and $ \iSWAP $ is a typical pair of mirror gates as shown below.
\begin{center}
    \begin{quantikz}[row sep=0.2cm, column sep=0.2cm, align equals at=1.5]
        & \ctrl{1} & \swap{1} & \ghost{S^\dagger}\qw \\
        & \targ{} & \targX{} & \ghost{S^\dagger}\qw
    \end{quantikz} = \begin{quantikz}[row sep=0.2cm, column sep=0.2cm, align equals at=1.5]
        & \gate{S^\dagger} & \qw & \gate[2]{\mathrm{iSWAP}} & \gate{H} & \qw \\
        & \gate{H} & \gate{S^\dagger} & & \qw & \qw
    \end{quantikz}
\end{center}

\begin{figure}[h!]
    \centering
    \includegraphics[height=0.3\textwidth]{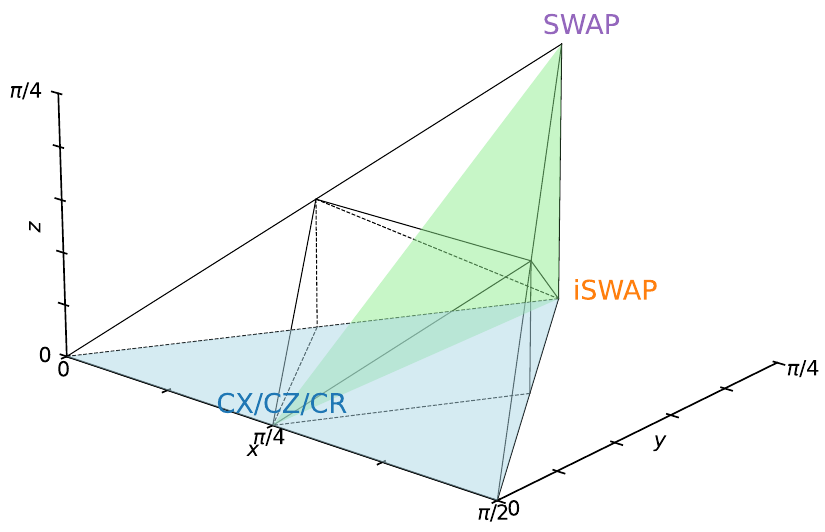}
    \caption{Mirror symmetry for $ \Can(a,b,0) $ and $ \Can(\frac{1}{2},b',c') $ gate families.}
    \label{fig:mirroring}
\end{figure}

In general, the mirroring rule for Canonical coefficients is described as
\begin{align*}
    \mathrm{SWAP} \cdot \mathrm{Can}(a,b,c)
    & \sim \left(a+\frac{1}{2}, b+\frac{1}{2}, c+\frac{1}{2}\right) 
    & \sim \left(a+\frac{1}{2}-1, b+\frac{1}{2}-1, c+\frac{1}{2}-1\right) 
    & \sim
    \begin{cases}
        \left(\frac{1}{2}-c, \frac{1}{2}-b, a - \frac{1}{2}\right), & \text{if } c \geq 0 \\
        \left(\frac{1}{2} + c, \frac{1}{2}-b, \frac{1}{2} - a\right), & \text{if } c < 0
    \end{cases}.
\end{align*}

The mirror pair of $ \CX $ and $ \iSWAP $ is a special case implying that a \CXISA-\iSWAPISA\ combination ISA could result in lower overhead in routing-synthesis collaborative optimization. Yale et al.~\cite{yale2024noise} once considers inserting $ \SWAP $ gates to get mirrored gates with lower synthesis overhead compared to the original gates, given the all-to-all topology and continuous $ \ZZ(\theta) $ gate set on ion-trapped hardware. McKinney et al.~\cite{mckinney2024mirage} discusses that integrating $ \SQiSW $'s mirror gate, i.e., $ \ECP\sim\Can\bigl(\frac{1}{4},\frac{1}{4},0\bigr) $ gate, into the powerful \SQiSWISA\ ISA, could further improve the ISA's synthesis capability and end-to-end routing-synthesis co-optimization on limited topologies.

\subsection{Commutative relation of canonical gates}\label{appendix:commutation_proof}

Herein we present detailed proof for \Cref{thm:commutation}. The \textit{if} direction is trivial, and hence we justify the \textit{only if} direction, relying on the following two lemmas.

\begin{lemma}\label{lemma:hamiltonian_exponential}
Let $A$, $B$ be two Hermitian matrices with eigenvalues in the range $[-2,2)$. If $[e^{-i\frac{\pi}{2}A},e^{-i\frac{\pi}{2}B}]=0$ then $[A,B]=0$.
\end{lemma}
\begin{proof}
This follows from the fact that compatible observables (commuting operators) can be simultaneously diagonalized. In this case, the respective unitary matrix $e^{-i\frac{\pi}{2}A}$ commutes with $e^{-i\frac{\pi}{2}B}$. Denote by $A_{\lambda}$ the eigenspace corresponding to the eigenvalue $\lambda$ of $e^{-i\frac{\pi}{2}A}$, i.e. $e^{-i\frac{\pi}{2}A} = \oplus_{\lambda} \lambda A_{\lambda}$. Then we have
\begin{align}
    \forall \vec{v} \in A_\lambda,\, e^{-i\frac{\pi}{2}B}e^{-i\frac{\pi}{2}A}\vec{v}=e^{-i\frac{\pi}{2}B}\lambda \vec{v}=\lambda e^{-i\frac{\pi}{2}B}\vec{v}=e^{-i\frac{\pi}{2}A}e^{-i\frac{\pi}{2}B}\vec{v},
\end{align}

and thus $e^{-i\frac{\pi}{2}B}\vec{v}\in A_\lambda$. Thus $A_\lambda$ is $e^{-i\frac{\pi}{2}B}$-invariant and the restriction $e^{-i\frac{\pi}{2}B}\big\rvert_{A_{\lambda}}$ of $e^{-i\frac{\pi}{2}B}$ to $A_{\lambda}$ is still unitary since it preserves inner products. Hence it is diagonalizable and we can find an orthonormal basis $w_{\lambda_1},w_{\lambda_2},\ldots,w_{\lambda_k}$ consisting of eigenvectors of $e^{-i\frac{\pi}{2}B}\big\rvert_{A_{\lambda}}$. Note that these are also eigenvectors of $e^{-i\frac{\pi}{2}A}$ (with eigenvalue $\lambda$). Following the same token as above, for each eigenspace $E_{\lambda_i}$ of $e^{-i\frac{\pi}{2}A}$, we can construct an orthonormal basis $\beta_i$ for it consisting of eigenvectors of $e^{-i\frac{\pi}{2}B}$. Finally since the eigenspaces of different eigenvalues of $e^{-i\frac{\pi}{2}A}$ are orthogonal to each other, $\beta=\cup_i\beta_i$ forms an orthonormal basis of the entire Hilbert space $\mathcal{H}_n$ consisting of the coeigenvectors of both $e^{-i\frac{\pi}{2}A}$ and $e^{-i\frac{\pi}{2}B}$.

Now let $U$ be a unitary matrix with the vectors in $\beta$ being its columns, then
\begin{align}
    \begin{aligned}
    U^\dagger e^{-i\frac{\pi}{2}A}U&=D_A\\
    U^\dagger e^{-i\frac{\pi}{2}B}U&=D_B
    \end{aligned}
\end{align}

In general, an eigenvector of $e^{-i\frac{\pi}{2}A}$ need \textit{not} be that of $A$. However, since $A$ has its eigenvalues in the range $[-2,2)$, the map
\begin{align}
    f:[-2,2)\rightarrow U(1),a\rightarrow e^{-i\frac{\pi}{2}a}
\end{align}
is injective. Consequently different eigenvalues of $A$ correspond to different eigenvalues of $e^{-i\frac{\pi}{2}A}$, and hence the eigenspaces of $e^{-i\frac{\pi}{2}A}$ and $A$ coincide. Therefore, we have that
\begin{align}
    \begin{aligned}
    U^\dagger AU&=\Sigma_A\\
    U^\dagger BU&=\Sigma_B
    \end{aligned}
\end{align}
and since $[\Sigma_A,\Sigma_B]=0$ as they are diagonal, $[A,B]=0$. We obtain the desired result.

\end{proof}

\begin{lemma}\label{lemma:xx_rotation}
Let $P_1=(a_1X_1X_2+b_1Y_1Y_2+c_1Z_1Z_2)I_3$, $P_2=I_1(a_2X_2X_3+b_2Y_2Y_3+c_2Z_2Z_3)$ with $|c_1|\le b_1\le a_1\le\frac{1}{2}$, $|c_2|\le b_2\le a_2\le\frac{1}{2}$. If $[P_1,P_2]=0$ and $P_1,P_2\ne0$, then $b_1=b_2=c_1=c_2=0$.
\end{lemma}
\begin{proof}
Consider the product $P_1P_2$. We assume for the sake of contradiction that $b_1\ne0$. Using $[X,Y]=2iZ$, $[Y,Z]=2iX$, $[Z,X]=2iY$, we expand
\begin{align}
[P_1,P_2] &= 2i\bigl(a_1b_2\,X_1Z_2Y_3 - b_1a_2\,Y_1Z_2X_3 + b_1c_2\,Y_1X_2Z_3\bigr) -2i\bigl(a_1c_2\,X_1Y_2Z_3 + c_1a_2\,Z_1Y_2X_3 + c_1b_2\,Z_1X_2Y_3\bigr).
\end{align}
Since the each Pauli string is linearly independent in the $8\times8$ operator basis, e.g. term $Y_1Z_2X_3$ cannot be canceled out by any other terms, contradictory to the fact that $[P_1,P_2]=0$. Hence, vanishing of $[P_1,P_2]$ requires
\begin{align}
a_1b_2 = a_1c_2 = b_1c_2 = b_1a_2 = c_1a_2 = c_1b_2 = 0.
\end{align}
Since $P_1,P_2\neq0$, at least $a_1,a_2$ is nonzero, leading to $b_1 = b_2=c_1=c_2=0$. 
\end{proof}

Using \Cref{lemma:hamiltonian_exponential} and \Cref{lemma:xx_rotation} above, it is straightforward to prove \Cref{thm:commutation}. We see that $\|P_1\|\le\|a_1X_1X_2I_3\|+\|b_1Y_1Y_2I_3\|+\|c_1Z_1Z_2I_3\|\le|a_1|+|b_1|+|c_1|\le\frac{3}{2}$, where $\|\cdot\|$ is the operator norm. Hence, eigenvalues of $P_1$ are in range of $[-2,2)$. Same as the eigenvalues of $P_2$. Now if $[e^{-i\frac{\pi}{2}P_1},e^{-i\frac{\pi}{2}P_2}]=0$, then we have that $[P_1,P_2]=0$ according to \Cref{lemma:hamiltonian_exponential}, and thus $b_1=b_2=c_1=c_2=0$ according to \Cref{lemma:xx_rotation}, which proves the \textit{only if} direction.

\end{document}